\documentclass[11pt, reqno]{amsart}

\usepackage[utf8]{inputenc}
\usepackage[T1]{fontenc}
\usepackage[english]{babel}
\usepackage{amsmath, amssymb, amsthm} 
\usepackage{mathtools} 
\usepackage{mathrsfs}   
\usepackage{bbm}        
\usepackage{dsfont}     
\usepackage{braket}     
\usepackage[colorlinks=true, linkcolor=blue, citecolor=red]{hyperref}
\usepackage{enumitem}
\usepackage{physics}
\usepackage{stmaryrd}
\usepackage{empheq} %encadrer plusieurs lignes
\usepackage{tikz}        % Pour l'environnement {tikzpicture} et les dessins
\usepackage{pgfplots}
\pgfplotsset{compat=1.18}
\usetikzlibrary{arrows.meta,calc,patterns,angles,quotes,decorations.pathreplacing}
\usepgfplotslibrary{fillbetween}
\usepackage{subcaption}  % Pour l'environnement {subfigure}
\usepackage{makecell}

\usepackage{mdframed}
\surroundwithmdframed[
  linewidth=0.5pt,
  skipabove=12pt,         % Espace entre le texte précédent et le cadre
  skipbelow=12pt,         % Espace entre le cadre et le texte suivant
  innertopmargin=0pt,     % Espace entre le haut du cadre et le titre
  innerbottommargin=5pt,
  innerleftmargin=5pt,
  innerrightmargin=5pt
]{theorem}

\usepackage{geometry}
\usepackage{etoolbox}

\patchcmd{\section}{\normalfont}{\normalfont\LARGE}{}{}

\makeatletter
\renewcommand{\subsection}{\@startsection{subsection}{2}%
  \z@{.7\linespacing\@plus\linespacing}% Espace AVANT le titre
  {.5\linespacing}% Espace APRÈS le titre (Positif = saut de ligne)
  {\normalfont\large\bfseries}}% Style (Gras et Large)
\makeatother

\usepackage{graphicx} 
\usepackage{float}      
\usepackage{subcaption} 
\graphicspath{ {./Figures/} }

\newcommand{\pder}[2]{\frac{\partial #1}{\partial #2}}

\newcommand{\xu}{x_1} %axe x_1
\newcommand{\xd}{x_2} %axe x_2 (comme ça, x = (x_1,x_2))
\newcommand{\R}{\mathbb{R}} %réels
\newcommand{\C}{\mathbb{C}} %complexes
\newcommand{\N}{\mathbb{N}}%entiers
\newcommand{\D}{\mathcal{D}}%domaine d'opérateur
\newcommand{\Q}{\mathcal{Q}}%domaine de forme
\newcommand{\q}{\mathfrak{q}} %forme quadratique
\newcommand{\Rx}{\mathbb{R}_x}%espace des configurations petit système
\newcommand{\Ry}{\mathbb{R}_y}%espace des configurations champ
\newcommand{\Rk}{\mathbb{R}_k}%espace des moments
\newcommand{\supp}{\mathrm{supp}\mspace{2mu}} %support fonction
\DeclareMathOperator{\Ran}{Ran} %range of an operator
\DeclareMathOperator{\argmin}{argmin}
\newcommand{\ext}{\mathrm{ext}} %text ext
\newcommand{\tot}{\mathrm{tot}} %text tot
\newcommand{\slim}[1]{\underset{#1}{\mathrm{s-lim }}} %limite forte
\newcommand{\dist}{\mathrm{dist}}
\newcommand{\Hi}{\mathscr{H}} %espace de Hilbert
\newcommand{\Ha}{\mathcal{H}} %Hamiltonien
\newcommand{\ha}{\mathfrak{h}} %Hamiltonien petit système
\newcommand{\dP}{d \mspace{1mu}\mathbb{P}} %mesure de probabilité infinitésimale
\newcommand{\F}{\mathcal{F}} %Transfo de fourier
\newcommand{\ac}{a^{\ast}} %opérateur création
\newcommand{\Fs}{\mathscr{F}_s}%espace de fock(input)
\newcommand{\Fsf}{\mathscr{F}^{\mathrm{fin}}_s} %espace de fock bosonique nombre fini de particules
\newcommand{\ldd}{L^2(\R^2)}%L22 passe pas je mets les initiales
\newcommand{\ciz}{\mathscr{C}^{\infty}_0} %fonctions C_inf à support compact
\newcommand{\ci}{\mathscr{C}^{\infty}}

\newcommand{\free}{\mathrm{free}} %pour pas réécrire le \mathrm à chaque fois
\newcommand{\eff}{\mathrm{eff}} %pour pas réécrire le \mathrm à chaque fois
\newcommand{\fin}{\mathrm{fin}} %pour pas réécrire le \mathrm à chaque fois
\newcommand{\el}{\mathrm{el}}

\newcommand{\fib}{\mathrm{fib}}
\newcommand{\ind}{\mathds{1}} % Indicatrice, utilise le package dsfont  
\newcommand{\alten}{\,\widehat{\otimes}\,} %produit tensoriel algébrique
\newcommand{\sess}{\sigma_{\mathrm{ess}}}%spctre essentiel
\newcommand{\sdisc}{\sigma_{\mathrm{disc}}}%spectre discret
\newcommand{\dG}{\mathrm{d}\Gamma} %opérateur de seconde quantification
\newcommand{\Nb}{\mathcal{N}} %opérateur nombre de particules
\newcommand{\Gc}{\check{\Gamma}} %foncteur de seconde quantification check
\newcommand{\Cut}{\mathrm{Cut}} %cutlocus
\newcommand{\straight}{\mathrm{straight}} %straightening operator
\newcommand{\fsum}{\dot{+}} %somme au sens des formes

\newcommand{\cc}[1]{\overline{#1}} %conjugué complexe
\newcommand{\cRM}[1]{\MakeUppercase{\romannumeral #1}} %nombre romain
\newcommand{\iph}[2]{\left\langle #1 \middle\vert #2 \right\rangle_{\Hi}}

\newcommand{\intd}[2]{\int_{\R^{#1}}^{\oplus} #2 \, dx} %intégrale directe

\newcommand{\nld}[1]{\left\lVert #1\right\rVert_{L^2}} %norme L^2
\newcommand{\nop}[1]{\left\lVert #1 \right\rVert_{\mathrm{op}}} %norme d'opérateur
\newcommand{\snfs}[1]{\left\lVert #1\right\rVert^2_{\mathscr{F}_s}} %squared norm espace de fock bosonique
\newcommand{\nldfs}[1]{\left\lVert #1 \right\rVert_{L^2 \otimes \mathscr{F}_s}} %norm produit tensoriel L^2 fock bosonique
\newcommand{\snldfs}[1]{\left\lVert #1 \right\rVert^2_{L^2 \otimes \mathscr{F}_s}}%squared norm produit tensoriel L^2 fock bosonique
\newcommand{\ipld}[2]{\left\langle #1 \middle\vert #2 \right\rangle_{L^2}} %produit scalaire L^2
\newcommand{\ipldfs}[2]{\left\langle #1 \middle\vert #2 \right\rangle_{L^2 \otimes \mathscr{F}_s}} %produit scalaire produit tensoriel L^2 fock bosonique
\newcommand{\ens}[2]{\left\{ #1 \,\middle\vert\, #2 \right\}}

\providecommand{\littletaller}{\vphantom{\sum}} 

\newcommand*\fonction[5]{
#1 \, \colon \left\{\begin{alignedat}{2}  #2   &\to      #3\\
                                #4    &\mapsto  #5
\end{alignedat} \right. \kern-\nulldelimiterspace}
\newcommand\restr[2]{{% we make the whole thing an ordinary symbol
  \left.\kern-\nulldelimiterspace % automatically resize the bar with \right
  #1 % the function
  \littletaller % pretend it's a little taller at normal size
  \right|_{#2} % this is the delimiter
  }}

\newtheorem{theorem}{Theorem}[section]
\newtheorem{lemma}[theorem]{Lemma}
\newtheorem{proposition}[theorem]{Proposition}
\newtheorem{corollary}[theorem]{Corollary}

\theoremstyle{definition}
\newtheorem{definition}[theorem]{Definition}
\newtheorem{remark}[theorem]{Remark}

\newtheorem*{notation}{Notation}

\newtheorem{innercustomthm}{Theorem}

\surroundwithmdframed[
  linewidth=0.5pt,
  skipabove=12pt,
  skipbelow=12pt,
  innertopmargin=0pt,
  innerbottommargin=5pt,
  innerleftmargin=5pt,
  innerrightmargin=5pt
]{innercustomthm}

\title[Spectral Analysis of a Waveguide Coupled to a Quantum Field]{Spectral Analysis of a Soft  Quantum Waveguide Coupled to a Massive Quantum Scalar Field}

\author[B. Alvarez]{Benjamin ALVAREZ}
\address[Benjamin Alvarez]{%
  Univ Toulon, Aix Marseille Univ, CNRS, CPT, Toulon, France
  \newline
  \normalfont\ttfamily benjamin.alvarez@univ-tln.fr}

\author[H. Gouttenegre]{Hugo GOUTTENEGRE}
\address[Hugo Gouttenegre]{%
  Univ Toulon, Aix Marseille Univ, CNRS, CPT, Toulon, France
  \newline
  \normalfont\ttfamily hugo-gouttenegre@etud.univ-tln.fr}

\date{\today}

\begin{document}

\begin{abstract}
We study a Nelson-type Hamiltonian describing a massive spinless non-relativistic particle moving in a planar soft quantum waveguide and linearly coupled to a massive scalar bosonic field. Our spectral analysis relies on comparing the curved waveguide with its straight counterpart. For the straight Hamiltonian, longitudinal translation invariance yields a fiber decomposition, and we prove that its spectrum is purely essential and gapless. For the curved Hamiltonian, we establish an HVZ-type formula for the bottom of the essential spectrum, involving that of the straight system and the one-boson threshold. We also derive a sufficient condition for the existence of non-empty discrete spectrum in the curved case, expressed in terms of an effective one-dimensional Schrödinger operator.
\end{abstract}

\maketitle
\setcounter{tocdepth}{2} % mettre 2 pour avoir les sous-sections
\tableofcontents

%================================================================
% --- SECTIONS ---
%================================================================

% 1) Introduction
\section{Introduction}
\label{chap:intro}

A standard framework in mathematical physics is what physicists call non-relativistic quantum electrodynamics (NRQED), which describes a particle interacting with a quantized electromagnetic field at low energies, while neglecting relativistic phenomena and particle-antiparticle pair creation. To study, for example, the fine structure of hydrogen or the behavior of light in matter, it is much more efficient to simplify the theory by keeping only the truly relevant degrees of freedom: slow electrons and transverse photons. Consequently, NRQED is an effective field theory that captures the relevant low-energy physics while avoiding unnecessary complications, and the work presented here falls within this framework.

\bigbreak

In this context, we focus on a physical setup closely related to models
considered in the physics literature (see, e.g., \cite{physics_waveguide_QED}), but which, in the form studied here, has yet to be investigated from a rigorous mathematical perspective. More precisely, we consider a \textit{spinless, charged, massive non-relativistic particle propagating in a waveguide and coupled to a quantum field.} In the context of quantum mechanics, waveguides have been widely studied (see, e.g.,~\cite{pv_book} and references therein), but this is not the case for quantum waveguides coupled to a quantum field: to the best of our knowledge, this is the first mathematical study of its kind. This leads us to \textit{propose a simplified model, for which we establish preliminary results}. These will serve as a methodological foundation, to be refined and extended as the model itself becomes more elaborate in future work.

\bigbreak

To construct this model, we must make specific choices regarding both the field interaction and the nature of the waveguide. Regarding the interaction, after the simplifications justified in Appendix~\ref{app:model} (and references therein), we obtain a \textit{Nelson-type Hamiltonian}. In particular, these simplifications lead us to consider a \textit{massive scalar bosonic quantum field and a linear coupling between the particle and the field}.

\bigbreak

Regarding the nature of the waveguide, one can distinguish between two types of waveguides: hard waveguides and soft waveguides. Hard waveguides are those where the particle is physically confined in the waveguide, meaning that the free Hamiltonian is described for instance by a Dirichlet Laplacian; they are comprehensively treated in \cite{pv_book}. Soft waveguides are described through a potential well playing the role of the waveguide: there is no physical boundary for the particle, but its propagation is constrained by the potential well; see, e.g., the review~\cite{review_soft_waveguides} by Exner. Physically, soft waveguides often provide a more realistic description, but they can be approximated by hard waveguides when the potential well is sufficiently deep. Mathematically, we are interested in them because they enable us to define the problem over the whole configuration space rather than just within the waveguide: it is then much simpler to define the quantum field. \\
Consequently, as a first step in the study of this physics problem, \textit{we will consider soft waveguides}. However, they have only been investigated recently, the first extensive study on the subject being the paper \cite{pv_spsqv} by Exner. This implies that these objects are not yet fully understood, which constitutes both a primary challenge and a motivation for the model under consideration.

\bigbreak

Beyond these intrinsic features, it turns out that our model does not fit perfectly into the established mathematical frameworks of the literature, which is precisely what makes it of interest. For instance, these standard frameworks typically fall into two main categories. The first considers a confined particle, leading to a small system with compact resolvent (see, e.g.,~\cite{acqftmpfh, lemma_power_dG_gerar_moller}). The second is the atomic setting, where the potential --- such as one generated by a nucleus --- vanishes at infinity and is thus relatively compact with respect to the Laplacian (see, e.g.,~\cite{analisa, gsnrqed}). Our model falls into neither of these standard categories, since we consider a waveguide extending infinitely in the longitudinal direction and generated by a bounded potential well, so that the resulting potential is non-confining and does not vanish at infinity. Although this may seem like a minor detail, it actually complicates the standard proof strategies. \\
A final standard framework is one where the Hamiltonian is translation invariant (see, e.g.,~\cite{moller_translation_invariant_nelson_model, thomas_thesis}). Once again, our model does not fit into this setting, although it remains partially applicable to a specific case that will be highly useful to us, namely the straight waveguide. Indeed, in such a case, the system exhibits translation invariance along the longitudinal direction of the guide.

\bigbreak

Because standard frameworks are not perfectly suited to our model, we rely on an additional source of inspiration: the methods developed for the spectral analysis of quantum waveguides. In the usual theory of quantum waveguides, a central idea is to treat a curved guide as a geometric perturbation of its associated straight guide, since the curvature may affect the propagation of the particle. This geometric intuition remains relevant for our model, but the coupling to a quantized field enriches the mathematical structure and requires adapted arguments.

\bigbreak

Following this general idea, Section~\ref{sec:ess_spec} exploits this comparison strategy to obtain results regarding the essential spectrum. To address the more delicate question of bound states for the curved system, we derive in Section~\ref{sec:binding} a \textit{binding condition}, which should be understood as a first step toward proving the existence of geometry-induced bound states, a problem left for future work. Sections~\ref{sec:model} and~\ref{sec:prelim} are devoted to presenting the model and preliminary results.

\bigbreak

We conclude this introduction by summarizing the main results of the paper. First, exploiting the longitudinal translation invariance of the straight waveguide, we decompose the straight Hamiltonian into momentum fibers and prove that its spectrum is purely essential and gapless; see Theorem~\ref{sec:ess_spec:theorem_spectrum_straight_waveguide}. For the curved waveguide, we prove an HVZ-type theorem by combining IMS localization with Persson's formula in order to compare the curved system with the straight one at infinity. We obtain that the bottom of the essential spectrum of the curved Hamiltonian is determined by the bottom of the essential spectrum of the straight system and by the one-boson threshold; see Theorem~\ref{sec:ess_spec:HVZ_theorem_curved_waveguide}. \\
Finally, we derive a sufficient condition for the existence of discrete spectrum below the essential spectrum of the curved Hamiltonian. We first prove that the infimum of the spectrum of the straight Hamiltonian is attained by a momentum fiber and that the corresponding fiber Hamiltonian admits a ground state; see Theorem~\ref{sec:binding:theorem_grond_state_fibered_hamiltonian}. We further show that this infimum is attained at zero longitudinal momentum; see Theorem~\ref{sec:binding:eta_0_equal_0}. Starting from the corresponding fiber ground state, we construct a trial state $\Upsilon$ for the curved Hamiltonian $\Ha_C$. Evaluating $\Ha_C$ on this trial state leads to an inequality involving an effective one-dimensional Schrödinger operator: if this effective operator has a negative-energy state, then $\Ha_C$ has non-empty discrete spectrum; see Theorem~\ref{sec:binding:theorem_binding_condition}. We refer to this negative-spectrum criterion as the binding condition. 
% 2) Physical and Mathematical Model
\section{Presentation of the model}
\label{sec:model}

\subsection{The quantum waveguide}\label{subsec:pres_waveguide}

We consider \textit{a spinless, non-relativistic, charged and massive particle} propagating in an \textit{infinite planar soft quantum waveguide without self-intersections}. Following the construction of \cite{pv_spsqv}, let $\gamma:\R\to\R^2$ be a unit-speed parametrization of a planar curve without self-intersections, satisfying the following hypotheses:
\begin{enumerate}[label=(H\arabic*$_{\mathrm{geom}}$), leftmargin=*, itemsep=10pt]
    \item\label{h1_geom_regularity_curve} \underline{Regularity:} in order to have a well-defined $\mathscr{C}^2(\R)$ signed curvature $\kappa$, we require the curve $\gamma$ to be $\mathscr{C}^4(\R)$. In particular $\kappa : \R \to \R$ is well defined and given by $\kappa(s)  = \dot{\gamma}_2(s)\ddot{\gamma}_1(s) - \dot{\gamma}_1(s) \ddot{\gamma}_2(s)$, where $\gamma = (\gamma_1,\gamma_2)$.  
    \item\label{h2_geom_compact_curvature} \underline{Compact localization of the curvature:} requiring $\kappa \in \mathscr{C}^2_0(\R)$ localizes the curved part of the waveguide in a compact region of the plane.
\end{enumerate}
Denoting by $\varphi : \R^2 \to \R^2$ the map expressing a point $x$ of the plane in the curvilinear coordinates $(s,u)$ (see Appendix \ref{app:curve:curv_coord}), i.e. $x = \varphi(s,u)$, we build an infinite planar strip of width $2a$ by setting:  
\begin{align*}
    \Omega^a := \ens{x\in\R^2}{\dist(x,\gamma)<a}.
\end{align*}
Unique curvilinear coordinates $(s,u)$ can be defined for almost every point in $\R^2$, except on the cut locus of the curve $\gamma$, denoted $\Cut(\gamma)$, which is a Lebesgue-null set (see Appendix \ref{app:curve:cutlocus} for the precise definition and relevant properties of the cut locus). By definition, there exists an open set $\mathcal{O}_\gamma$ such that 
\begin{align*}
    \fonction{\varphi}{\mathcal{O}_\gamma}{\R^2\setminus \Cut(\gamma)}{(s,u)}{\varphi(s,u) =: x} \hspace{1cm} \text{is a diffeomorphism}.
\end{align*}
Using those curvilinear coordinates, we make sure that the strip $\Omega^a$ \textit{does not intersect itself} with the following assumption: 
\begin{enumerate}[label=(H\arabic*$_{\mathrm{geom}}$), leftmargin=*]
    \setcounter{enumi}{2}
	\item\label{h3_geom_no_self_intersection} \underline{No self-intersection:} no point of the strip $\Omega^a$ is in the cut locus of the curve $\gamma$, i.e. 
    \begin{align*}
        \Omega^a \subset \R^2\setminus \Cut(\gamma) \iff \R \hspace{2pt}\times \hspace{2pt} ]-a,a[ \hspace{4pt} \subset \mathcal{O}_\gamma, 
    \end{align*}
    or equivalently $\varphi |_{\R \times ]-a,a[}$ is a diffeomorphism onto $\Omega^a$.
\end{enumerate}

\begin{remark}\label{sec:model:remark_jacobian}
    Let $J_\varphi$ be the Jacobian of the map $\varphi$: one has $J_\varphi(s,u) = 1  +  u \hspace{1pt} \kappa(s)$ (see, e.g., \cite[Section 1.1]{pv_book}). If $\varphi$ is a global diffeomorphism on $\mathcal{O}_\gamma$, it is in particular a local one; hence its Jacobian cannot vanish and $1 \hspace{2pt} + \hspace{2pt} u \hspace{1pt} \kappa(s) \neq 0$ for all $(s,u) \in \mathcal{O}_\gamma$.
\end{remark}

\noindent Since the quantum waveguide is soft, it is defined through a potential, the latter being defined as a potential well on the transverse variable replicated along the longitudinal one: 
\begin{enumerate}[label=(H\arabic*$_{\mathrm{pot}}$), leftmargin=*]
    \setcounter{enumi}{0}
    \item\label{hyp:potential} \underline{Potential:}  the transverse potential well is a \textit{non-positive} function $V \in L^\infty(\R)$, $V\not\equiv 0$, supported in $[-a,a]$. If the waveguide is straight, its potential $V_S$ is obtained by replicating $V$ along the longitudinal direction $\xu$:
    \begin{align*}
        V_S(x_1,x_2) := V(x_2).
    \end{align*}
    If the waveguide is curved, its potential $V_C$ is obtained by replicating $V$ along the arc-length coordinate $s$:
    \begin{align*}
        \forall x = \varphi(s,u) \in \R^2\setminus \Cut(\gamma),  \hspace{5pt} V_C(x) = V_C\hspace{-2pt}\left(\varphi(s,u)\right) := V(u). 
    \end{align*}
\end{enumerate}

\begin{notation}
    In the following, $V_\bullet$ stands for either $V_S$ or $V_C$, and all statements involving $V_\bullet$ hold for both choices.
\end{notation}

\noindent One can see that this definition of the potential is \textit{non-confining}, meaning that there can be quantum tunneling between different points of the waveguide. To simplify the computation of the essential spectrum, we want to avoid quantum tunneling at infinity (since a residual coupling would mix the transverse bound states). We therefore require the distance between the two branches of the waveguide to grow indefinitely, so that they asymptotically decouple and each becomes equivalent to a straight waveguide. To do so, we assume: 
\begin{enumerate}[label=(H\arabic*$_{\mathrm{geom}}$), leftmargin=*]
    \setcounter{enumi}{3}
    \item\label{h4_geom_not_u_shaped} \underline{Not U-shaped}: assuming $\abs{\gamma(s)-\gamma(s')} \to +\infty$ as $\abs{s-s'} \to +\infty$ ensures that the curve is not U-shaped. Together with Assumption~\ref{h2_geom_compact_curvature}, it means that the cut radii maps $c_\pm$ go to infinity with the arc length (see Proposition \ref{prop:appendixA_cut_radius_infty}):
    \begin{align*}
         c_\pm(s) \underset{\abs{s} \to + \infty}{\longrightarrow} + \infty.
    \end{align*}
\end{enumerate}

\noindent Consequently, our small system is described by the \fbox{\textit{Schrödinger operator} $-\Delta + V_\bullet(X)$}, where, for $V_\bullet(X)$, we use the following notation:

\begin{notation}
    For a function $f$ on the configuration space, we denote by $f(X)$ the associated multiplication operator on $L^2(\R^n, dx)$. If it is defined on $\mathcal{O}_\gamma$, we denote by $f(S,U)$ the associated multiplication operator on $L^2(\mathcal{O}_\gamma, ds \mspace{2mu} du)$.
\end{notation}
\begin{notation}
    The measures $dx$ and $ds \mspace{2mu} du$ are dropped from now on, unless clarity  requires otherwise.
\end{notation}

Thanks to all of this, we can define the \textit{unitary straightening operator}: 
\begin{align*}
    \fonction{U_{\text{straight}}}{L^2\big(\R^2 \big)}{L^2(\mathcal{O}_\gamma)}{f}{\abs{\det J_\varphi}^{1/2}\,  f\circ\varphi } 
     \text{ with inverse }  \hspace{0.2cm}       \fonction{U^{\ast}_{\text{straight}}}{L^2(\mathcal{O}_\gamma)}{L^2\big(\R^2\big)}{g}{\abs{\det J_{\varphi^{-1}}}^{1/2}\,  g\circ\varphi^{-1} }.
\end{align*}
The latter enables us to transform our Schrödinger operator by mapping the Laplacian to the modified Laplace--Beltrami operator $-\Delta_\gamma$. On $\ciz(\mathcal{O}_\gamma)$, it acts as follows (see~\cite[(1.7), (1.8)]{pv_book}):
\begin{align*}
    U_\straight (-\Delta) U_\straight^\ast
    =: -\Delta_\gamma
    = - \partial_s (1 + u \kappa)^{-2} \partial_s
      - \partial_u^2 + V_\kappa(S,U).
\end{align*}
Here, the curvature-induced potential $V_\kappa$ is given by
\begin{align*}
    V_\kappa(s,u)
    := -\frac{\kappa^2(s)}{4\left(1 + u\,\kappa(s)\right)^2}
       - \frac{5}{4}
         \frac{u^2 \dot{\kappa}(s)^2}
              {\left(1 + u\,\kappa(s)\right)^4}
       + \frac{u\,\ddot{\kappa}(s)}
              {2\left(1 + u\,\kappa(s)\right)^3},
\end{align*}
which is well defined pointwise thanks to Remark~\ref{sec:model:remark_jacobian} and Hypothesis~\ref{h3_geom_no_self_intersection}. As for the potential $V_C$, it is mapped to the transverse well $V$ by $U_{\text{straight}}$, so that:
\begin{align*}
U_\straight (-\Delta + V_C(X)) U_\straight^\ast = -\Delta_\gamma + V(U).
\end{align*}

Having switched to $(s,u)$ coordinates, we can now see the influence of the curvature and understand in what sense the curved waveguide is a perturbation of the straight one:
\begin{itemize}
\item \underline{Straight waveguide}: it is characterized by $\kappa = 0$ everywhere so that $-\Delta_\gamma = -\partial_s^2-\partial_u^2$, the cut locus $\Cut(\gamma)$ is empty and $\mathcal{O}_\gamma = \R^2$ since for all $s$, $c_\pm(s) = +\infty$. The Hamiltonian of the system is: 
\begin{align*}
\ha_{\mspace{1mu}\el,S} = -\partial_s^2 - \partial_u^2 + V(U), 
\end{align*}
on the Hilbert space $L^2(\R^2)$.
\item \underline{Curved waveguide}: it is characterized by $\kappa$ not identically zero, meaning that the cut locus $\Cut(\gamma)$ is not empty,  $\mathcal{O}_\gamma \neq \R^2$ (since $c_\pm$ are not identically $+\infty$) and the curvature's effects are encoded in $-\Delta_\gamma$. The Hamiltonian of the system is:
\begin{equation}\label{eq:curved_electronic_hamiltonian}
    \ha_{\mspace{1mu}\el,C} = 
    \left\{
    \begin{array}{@{}l@{}} 
        -\partial_s (1 + u \kappa)^{-2} \partial_s - \partial_u^2 + V_\kappa(S,U) + V(U) = -\Delta_\gamma+V(U) \quad \text{when } s \in \supp \, \kappa \\[10pt]
        -\partial_s^2 - \partial_u^2 + V(U) \hfill \text{when } s \in \R\setminus \supp \, \kappa 
    \end{array}
    \right.,
\end{equation}
on the Hilbert space $L^2(\mathcal{O}_\gamma) \neq L^2(\R^2)$.
\end{itemize}
Since the action of the Hamiltonian for the curved waveguide exactly matches that of the straight waveguide outside the support of the curvature (modulo the slight difference in their underlying Hilbert spaces), it can naturally be viewed as a localized perturbation of the latter. 

\medbreak

Of course, everything can also be expressed in Cartesian coordinates, and we might switch between them as necessary. \textit{By abuse of notation, we will omit the straightening operator}: 

\begin{table}[H]
    \centering
    \renewcommand{\arraystretch}{1.5}
    \begin{tabular}{|c|c|c|}
        \hline
         & \textbf{Curvilinear coordinates} & \textbf{Cartesian coordinates} \\
         & Hilbert space: $L^2(\mathcal{O}_\gamma, ds\,du)$
         & Hilbert space: $L^2(\R^2,dx)$ \\
        \hline
        \makecell{\textbf{Straight} \\ \textbf{waveguide}}
        &
        \makecell{
            \rule{0pt}{15pt}
            $\ha_{\mspace{1mu}\el,S}
            = -\partial_s^2 \otimes I
            + I \otimes \bigl(-\partial_u^2 + V(U)\bigr)$
            \\[10pt]
            \small
            $\mathcal{O}_\gamma=\R^2$ and the variables separate
            \rule[-5pt]{0pt}{0pt}
        }
        &
        \makecell{
            \rule{0pt}{25pt}
            $\displaystyle
            \begin{aligned}
                \ha_{\mspace{1mu}\el,S}
                &= -\Delta + V_S(X) \\
                &= -\partial_{\xu}^2 \otimes I
                + I \otimes \bigl(-\partial_{\xd}^2 + V(X_2)\bigr)
            \end{aligned}$
            \\[15pt]
            \small
            The variables separate
            \rule[-5pt]{0pt}{0pt}
        }
        \\
        \hline
        \makecell{\textbf{Curved} \\ \textbf{waveguide}}
        &
        \makecell{
            \rule{0pt}{15pt}
            $\ha_{\mspace{1mu}\el,C} = -\Delta_\gamma + V(U)$
            \\[5pt]
            \small
            $\mathcal{O}_\gamma\neq\R^2$ and the variables do not separate
            \rule[-6pt]{0pt}{0pt}
        }
        &
        \makecell{
            \rule{0pt}{15pt}
            $\ha_{\mspace{1mu}\el,C} = -\Delta + V_C(X)$
            \\[5pt]
            \small
            The variables do not separate
            \rule[-6pt]{0pt}{0pt}
        }
        \\
        \hline
    \end{tabular}
    \caption{Comparison of the Hamiltonian representations in curvilinear and Cartesian coordinates.}
    \label{tab:hamiltonian_comparison_coordinates}
\end{table}
\noindent One notices that for the straight waveguide, the curvilinear coordinates $(s,u)$ and the Cartesian ones $(\xu,\xd)$ are exactly the same (i.e. $U_\straight$ is the identity), so that $-\Delta_\gamma = - \Delta$ in that case.

%%% ------------------------------------------------------------
%===============================================================
%%% ------------------------------------------------------------

\subsection{Coupling with the quantum field} 

We now couple the particle in the quantum waveguide to a quantized field. Since the particle is non-relativistic, the physically natural framework would be NRQED, or equivalently the Pauli--Fierz model. In the present work, however, we use the simplified model discussed in Appendix~\ref{app:model}: the electromagnetic field is replaced by a scalar massive bosonic field, and the coupling is taken to be linear. We are therefore led to a Nelson model. \\
For simplicity, we take the bosonic field to be two-dimensional, as is the
configuration space of the particle. Allowing a three-dimensional field would
mainly affect notation and threshold formulas, but not the basic
operator-theoretic arguments.

\smallbreak

As usual in quantum field theory, we describe each boson in the momentum space representation, so the one boson Hilbert space is $L^2(\R^2,dk)$, and the Hilbert space for the bosonic quantum field $\Hi_b$ is the symmetric Fock space over $L^2(\R^2,dk)$ (see Appendix~\ref{app:fock} for the Fock space notations and properties used throughout the paper): 
\begin{align*}
    \Hi_b := \Fs\hspace{-2pt}\left( L^2\big(\R^2, dk\big) \right) =  \bigoplus_{n=0}^{+\infty} \sideset{}{_s}\bigotimes_{i=1}^n L^2(\R^2, dk). 
\end{align*}
If needed, we might describe the quantum field in the configuration space, where the spatial variable is $y = i \nabla_k$, i.e. we might use $\Hi_b \cong \Fs(L^2(\R^2,dy))$. 

\begin{notation}
    From now on, we denote $L^2(\Rk^2) := L^2(\R^2,dk)$ and $L^2(\Ry^2) := L^2(\R^2,dy)$ the momentum and configuration space for the quantum field respectively; we denote $L^2(\Rx^2) := L^2(\R^2,dx)$ the configuration space for the particle. 
\end{notation}

 The free energy observable for the quantum field is given by the \fbox{\textit{second quantization} $\dG(\omega(K))$} of the operator $\omega(K)$ corresponding to the \textit{free energy of a massive boson}:
\begin{enumerate}[label=(H\arabic*$_{\mathrm{field}}$), leftmargin=*]
    \setcounter{enumi}{0}
    \item\label{h1_field_massive} \underline{The quantum field is massive} and the free energy of a massive boson is given by the multiplication operator $\omega(K) = \sqrt{K^2+m^2}$ ($m>0$) on the Hilbert space $L^2(\Rk^2)$. If one works in the configuration space representation, one considers $\omega(-i\mspace{1mu} \nabla_y) = \sqrt{-\Delta_y+m^2}$.
\end{enumerate}
 
\begin{notation}
    For a function $f$ on the momentum space, we denote by $f(K)$ the associated multiplication operator on $L^2(\Rk^2)$.
\end{notation}

The coupling between the quantum field and the particle is assumed to be \textit{scalar and linear}, so it is done through a \textit{Segal field operator at each point} $x \in \R^2$:
\begin{align*}
  \boxed{\phi(v_x)  =  \frac{1}{\sqrt{2}}\,\overline{\vphantom{\big|} a(v_x) + \ac(v_x) }}\, .
\end{align*}
In order to make sense of the latter, we impose the following assumption on $v_x$:
\begin{enumerate}[label=(H\arabic*$_{\mathrm{field}}$), leftmargin=*]
    \setcounter{enumi}{1}
    \item\label{h2_field_interaction} \underline{Interaction kernel $v_x$ and form factor $v$:} for all $x \in \R^2$ and almost all $k$, $v_x(k) = v(k) e^{-ik\cdot x}$, where $v \in L^2(\Rk^2)$, so that the Segal field operator is well defined. If one works in the configuration space representation, one considers $v_x = v(\cdot-x)$. The map $(x,k) \mapsto v_x(k)$ is called the interaction kernel; $v$ is called the form factor.
\end{enumerate}

The standard, physically motivated expression for $v$ is $\chi \mspace{2mu} \omega^{-1/2}$, with $\chi \in \ciz(\R^2)$ a radial ultraviolet cutoff (see Appendix \ref{app:model} or \cite[(3.4)]{acqftmpfh}, \cite[Section 14.5]{Arai}, \cite[Assumption 4.2]{hiroshima2019ground}). Given this physically motivated form, we introduce the following assumptions, which will be used in certain parts of the sequel and are verified for the physical expression:

\begin{enumerate}[label=(H\arabic*$_{\mathrm{field}}$), leftmargin=*]
    \setcounter{enumi}{2}
    \item\label{h3_field_rotation_invariance} \underline{Rotation invariance of the form factor $v$}:  $v(\mathcal{R} \mspace{3mu} \cdot) = v$ for every rotation matrix $\mathcal{R}$ acting on $\R^2$ (see, e.g., \cite[Condition 1.4]{moller_translation_invariant_nelson_model}).
\end{enumerate}

\begin{enumerate}[label=(H\arabic*$_{\mathrm{field}}$), leftmargin=*]
    \setcounter{enumi}{3}
    \item\label{h4_field_real_valued_configuration_space} \underline{The form factor $v$ is real-valued on the configuration space} i.e. $v(-k) = \cc{v(k)}$ in momentum space (see again \cite[Assumption 4.2]{hiroshima2019ground}).
\end{enumerate}
\vspace{0.3cm}

Throughout the remainder of this work, \textbf{hypotheses} \ref{h1_geom_regularity_curve}-\ref{h4_geom_not_u_shaped}, \ref{hyp:potential}, \ref{h1_field_massive}, \ref{h2_field_interaction} \textbf{will be assumed to hold} without further mention, whereas hypotheses \ref{h3_field_rotation_invariance} and \ref{h4_field_real_valued_configuration_space} will be explicitly invoked whenever they are needed.

%%% ------------------------------------------------------------
%===============================================================
%%% ------------------------------------------------------------

\subsection{The model}\label{sec:mode:subsec_the_model}

Combining the two previous sections, we can state the Hilbert space of the problem and the associated Hamiltonians, depending on the choice of coordinates for the waveguide (see Table \ref{tab:hamiltonian_comparison_coordinates}):
\begin{itemize}
    \item \underline{The Hilbert space of the problem:} 
    \begin{empheq}[box=\fbox]{alignat*=2}
    \Hi &= L^2(\Rx^2) \otimes \Fs\left( L^2\big(\Rk^2\big) \right) 
         &\hspace{0.5cm}\text{or}\hspace{0.5cm} 
     \Hi &= L^2(\mathcal{O}_\gamma) \otimes \Fs\left( L^2\big(\Rk^2\big) \right) \\
     &\cong \int_{\R^2}^{\oplus} \Fs\left( L^2\big(\Rk^2\big) \right) dx 
         &
     &\cong \int_{\mathcal{O}_\gamma}^{\oplus} \Fs\left( L^2\big(\Rk^2\big) \right) ds \mspace{2mu} du. 
    \end{empheq}

    \item \underline{The free Hamiltonian} is given by the Hamiltonian of the small system and the one of the free bosonic field:  
    \begin{alignat*}{2}
     \hspace{1cm} \Ha_\free &= \underbrace{(- \Delta + V_\bullet(X)) \otimes I}_{:= \, \Ha_\el} + \underbrace{I  \otimes \dG(\omega(K))}_{:= \, \Ha_f}  &\hspace{0.5cm} \text{or} \hspace{0.5cm}
     \Ha_\free &= \underbrace{(- \Delta_\gamma + V(U)) \otimes I}_{:= \, \Ha_\el} + I  \otimes \dG(\omega(K)) \\ 
     &:= \Ha_\el + \Ha_f
               &
     &:= \Ha_\el + \Ha_f.
    \end{alignat*}

    \item \underline{The interaction Hamiltonian} is given by the direct integral of the Segal field operator: 
    \begin{alignat*}{2}
    \hspace{1cm} \Ha_I &= \int_{\R^2}^{\oplus} \phi(v_x) \, dx  
     &\hspace{0.7cm} \text{or} \hspace{1cm} 
     \Ha_I &= \int_{\mathcal{O}_\gamma}^{\oplus} \phi(v_{\varphi(s,u)}) \, ds \mspace{2mu} du \\
     &= \int_{\R^2}^{\oplus} \frac{1}{\sqrt{2}}\,\overline{ a(v_x) + \ac(v_x) }\,dx 
    &
    &= \int_{\mathcal{O}_\gamma}^{\oplus} \frac{1}{\sqrt{2}}\,\overline{ a(v_{\varphi(s,u)}) + \ac(v_{\varphi(s,u)}) } \, ds \mspace{2mu} du.
    \end{alignat*}

    \item \underline{The Hamiltonian of the problem} is then given by the sum of the free one and the interaction one, with the intensity of the coupling controlled by a coupling constant $g>0$: 
    \begin{empheq}[box=\fbox]{align*}
        \Ha &=  \Ha_\el + \Ha_f + g\Ha_I = \Ha_\free +  g\Ha_I \\
            &=  (- \Delta + V_\bullet(X)) \otimes I  + I \otimes \dG(\omega(K)) + g\int_{\R^2}^{\oplus} \phi(v_x)\,dx \\
            &= (- \Delta_\gamma + V(U)) \otimes I  + I \otimes \dG(\omega(K)) + g\int_{\mathcal{O}_\gamma}^{\oplus} \phi(v_{\varphi(s,u)})\,ds \mspace{2mu} du. 
    \end{empheq}
\end{itemize}

%%% ------------------------------------------------------------
%               General ideas
%%% ------------------------------------------------------------

\subsection{General ideas and goals}\label{model:subsec_general_ideas}

In this subsection, we outline the spectral picture that guides the rest of the paper. For quantum waveguides without a quantized field, the curvature may alter the behavior of the particle and can in particular slow it down or even trap it, thereby creating \textit{curvature-induced bound states}. This phenomenon is especially meaningful since the curved and straight waveguides have the same essential spectrum (\cite[Proposition 3.1]{pv_spsqv}):
\begin{equation}\label{eq:sess_curved_straight_equal}
    \sigma\hspace{-1pt}\left(\ha_{\mspace{1mu}\el,S}\right) = \sess\hspace{-1pt}\left(\ha_{\mspace{1mu}\el,S}\right)
    =
    \sess\hspace{-1pt}\left(\ha_{\mspace{1mu}\el,C}\right)
    =
    \left[\varepsilon_0,+\infty\right[, 
\end{equation}
where $\varepsilon_0$ is the infimum of the transverse Schrödinger operator $-\partial_{\xd}^2+V(X_2)$. This equality is natural: the essential spectrum is governed by the behavior of the system at infinity (see Subsection \ref{prelim:subsec_persson}), and the curvature has compact support, so the waveguide is straight at infinity. Moreover, the essential spectrum has no gaps: it is a closed half-line.

\smallbreak

When the particle is coupled to a quantized field, we expect this gapless structure to persist. Physically, the geometry is still straight at infinity, and there is no apparent mechanism that should open gaps in the essential spectrum. However, the bottom of the essential spectrum is no longer clear a priori, since the interaction with the field may shift the relevant thresholds. The purpose of \textbf{Section~\ref{sec:ess_spec}} is therefore to \textbf{determine and compare the essential spectra} of the coupled Hamiltonians associated with the straight and curved waveguides.

\smallbreak

Once the essential spectrum is understood, the next question is whether the curved coupled system has a non-empty discrete spectrum. \textbf{Section~\ref{sec:binding}} derives an abstract \textbf{sufficient condition ensuring the existence of such a discrete spectrum}. Turning this condition into an explicit geometric criterion on the waveguide is left for future work.

\smallbreak

The following figures illustrate the existing spectral results for quantum waveguides, which serve as the reference picture for the coupled model.

% --- Figure 1 : Cas droit ---
\begin{figure}[H]
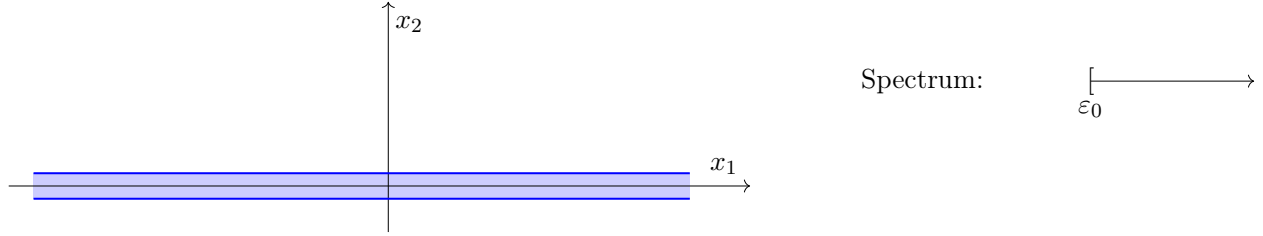

    \centering
    \resizebox{\textwidth}{!}{
 }
    \caption{Straight case: operator $ -\Delta + V_S(X)$.}
    \label{fig:straight_waveguide_and_spectrum}
\end{figure}

% --- Figure 2 : Cas courbé ---
\begin{figure}[H]
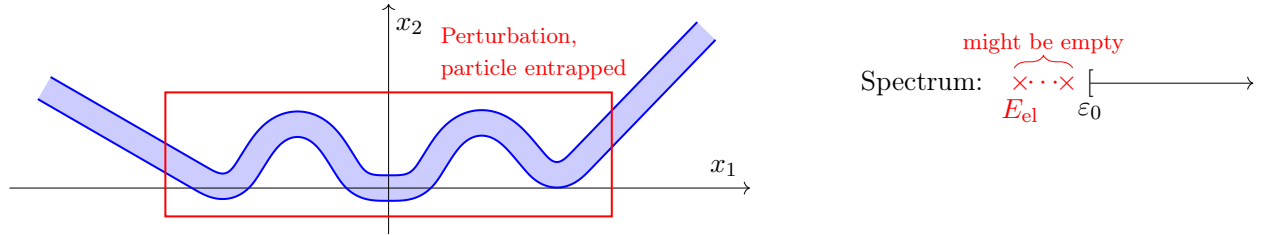

    \centering
    \resizebox{\textwidth}{!}{%
 }
    \caption{Curved case: operator $-\Delta + V_C(X)$.}
    \label{fig:curved_waveguide_and_spectrum}
\end{figure}

%===============================================================
%           À FAIRE DANS CETTE SECTION
%===============================================================

%================================================================
 
% 3) Hypothèses et simplifications
\section{Preliminary results}
\label{sec:prelim}

In this section, we establish some straightforward preliminary results in order to get a better understanding of the subject.

%==========================================================
%       Self-adjointness
%==========================================================

\subsection{Self-adjointness, domain and associated quadratic form}\label{prelim:subsec_domain}
The first preliminary result is obviously the self-adjointness of $\Ha$, for which we need the self-adjointness and lower semi-boundedness of $\Ha_\el$, $\Ha_f$.

\begin{lemma}\label{sec:prelim:hfree_self_adjoint}
    The operators $\Ha_\el$ and $\Ha_f$ strongly commute so that $\Ha_\el$, $\Ha_f$ and $\Ha_\free$ (with domain $\D(\Ha_\free) = \D(\Ha_\el) \cap \D(\Ha_f)$\textup{)} are self-adjoint and lower semi-bounded. 
\end{lemma}
\begin{proof}
    One uses \cite[Proposition 3.5(ii), Theorem 3.5(iv), Theorem 3.8(iv)]{Arai} to argue as in \cite[Proposition 2.4(1)]{hiroshima2020feynman2}, since $V_\bullet(X)$ is bounded and hence relatively $-\Delta$-bounded with bound $< 1$. 
\end{proof}

\noindent We also record the self-adjointness of the interaction Hamiltonian, which will allow us to use its associated quadratic form below.

\begin{lemma}
    The interaction Hamiltonian $\Ha_I$ is self-adjoint.
\end{lemma}
\begin{proof}
    For all $x \in \R^2$, $v_x \in L^2(\R_k^2)$ so $\phi(v_x)$ is self-adjoint (\cite[Theorem 5.23]{Arai}) and \cite[Theorem 2.7]{Arai} gives the result upon proving that $x \mapsto \left(\phi(v_x)+i\right)^{-1}$ is measurable. To show this, it suffices to write, thanks to \cite[Theorems 1.32, 4.15, 5.32]{Arai}, that for all $x \in \R^2$: 
    \begin{align*}
        \phi(v_x) = \Gamma\big(e^{\mspace{1mu}-i \mspace{1mu} K \cdot \mspace{2mu} x}\big) \mspace{2mu}\phi(v) \mspace{2mu}\Gamma\big(e^{\mspace{1mu}i \mspace{1mu} K \cdot \mspace{2mu} x}\big) 
        \Longrightarrow \left(\phi(v_x)+i\right)^{-1} = e^{\mspace{2mu}-\mspace{1mu}i \mspace{3mu} \dG(K) \mspace{2mu}\cdot \mspace{3mu} x}\mspace{2mu}\left(\phi(v) + i\right)^{-1} \mspace{2mu} e^{\mspace{2mu} i \mspace{3mu} \dG(K) \mspace{2mu}\cdot \mspace{3mu} x}.
    \end{align*}
    The maps $x \mapsto e^{\mspace{2mu} \pm \mspace{2mu}i \mspace{3mu} \dG(K) \mspace{2mu}\cdot \mspace{3mu} x}$ are strongly continuous (as evolution groups) and hence measurable; therefore, so is the map $x \mapsto \left(\phi(v_x)+i\right)^{-1}$.
\end{proof}

\vspace{-10pt}

\begin{theorem}\label{prelim:th:self-adjointness_Ha}
    The Hamiltonian $\Ha$ with domain $\D(\Ha) = \D(\Ha_\free)$ is self-adjoint and lower semi-bounded for any value of the coupling constant $g>0$.
\end{theorem}
\begin{proof}
    It is sufficient to apply the Kato-Rellich theorem, proceeding as in \cite[Proposition 2.4(2)]{hiroshima2020feynman2}. As it will be useful in the sequel, we compute explicitly the relative bound $b_\varepsilon$. To do so, replace in their computations $\varphi / \sqrt{\omega}$ by our form factor $v$. Let $C = 2 \nld{v / \sqrt{\omega}} + \nld{v} $, which is well defined because $\omega^{-1/2} \leq m^{-1/2}$ and $v \in L^2(\Rk^2)$ by Assumption~\ref{h2_field_interaction}. Using the inequality $\lambda^{1/2} \leq \delta \lambda  + \frac{1}{4\delta}$ ($\delta > 0$), the functional calculus, and setting $\delta = \varepsilon / C$ ($\varepsilon > 0$), we find: 
    \begin{align*}
        b_\varepsilon = \varepsilon \abs{1 - E_\el} + \frac{C^2}{4\varepsilon},
    \end{align*}
    where $E_\el = \inf \sigma(\Ha_\el)$. Upon multiplying everything by $g>0$ and taking $\varepsilon>0$ small enough, one gets the result.
\end{proof}

Thanks to the fact that $V_\bullet(X)$ is $-\Delta$-bounded with relative bound $<1$, the abstract domain obtained above can be described more explicitly. We collect in the following proposition the corresponding domain identification, a useful splitting of the Schrödinger operator, and the resulting core properties. The vector-valued Sobolev and $L^2$ spaces used below are defined in Appendix \ref{app:vector-valued}. From this point on, we shall rely on the framework, notation, and conventions introduced in that appendix.

\begin{proposition}\label{sec:prelim:domain_splitting_and_cores}
The following assertions hold.
\begin{enumerate}[label=(\roman*)]
\item If $W(X)$ is a multiplication operator on $L^2(\Rx^2)$ which is $-\Delta$-bounded with relative bound $<1$, then:
\begin{align*}
    (-\Delta + W(X)) \otimes I 
    = -\partial_{\xu}^2 \otimes I 
    - \partial_{\xd}^2 \otimes I 
    + W(X) \otimes I
\end{align*}
and $\D\big((-\Delta + W(X)) \otimes I\big) = \D(-\Delta \otimes I)$.
\item One has:
\begin{align*}
\D(\Ha)
= \D(-\Delta \otimes I) \cap \D(\Ha_f)
\cong  H^2(\R^2,\Fs) \cap L^2\hspace{-2pt}\left(\R^2,\D\big(\dG(\omega(K))\big) \right).
\end{align*}
\item Any core for $-\Delta \otimes I + I \otimes \dG(\omega(K))$ is a core for $\Ha$. In particular, if $\D$ is a core for $\dG(\omega(K))$, then $\ciz(\R^2) \alten \mathcal \D$ is a core for $\Ha$.
\end{enumerate}
\end{proposition}

\begin{proof}
Statement $(i)$ follows from Proposition~\ref{app:op:standard_properties}(iii) (since $-\partial_{\xu}^2$, $-\partial_{\xd}^2$ strongly commute and are positive) and Proposition~\ref{app:op:tensor_prod_relative_boundedness}. \\
We prove $(ii)$. Since $V_\bullet(X)$ is a bounded multiplication operator, $(i)$ gives 
$\D(\Ha_\el)=\D(-\Delta\otimes I)$. Using Appendix \ref{app:vector-valued}, we find:
\begin{align*}
\D(\Ha_f)
\cong
L^2\hspace{-2pt}\left(\R^2,\D\big(\dG(\omega(K))\big)\right), \hspace{1cm} \D(-\Delta \otimes I)\cong H^2(\R^2,\Fs).
\end{align*}
Since Theorem \ref{prelim:th:self-adjointness_Ha} gives $\D(\Ha)=\D(\Ha_\free)=\D(\Ha_\el)\cap\D(\Ha_f)$, the claim follows. \\
Finally, we prove $(iii)$. Set $\Ha_0:=-\Delta\otimes I+I\otimes\dG(\omega(K))$. The same argument as in the proof of Theorem \ref{prelim:th:self-adjointness_Ha}, with $-\Delta \otimes I$ and $\Ha_0$ instead of $\Ha_\el$ and $\Ha_\free$ respectively, shows that $\Ha_I$ is infinitesimally $\Ha_0$-bounded. Since $V_\bullet(X) \otimes I$ is bounded, the Kato-Rellich theorem gives $\D(\Ha)=\D(\Ha_0)$, with any core for $\Ha_0$ being a core for $\Ha$.
\end{proof}

The natural subsequent step is to introduce the quadratic form associated with the model, which we will use at the beginning of Section~\ref{sec:ess_spec}. 

\begin{notation}
    The domain of a quadratic form $\q$ is denoted $\D(\q)$, and if it is associated with a self-adjoint operator $A$, we denote it $\q_A$, with domain $\Q(A)$. One has $\D(\q_A) = \Q(A) = \D\big(\abs{A}^{1/2}\big)$. \\
    The form sum of two self-adjoint operators, when it is well defined, is denoted $A \fsum B$.
\end{notation}

\begin{remark}
    Note that the operator $-i \nabla$ extends to $L^2(\R^2,\Fs)$: see again Appendix \ref{app:vector-valued}.
\end{remark}

\begin{proposition}\label{prelim:prop:general_quadratic_form}
    The Hamiltonian $\Ha = \Ha_\el + \Ha_f + g \Ha_I$ is the unique self-adjoint operator associated with the lower semi-bounded quadratic form $\q_\Ha$ defined by:
    \begin{align*}
        \q_\Ha[\psi] = \nldfs{\nabla\psi}^2 + \q_V[\psi]  + \q_{\Ha_f}[\psi] + g\,\q_{\Ha_I}[\psi] =  \q_{\Ha_\el}[\psi]  + \q_{\Ha_f}[\psi] + g \, \q_{\Ha_I}[\psi],
    \end{align*}
    with form domain:
    \begin{align*}
        \Q(\Ha) = \Q(\Ha_\el) \cap \Q(\Ha_f) = \Q(-\Delta \otimes I) \cap \Q(\Ha_f) = H^1(\R^2,\Fs) \cap L^2\hspace{-2pt}\left(\R^2,\D\big(\dG(\omega(K))^{1/2}\big) \right).
    \end{align*}
\end{proposition}

\begin{proof}
Since $\Ha_I$ is $\Ha_\free$-bounded in the operator sense, it is also
$\Ha_\free$-bounded in the form sense \cite[Theorem \cRM{10}.18]{RS}. Hence, by the KLMN theorem, the form $\q_\Ha=\q_{\Ha_\el}+\q_{\Ha_f}+g \, \q_{\Ha_I}$ is closed and lower semi-bounded on $\Q(\Ha_\el)\cap\Q(\Ha_f)$, and therefore defines a unique self-adjoint operator $\Ha_\el \fsum \Ha_f \fsum g\mspace{2mu}\Ha_I$. \\
The operator sum is a restriction of the form sum; see, e.g.,
\cite[Proposition 10.22]{Sch}. But $\Ha$ is self-adjoint by
Theorem~\ref{prelim:th:self-adjointness_Ha}, and since a self-adjoint operator has no proper self-adjoint extension, we get the equality of the operator sum and the form sum. Consequently,
$\Ha$ is the unique self-adjoint operator associated with $\q_\Ha$. \\
It remains to identify the form domain. The previous proposition gives $\D(\Ha_\el) = \D(-\Delta\otimes I)$, so \cite[Theorem 1.34]{Arai} yields
$\Q(\Ha_\el)=\Q(-\Delta\otimes I) = \D(-i \mspace{1mu}\nabla \otimes I)$. The vector-valued identifications of
Appendix~\ref{app:vector-valued} then give the result.
\end{proof}

%==========================================================
%       Ground state when small coupling constant
%==========================================================

\subsection{Existence of a ground state at small coupling constant}

Another quick result to obtain is the existence of a ground state \textit{if the coupling constant is small enough}. Indeed, perturbation theory (see \cite[Chapter \cRM{12}]{RS}) enables us to know that eigenvalues remain eigenvalues under small perturbations. In our case, \textit{since we chose massive bosons}, we know that $\Ha_{\free}$ has non-empty discrete spectrum at the bottom of its spectrum, provided we assume the same for $-\Delta+V_C(X)$. This is what we prove in this subsection, beginning with the following lemma:

\begin{lemma}
One has:
\begin{enumerate}[label=(\roman*)]
\item The spectrum of $\dG(\omega(K))$ is given by $\sigma\big(\dG(\omega(K))\big) = \{0\} \cup [m,+\infty[$ and the vacuum $\Omega$ is a non-degenerate ground state.

\item $\sigma(\Ha_\free) = \sigma(-\Delta+V_\bullet(X)) + \sigma\big(\dG(\omega(K))\big)$.  

\item  In the curved case, assume $\sdisc(-\Delta+V_C(X)) \neq \emptyset$. Then $\sdisc(\Ha_\free) \neq \emptyset$.
\end{enumerate}
\end{lemma}
\begin{proof}
    $(i)$ follows from $\omega(K) \geq m > 0$ (see \ref{h1_field_massive}) and \cite[Theorem 4.9(\romannumeral 3)]{Arai}; about the vacuum, see, e.g., \cite[Section \cRM{13}.12, Example 3]{RS}. \\
    $(ii)$ follows from \cite[Theorem 3.8(\romannumeral 1)]{Arai} and the fact that the sum of two lower bounded closed sets in $\R$ is closed. \\
    $(iii)$ is a direct application of \cite[\S 8.2.3]{C0commutatorGeorgescuAmrein} or \cite[Theorem 3.14(\romannumeral 1)]{Arai}, together with the fact that the infima of both spectra, $E_\el$ and $0$, are isolated in $\sigma(-\Delta+V_C(X))$ and $\sigma\big(\dG(\omega(K))\big)$, respectively.
\end{proof}

This is summed up in the following scheme: 

\begin{figure}[h!]
\centering
% --- Subfigure 1 ---
\begin{subfigure}[b]{0.32\textwidth}
    \centering
    \begin{tikzpicture}[scale=0.75] % Echelle ajustée pour matcher la 3ème
        % Draw the half-axis starting from epsilon_0
        \draw[->] (1.1,0) -- (4,0);
    
        % Mark the starting point at epsilon_0 with a bracket on the axis and epsilon_0 below it
        \draw (1.1,0) node[] {$[$} node[yshift=-12pt] {$\varepsilon_0$} ;
    
           % Draw three crosses aligned to the left of the axis
        \draw (-0.5, 0) node[] {$\times$} node[yshift=-12pt] {$E_{\el}$};
        \draw (0.05, 0) node[] {$\cdots$};
        \draw (0.6, 0) node[] {$\times$};      
    \end{tikzpicture}
    \caption{$\sigma(-\Delta+V_C(X))$}
\end{subfigure}
\hfill
% --- Subfigure 2 ---
\begin{subfigure}[b]{0.32\textwidth}
    \centering
    \begin{tikzpicture}[scale=0.75] % Echelle ajustée
        % Draw the half-axis starting from epsilon_0
        \draw[->] (1,0) -- (4,0);
    
        % Mark the starting point at epsilon_0 with a bracket on the axis and epsilon_0 below it
        \draw (1,0) node[] {$[$} node[yshift=-12pt] {$m$} ;
    
           % Draw three crosses aligned to the left of the axis
        \draw (0, 0) node[] {$\times$} node[yshift=-12pt] {$0$};
        
    \end{tikzpicture}
    \caption{$\sigma\big(\dG(\omega(K))\big)$}
    \label{sigma_dgamma}
\end{subfigure}
\hfill
% --- Subfigure 3 ---
\begin{subfigure}[b]{0.32\textwidth}
    \centering
    \begin{tikzpicture}[scale=0.75]
        % Draw the half-axis starting from epsilon_0
        \draw[->] (1.1,0) -- (4,0);
    
        % Mark the starting point at epsilon_0 with a bracket on the axis and epsilon_0 below it
        \draw (1.1,0) node[] {$[$} node[yshift=-12pt] {$\alpha$} ;
    
           % Draw three crosses aligned to the left of the axis
        \draw (-0.5, 0) node[] {$\times$} node[yshift=-12pt] {$E_{\el}$};
        \draw (0.05, 0) node[] {$\cdots$};
        \draw (0.6, 0) node[] {$\times$};        
    \end{tikzpicture}
    \caption{$\sigma(\Ha_{\free})$}
    \label{sp_hfree}
\end{subfigure}
\caption{Spectra of the components of the free Hamiltonian, $\alpha = \min(\varepsilon_0, E_{\el}+m)$.}
\label{fig:three_spectra}
\end{figure}

\begin{remark}
    Note that \textit{this only holds because we consider massive bosons}: otherwise, we would have $\sigma\big(\dG(\omega(K))\big) = [0,+\infty[$ and $\Ha_\free$ would not have any discrete spectrum. \\
    Likewise, from a geometric standpoint, we only focus on the curved case since the straight potential $V_S$ is known to yield an empty discrete spectrum (see Equation \eqref{eq:sess_curved_straight_equal}, Figure \ref{fig:straight_waveguide_and_spectrum}).
\end{remark}

Now that we know that $\Ha_\free$ has a non-empty discrete spectrum (provided we assume the same for $-\Delta+V_C(X)$), we know that the same will be true for $\Ha_C$ for small enough values of the coupling constant $g$.

\begin{proposition}
Let $\lambda$ be a discrete eigenvalue of $\Ha_{\free}$ of multiplicity $p$. Then there are $p$ not necessarily distinct single-valued functions $\lambda^{(1)}, \dots,\lambda^{(p)}$, analytic near $g=0$, with $\lambda^{(k)}(0) = \lambda$, that are, for $g$ near 0, the only eigenvalues (counted with multiplicity) of $\Ha_C$ in a neighborhood of $\lambda$. Let $d = \frac{1}{2} \text{dist}\big(\lambda, \sigma(\Ha_{\free} )\setminus \{\lambda\}\big)$ and $P(g)$ the Riesz projection onto the eigenspaces associated with $\lambda^{(1)}, \dots,\lambda^{(p)}$. The radius of convergence $r$ of the Taylor series of $P(g)$ at $g=0$ satisfies:
\begin{align*}
r \geq \frac{d}{C \sqrt{\abs{\lambda} + 2d + \abs{1-E_\el}}}, 
\end{align*} 
where $C = 2\nld{v / \sqrt{\omega}} + \nld{v} $, $E_\el = \inf\sigma\bigl(-\Delta+V_C(X)\bigr)$. 
\end{proposition}

\begin{proof}
Using the relative boundedness proved in Theorem \ref{prelim:th:self-adjointness_Ha}, \cite[Vol. IV, Lemma on p. 16]{RS} and the paragraph above it, we deduce that $\Ha_\free + g\Ha_I$ defines an analytic family of type (A) near $g=0$ and that we can apply \cite[Theorem \cRM{12}.13]{RS}. This yields the first part of the proposition. \\ 
Using \cite[Remarks 2.9 and 2.10 in Section 7.2.3]{kato_book} (with the closed curve $\Gamma$ being the circle of center $\lambda$ and radius $d$), we can extend the formula for the radius of convergence given in \cite[Theorem \cRM{12}.11]{RS} to the degenerate case, provided we consider the Riesz projection $P(g)$ onto the eigenspaces associated with $\lambda^{(1)}, \dots,\lambda^{(p)}$ rather than the non-degenerate eigenvalue from the theorem. Because of the infinitesimal boundedness, we can choose the relative bounds, so we find:
\begin{align*}
%if a,x \geq 0, \sqrt(a^2+x^2) \leq \sqrt(a^2+x^2+2ax) = a+x
r \geq \underset{\varepsilon >0}{\max} \, \frac{d}{\varepsilon(\abs{\lambda}+2d) + \varepsilon \abs{1-E_\el} + \frac{C^2}{4\varepsilon}} = \frac{d}{C \sqrt{\abs{\lambda} + 2d + \abs{1-E_\el}}}. 
\end{align*} 
\end{proof}

%====================================================================
%    Another proof of the essential spectrum of the small system
%====================================================================

\subsection{A first example of Persson's formula}\label{prelim:subsec_persson}

To conclude this section, we introduce a \textit{method of proof that will be used in the next section}, by applying it first to an already known result. More precisely, we give a new proof of the equality of the bottoms of the essential spectra of the straight and curved soft waveguides, established in \cite[Proposition 3.1]{pv_spsqv}, without using the Neumann bracketing technique of the original proof. This provides a simplified prototype of the argument that will later be adapted to the coupled system, where the quantum field has to be incorporated. The details are given here in the uncoupled setting so that, in the next section, we will only have to explain the additional points due to the field. The method is based on Persson's formula, first stated in \cite{Persson1960}; we shall use the version stated in \cite[Theorem 3.12]{SchroOPGlobalGeom} and \cite[Theorem 14.11]{hislop1996introduction}.

\begin{proposition}[Persson's formula]
    Let $V$ be a real-valued potential that is $-\Delta$-bounded with relative bound $< 1$ and let us denote $\D_R := \ciz\big(\R^n \setminus B(0,R)\big)$, where $B(0,R)$ is the ball of center $0$ and radius $R> 0$. Then:
    \begin{align*}
     \inf\sess (-\Delta+V(X)) = \lim_{\rule{0pt}{1.6ex} R \to +\infty} \hspace{0.2cm} \inf_{\substack{\rule{0pt}{1.6ex} \psi \, \in \, \D_R \\[2.5pt] \norm{\psi} \, = \, 1}} \: \ip{\psi}{(-\Delta+V(X))\psi}   :=  \underset{R \to +\infty}{\lim} \Sigma_R.
    \end{align*}
\end{proposition}

Persson's formula expresses the fact that the bottom of the essential spectrum is determined by the behavior of the Schrödinger operator at infinity. For a waveguide with compactly supported curvature (Hypothesis~\ref{h2_geom_compact_curvature}), infinity corresponds to the straight part of the guide. One therefore expects the curved and straight waveguides to have the same essential spectrum, as stated in \cite[Proposition 3.1]{pv_spsqv} and recalled in Equation~\eqref{eq:sess_curved_straight_equal}. The figure below illustrates this idea.

\begin{figure}[H]
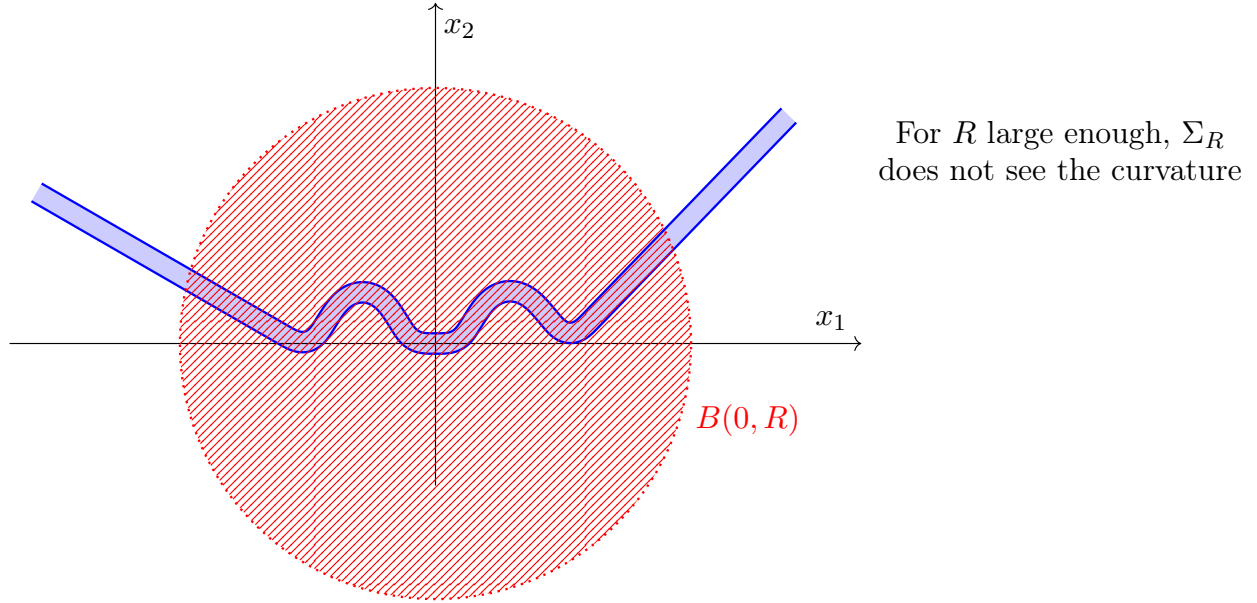

\centering
\resizebox{1\textwidth}{!}{
 }
\caption{The use of Persson's formula for the quantum waveguide.}
\label{fig:persson}
\end{figure}

We now turn this idea into a new proof of \cite[Proposition 3.1]{pv_spsqv}, as announced above.

\begin{proposition}\label{sec:prelim:prop_persson_small_system}
    The Hamiltonians $\ha_{\mspace{1mu}\el,S}$ and $\ha_{\mspace{1mu}\el,C}$, corresponding to the straight and curved waveguides respectively, have the same essential spectrum. 
\end{proposition}
\begin{proof}
    We rely on the Weyl-sequence construction from the proof of \cite[Proposition 3.1]{pv_spsqv}, which shows that every energy above the straight threshold belongs to both essential spectra, so it only remains to show that their infima are equal, which we establish using the previous proposition (whose hypotheses are satisfied since $V_\bullet(X)$ is bounded). Denote $\Sigma^S_R$ and $\Sigma^C_R$ the quantities from Persson's formula for the straight and the curved waveguide respectively. It suffices to show $\underset{R \to +\infty}{\lim}\Sigma_R^C \geq \underset{R \to +\infty}{\lim}  \Sigma_R^S$ and $\underset{R \to +\infty}{\lim}\Sigma_R^C \leq \underset{R \to +\infty}{\lim}  \Sigma_R^S$, which we will do using the IMS localization formula and a partition of $\R^2$ close to the one used in \cite[Proposition 3.1]{pv_spsqv}. Since we take the limit $R \to +\infty$, all the arguments only need to hold for $R$ sufficiently large.
    \smallbreak
    The curved part of the waveguide corresponds to the curvilinear coordinates in $\supp \kappa \times ]-a,a[$, which is a subset of $\mathcal{O}_\gamma$ by Assumption~\ref{h3_geom_no_self_intersection}. Taking $R_0 = \underset{\supp \, \kappa}{\sup} \abs{\gamma} + a$, we see that all of this curved part is contained in $B(0,R_0)$ (see Appendix \ref{app:curve} for curvilinear coordinates).  
    \smallbreak
    \noindent Consequently for $R \geq R_0$ we are left, outside of $B(0,R)$, with 3 parts: the two \textit{straight} branches of the waveguide outside the ball and the rest, where $V_C = 0$. Because we excluded the curvature with $B(0,R)$, the Hamiltonian is $-\partial_s^2-\partial_u^2+V(U)$ on the two straight branches (see \eqref{eq:curved_electronic_hamiltonian}), but it does not exactly correspond to the straight operator $\ha_{\mspace{1mu}\el,S}$ since the transverse variable $u$ does not live in $\R$ but only in ${]-c_-(s), c_+(s)[}$. However, since we assumed that the waveguide is not U-shaped in Assumption~\ref{h4_geom_not_u_shaped}, we have $c_\pm(s) \underset{\abs{s} \to + \infty}{\longrightarrow}+\infty$, so at infinity, the transverse variable runs through $\R$ and we should retrieve the straight Hamiltonian. This is also the core idea of the proof from \cite{pv_spsqv}, where the arc-length parameter must be taken arbitrarily far along the branches so that the transverse variable can range over arbitrarily large intervals. Here, we choose to use the IMS localization formula. 

    \smallbreak
    
    For the remainder of the proof, we assume $R \geq R_0$ and we fix an origin $s=0$ for the arc-length parametrization of $\gamma$. For $R$ large enough that $\varphi( \{0 \} \times [-a,a]) \subset B(0,R)$, let $s_+(R)>0$ and $s_-(R)<0$ denote the exit parameters of the positive and negative branches from $B(0,R)$, that is, the first arc-length parameters encountered when moving from $s=0$ toward $+\infty$ and $-\infty$, respectively, for which the transverse segment $\varphi(\{s\}\times[-a,a])$ meets $\partial B(0,R)$: 
    \begin{align*}
    s_+(R) &:= \sup\ens{ 
        s > 0 \,}{\, \varphi\bigl( [0,s] \times [-a,a] \bigr) \subset B(0,R)} \\
    s_-(R) &:= \inf\ens{ 
        s < 0 \,}{\, \varphi\bigl( [s,0] \times [-a,a] \bigr) \subset B(0,R) }.
    \end{align*}
    Observe that since each branch is straight at infinity, it differs from the straight waveguide by a fixed Euclidean transformation, so one has $s_\pm(R) = \pm R + O(1)$ as $R \to +\infty$.  Now let $S_\pm(R):=s_\pm(R)\mp\sqrt{R}$ denote the longitudinal transition thresholds used below in the construction of the partition of unity and let
    \begin{align*}
        u_+(R) = \frac{1}{2} \: \, \underset{s \, \geq \, S_{+}(R)}{\inf}\min \big(c_-(s), c_+(s)\big) \\
        u_-(R) = \frac{1}{2} \: \, \underset{s \, \leq \, S_{-}(R)}{\inf}\min \big(c_-(s), c_+(s)\big)
    \end{align*}
    denote minimal safe half-widths ensuring that the strips $\Omega_{R,+}$ (defined below) avoid the cut locus. We set: 
    \begin{align*}
        \Omega_{R,+} = \ens{\varphi(s,u)}{(s,u) \in \mathcal{O}_\gamma, \, s \geq S_+(R), \,u \in [-u_+(R), u_+(R)]} \subset \R^2 \setminus \Cut(\gamma)\\
        \Omega_{R,-} = \ens{\varphi(s,u)}{(s,u)\in \mathcal{O}_\gamma, \, s \leq S_-(R), \,u \in [-u_-(R), u_-(R)]} \subset \R^2 \setminus \Cut(\gamma).
    \end{align*}
    Let $\theta\in\ci(\R,[0,1])$ with $\theta(t)=0$ for $t\leq-1$ and $\theta(t)=1$ for $t\geq0$, and let $\rho\in\ci(\R,[0,1])$ with $\rho=1$ on $[-\frac12,\frac12]$ and $\rho=0$ outside $[-1,1]$. Set
    \begin{align*}
    \eta_{R,+}(s,u) &:= \theta\left(\frac{s-s_+(R)}{\sqrt R}\right) \rho\left(\frac{u}{u_+(R)}\right)\\
    \eta_{R,-}(s,u) &:= \theta\left(\frac{s_-(R)-s}{\sqrt R}\right) \rho\left(\frac{u}{u_-(R)}\right)\\
    N_R &:= \left( \eta_{R,+}^2+\eta_{R,-}^2 +\left(1-\eta_{R,+}-\eta_{R,-}\right)^2 \right)^{1/2},
    \end{align*}
    and define
    \begin{align*}
    \chi_{R,\pm} &:= \frac{\eta_{R,\pm}}{N_R} & \chi_R &:= \frac{1-\eta_{R,+}-\eta_{R,-}}{N_R}.
    \end{align*}
    For $R$ sufficiently large, the supports of $\eta_{R,+}$ and $\eta_{R,-}$ are disjoint; hence $1-\eta_{R,+}-\eta_{R,-}\in[0,1]$, and whenever one of the functions $\eta_{R,\pm}$ equals one, the other vanishes. Moreover, the Cauchy--Schwarz inequality gives $1 = \left( \eta_{R,+}+\eta_{R,-} +1-\eta_{R,+}-\eta_{R,-} \right)^2 \leq 3N_R^2$ so that $N_R\geq 1/\sqrt3$. Hence these functions are smooth and satisfy $\chi_{R,+}^2+\chi_{R,-}^2+\chi_R^2=1$. Furthermore,
    \begin{align*}
    \supp(\chi_{R,\pm}\circ\varphi^{-1}) = \supp(\eta_{R,\pm}\circ\varphi^{-1}) \subset\Omega_{R,\pm}.
    \end{align*}
    Then $\tilde{j}_{R,\pm}:=\chi_{R,\pm}\circ\varphi^{-1}$ and $\tilde{j}_{R}:=\chi_R\circ\varphi^{-1}$ ($\varphi$ is the change of variables from Subsection~\ref{subsec:pres_waveguide}) define a partition of unity (see Definition~3.1 in \cite{SchroOPGlobalGeom}) adapted to the decomposition of $\R^2\setminus B(0,R)$ into $\Omega_{R,+}$, $\Omega_{R,-}$, and $\R^2\setminus \left(B(0,R)\cup\Omega_{R,+}\cup\Omega_{R,-}\right)$. The functions $\tilde{j}_{R,\pm}$ vanish in a neighborhood of the cut locus, so their extensions by zero across it are smooth on $\R^2$, whereas $\tilde{j}_{R}$ is equal to one in a neighborhood of the cut locus and therefore extends smoothly by one across it. The geometry of this partition is illustrated in the following figure:

    \begin{figure}[H]
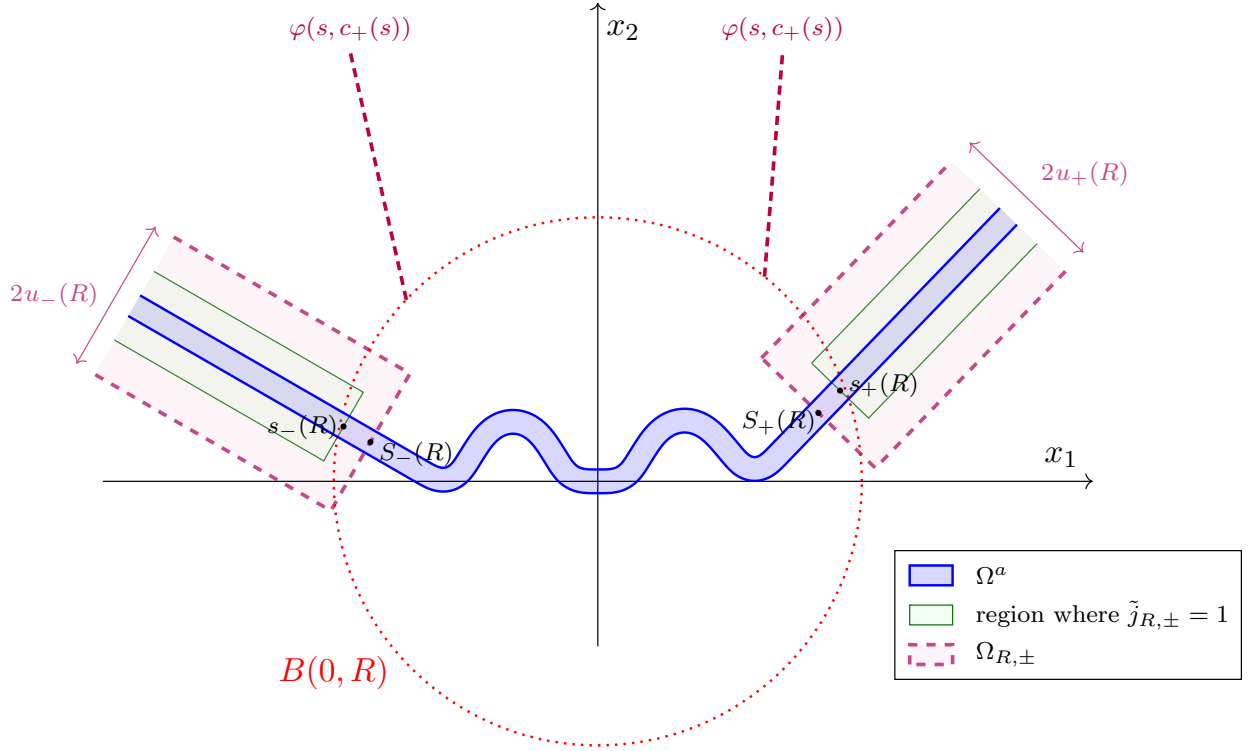

    \centering
    \resizebox{1\textwidth}{!}{
 }
    \caption{Schematic representation of the partition of unity used in the proof. The cutoffs $\widetilde j_{R,\pm}$ are supported in $\Omega_{R,\pm}$, whereas $\widetilde j_R$, which is not shown in the figure, is supported outside the green regions. Only part of the cut locus is depicted for illustration, and the figure is not to scale.}
    \label{fig:partition_of_unity}
    \end{figure}

    Using the IMS localization formula for Schrödinger operators (see, e.g., \cite[Theorem 3.2]{SchroOPGlobalGeom}), one has, for $\psi \in \D_R$ with $\norm{\psi}=1$: 
    \begin{align*}
         \ip{\psi}{\ha_{\mspace{1mu}\el,C}\mspace{1mu}\psi} &= \ip{\tilde{j}_{R}\mspace{1mu}\psi}{\ha_{\mspace{1mu}\el,C}\mspace{1mu}\tilde{j}_{R}\mspace{1mu}\psi} + \sum_{i=+,-}\ip{\tilde{j}_{R,i}\mspace{1mu}\psi}{\ha_{\mspace{1mu}\el,C}\mspace{1mu}\tilde{j}_{R,i}\mspace{1mu}\psi} +\underset{R \, \to \,  +\infty}{o\mspace{2mu}(1)}.
    \end{align*}
    Let us look at each term separately (in particular, we justify below that the last term is indeed negligible). \\  
    For $R$ sufficiently large, the values $\abs{S_\pm(R)}$ are large enough to ensure that $\tilde{j}_{R,\pm}\mspace{2mu}\psi$ are supported in disjoint regions where $\kappa = 0$. Upon passing to curvilinear coordinates via the unitary straightening map, and canonically identifying a function with support in $\mathcal{O}_\gamma$ through $L^2(\mathcal{O}_\gamma, ds \mspace{1mu} du) \hookrightarrow L^2(\R^2, ds \mspace{1mu} du)$, this means:
    \begin{align*}
        \ip{\tilde{j}_{R,\pm}\mspace{1mu}\psi}{\ha_{\mspace{1mu}\el,C}\mspace{1mu}\tilde{j}_{R,\pm}\mspace{1mu}\psi}= \ip{\tilde{j}_{R,\pm}\mspace{1mu}\psi}{\ha_{\mspace{1mu}\el,S}\mspace{1mu}\tilde{j}_{R,\pm}\mspace{1mu}\psi}.
    \end{align*}
    However, after the above change of coordinates and identification with the reference straight waveguide, the support of $\tilde{j}_{R,\pm}\mspace{1mu}\psi$ need not remain outside $B(0,R)$. Since this identification is given on each straight branch by a fixed Euclidean transformation, there exists a constant $b>0$, independent of $R$, such that $\supp (\tilde{j}_{R,\pm}\mspace{1mu}\psi) \subset \R^2\setminus B(0,R-b)$ in $(s,u)$ coordinates. Hence, by definition of $\Sigma_{R-b}^S$:
    \begin{align*}
        \ip{\tilde{j}_{R,\pm}\mspace{1mu}\psi}{\ha_{\mspace{1mu}\el,C}\mspace{1mu}\tilde{j}_{R,\pm}\mspace{1mu}\psi}= \ip{\tilde{j}_{R,\pm}\mspace{1mu}\psi}{\ha_{\mspace{1mu}\el,S}\mspace{1mu}\tilde{j}_{R,\pm}\mspace{1mu}\psi} \geq \Sigma_{R-b}^S \norm{\tilde{j}_{R,\pm}\psi}^2.  
    \end{align*}
    Since $S_\pm(R) \underset{R\to + \infty}{\to} \pm \infty$ and the cut radii maps go to infinity with the arc length, we have $u_\pm(R)\underset{R\to + \infty}{\to} +\infty$ and $u_\pm(R)\geq 2a$ for $R$ sufficiently large. Hence, at every point of $\Omega^a\setminus B(0,R)$, the corresponding curvilinear coordinates satisfy either $s\geq s_+(R)$ or $s\leq s_-(R)$, and $\abs{u}\leq a$. Therefore, either $\tilde{j}_{R,+}=1$ or $\tilde{j}_{R,-}=1$ at such a point, so that $\tilde{j}_{R}=0$. It then follows from $\psi = 0$ on $B(0,R)$ that $\supp (\tilde{j}_{R} \mspace{1mu} \psi) \subset \R^2 \setminus (\Omega^a \cup B(0,R))$. On this region, $V_C= 0$, $\ha_{\mspace{1mu}\el,C} = -\Delta$, and since $V_S \leq 0$, we find $\ha_{\mspace{1mu}\el,C} = -\Delta \geq - \Delta + V_S = \ha_{\mspace{1mu}\el,S}$. This yields, by definition of $\Sigma_R^S$:
    \begin{align*}
        \ip{\tilde{j}_{R}\mspace{1mu}\psi}{\ha_{\mspace{1mu}\el,C}\mspace{1mu}\tilde{j}_{R}\mspace{1mu}\psi} \geq \ip{\tilde{j}_{R}\mspace{1mu}\psi}{\ha_{\mspace{1mu}\el,S}\mspace{1mu}\tilde{j}_{R}\mspace{1mu}\psi} \geq  \Sigma_R^S \norm{\tilde{j}_{R} \mspace{1mu}\psi}^2.
    \end{align*}
    Combining the above estimates and using that $R\mapsto\Sigma_R^S$ is nondecreasing (so that $\Sigma_R^S \geq \Sigma_{R-b}^S$), we obtain from the IMS localization formula that: 
    \begin{align*}
        \ip{\psi}{\ha_{\mspace{1mu}\el,C}\mspace{1mu}\psi}  &\geq \Sigma_{R-b}^S\underbrace{\left(\norm{\tilde{j}_{R,+}\psi}^2 + \norm{\tilde{j}_{R,-}\psi}^2 +\norm{\tilde{j}_{R}\psi}^2\right)}_{=\norm{\psi}^2=1} + \underset{R\to +\infty}{o\,(1)} = \Sigma_{R-b}^S + \underset{R\to +\infty}{o\,(1)}.
    \end{align*}
    The negligible term comes from the bounded gradients of the functions from the partition of unity (see again \cite[Definition 3.1 and Theorem 3.2]{SchroOPGlobalGeom}): since $N_R \geq 1 / \sqrt{3}$, the normalization preserves the gradient estimates up to a constant factor, and the term is indeed negligible because $u_\pm(R)$, $\sqrt{R}$ go to $+\infty$ as $R$ does. Then, taking the infimum over $\psi \in \D_R$ with $\norm{\psi}=1$ yields: 
    \begin{align*}
        \Sigma_R^C \geq \Sigma_{R-b}^S + \underset{R\to +\infty}{o\,(1)}.
    \end{align*}
    Taking $R \to +\infty$ in this inequality and using that $b>0$ is independent of $R$, one finds $\underset{R \to +\infty}{\lim}\Sigma_R^C \geq \underset{R \to +\infty}{\lim}  \Sigma_R^S$ as wanted.

    \smallbreak

    For the converse inequality, one can apply exactly the same method by adapting the previous geometric decomposition to the straight waveguide: we take $s_\pm(R)=\pm\sqrt{R^2-a^2}$ and $S_\pm(R):=s_\pm(R)\mp\sqrt R$, and, in order to transport the localized functions to the curved branches within a region where the curvilinear coordinates are one-to-one, we define $u_\pm(R)$ as before through the cut radii maps of the curved waveguide and use the corresponding normalized partition of unity. The estimates on the straight strips are obtained by exchanging the roles of the straight and curved Hamiltonians; the same fixed Euclidean transformations yield the same constant $b>0$, and hence
    \begin{align*}
        \ip{\tilde{j}_{R,\pm}\mspace{1mu}\psi}{\ha_{\mspace{1mu}\el,S}\mspace{1mu}\tilde{j}_{R,\pm}\mspace{1mu}\psi}
        = \ip{\tilde{j}_{R,\pm}\mspace{1mu}\psi}{\ha_{\mspace{1mu}\el,C}\mspace{1mu}\tilde{j}_{R,\pm}\mspace{1mu}\psi}         \geq \Sigma_{R-b}^C\norm{\tilde{j}_{R,\pm}\psi}^2.
    \end{align*}
    For the remaining term, one replaces $\Omega^a$ in the previous argument by the straight strip $\R\times[-a,a]$: outside $B(0,R)$, either $\tilde{j}_{R,+}=1$ or $\tilde{j}_{R,-}=1$ on this strip, so that
    \begin{align*}
        \supp(\tilde{j}_{R}\mspace{1mu}\psi) \subset \R^2\setminus\left(\bigl(\R\times[-a,a]\bigr)\cup B(0,R)\right).
    \end{align*}
    On this support, $V_S=0$, whereas $V_C\leq0$, and therefore
    \begin{align*}
        \ip{\tilde{j}_{R}\mspace{1mu}\psi}{\ha_{\mspace{1mu}\el,S}\mspace{1mu}\tilde{j}_{R}\mspace{1mu}\psi}
        \geq \ip{\tilde{j}_{R}\mspace{1mu}\psi}{\ha_{\mspace{1mu}\el,C}\mspace{1mu}\tilde{j}_{R}\mspace{1mu}\psi}
        \geq \Sigma_R^C\norm{\tilde{j}_{R}\mspace{1mu}\psi}^2.
    \end{align*}
    The rest of the proof is exactly as above, with $S$ and $C$ exchanged, and gives
    \begin{align*}
        \underset{R\to+\infty}{\lim}\Sigma_R^S \geq \underset{R\to+\infty}{\lim}\Sigma_R^C,
    \end{align*}
    as wanted. Having both inequalities and Persson's formula enables us to conclude that
    \begin{align*}
        \inf\sess(\ha_{\mspace{1mu}\el,S}) = \inf\sess(\ha_{\mspace{1mu}\el,C}).
    \end{align*}
    Together with the Weyl-sequence construction from the proof of \cite[Proposition 3.1]{pv_spsqv}, which gives the gapless structure of both essential spectra, this proves
    \begin{align*}
        \sess(\ha_{\mspace{1mu}\el,S}) = \sess(\ha_{\mspace{1mu}\el,C}).
    \end{align*}
\end{proof}

%%%%%%%%%%%%%%%%%%%%%%%%%%

%À FAIRE DANS CETTE SECTION

%%%%%%%%%%%%%%%%%%%%%
 
% 4) Détermination du spectre essentiel
\section{Determination of the essential spectrum}\label{sec:ess_spec}

\begin{notation}
    We denote $\Ha_S$ (resp. $\Ha_C$) the Nelson Hamiltonian corresponding to the straight waveguide (resp. the curved waveguide), i.e. with the small system Hamiltonian given by $\ha_{\mspace{1mu}\el,S}$ (resp. $\ha_{\mspace{1mu}\el,C}$) - see Table \ref{tab:hamiltonian_comparison_coordinates}. When a statement applies to both $\Ha_S$ and $\Ha_C$, we simply write $\Ha$.
\end{notation}

In this section, our goal is to determine the essential spectra of the Nelson Hamiltonians $\Ha_C$ and $\Ha_S$, in order to establish the spectral picture announced in Subsection~\ref{model:subsec_general_ideas}. More precisely, we prove that the essential spectra have a gapless structure, relate the bottom of $\sess(\Ha_C)$ to the bottom of $\sigma(\Ha_S)$, and show that $\Ha_S$ has no discrete spectrum. \\
The analysis proceeds differently in the two cases. The straight Hamiltonian is studied through its direct-integral decomposition into longitudinal-momentum fibers. For the curved waveguide, the strategy is based on the Persson-type argument introduced in the previous section. In particle-field models with decaying potentials, the limit of the localized quantities $\Sigma_R$ is usually interpreted as the ionization threshold; see for instance \cite[Theorem 6]{editnrqed}, \cite[Proposition 4.2]{analisa} and \cite[Hypothesis 2]{AC_rayleigh_scatt_1}. This is the framework from which our argument is mainly inspired: for massive Nelson Hamiltonians, one typically obtains, under suitable assumptions, an HVZ-type lower bound of the form $\inf \sess(\Ha) \geq \min(\Sigma,\inf\sigma(\Ha)+m)$; see for example \cite[Theorem 2]{AC_rayleigh_scatt_1} and \cite[Theorem 6.1]{analisa}. However, our waveguide potential is not decaying in the longitudinal direction, so these results cannot be applied directly. We therefore adapt the Persson argument to the present geometry and to the coupling with the quantized field.

\bigbreak 

We first recall the definition of the ionization threshold, following, e.g., \cite[(3)]{editnrqed}: 
\vspace{5pt}
\begin{mdframed}[innertopmargin=-2pt]
\begin{definition}[Ionization threshold]\label{sec:ess:def_ionization_threshold}
Let $\Q_R = \ens{\psi \in \Q(\Ha)}{\psi = 0  \text{ a.e. on }B(0,R)}$. Then the ionization threshold $\Sigma$ is given by: 
\begin{align*}
    \Sigma = \lim_{\rule{0pt}{1.6ex} R \to +\infty} \Sigma_R = \lim_{\rule{0pt}{1.6ex} R \to +\infty} \hspace{0.2cm} \inf_{\substack{\rule{0pt}{1.6ex} \psi \, \in \, \Q_R \\[2.5pt] \norm{\psi} \, = \, 1}} \hspace{3pt} \q_\Ha[\psi].  
\end{align*}
\end{definition}
\end{mdframed}
\vspace{7pt}
\noindent To prepare the localization argument of Proposition~\ref{sec:ess_spec:prop_equality_ionization_thresholds}, we first identify $\Q_R$ explicitly and show that the infimum defining $\Sigma$ may be restricted to vectors that are smooth and compactly supported in the particle variable.

\begin{proposition}\label{sec:ess:prop_computation_sigma_D_R}
    One has, as a subset of $L^2(\R^2,\Fs)$: 
    \begin{align*}
        \Q_R = H^1_0(\R^2 \setminus \overline{B(0,R)},\Fs) \cap L^2\hspace{-2pt}\left(\R^2\setminus \overline{B(0,R)},\D\big(\dG(\omega(K))^{1/2}\big) \right). 
    \end{align*}
    Moreover, it is enough to compute $\Sigma$ on $\D_R = \ciz\big(\R^2 \setminus \overline{B(0,R)}\big) \alten \D$, i.e.: 
    \begin{align*}
        \Sigma = \lim_{\rule{0pt}{1.6ex} R \to +\infty} \hspace{0.2cm} \inf_{\substack{\rule{0pt}{1.6ex} \psi \, \in \, \D_R \\[2.5pt] \norm{\psi} \, = \, 1}}  \ipldfs{\psi}{\Ha \psi},  
    \end{align*}
    where $\D$ is a core for $\dG(\omega(K))$.
\end{proposition}
\begin{proof}
    Let $\mathcal O_R := \R^2 \setminus \overline{B(0,R)}$. Let $\psi\in\Q_R \subset \Q(\Ha) = H^1(\R^2,\Fs)\cap L^2\hspace{-2pt}\left(\R^2,\D\big(\dG(\omega(K))^{1/2}\, \big)\right)$ (Proposition~\ref{prelim:prop:general_quadratic_form}) with $\psi = 0$ a.e. on $B(0,R)$. By the trace theorem for vector-valued functions \cite[Theorem 7.11]{trace_vectorvalued_sobolev} and the equality of interior and exterior traces for globally $H^1$ functions, the restriction $\psi|_{\mathcal O_R}$ has zero trace on $\partial\mathcal O_R$; therefore $\psi|_{\mathcal O_R}\in H^1_0(\mathcal O_R,\Fs)$. \\ Conversely, any element of $H^1_0(\mathcal O_R,\Fs)\cap L^2\hspace{-2pt}\left(\mathcal O_R,\D\big(\dG(\omega(K))^{1/2} \,\big)\right)$ extends by zero on $\overline{B(0,R)}$ to an element of $\Q(\Ha)$ vanishing on $B(0,R)$. 

    \smallbreak

    For the second claim, it suffices to prove that $\D_R$ is dense in $\Q_R$ for the norm associated to $\q_\Ha$ (since $\q_\Ha = \ipldfs{\cdot}{\mspace{1mu} \Ha \mspace{2mu \cdot}}$ on $\D_R$). One sees that $\Q_R$ is the form domain of the operator $ -\Delta^D_{\mathcal O_R}\otimes I +  I\otimes \dG(\omega(K))$, where $-\Delta^D_{\mathcal O_R}$ is the Dirichlet Laplacian on $\mathcal O_R$. Any operator core is a form core, so $\ciz(\mathcal{O}_R) \alten \D$ is dense in $\Q_R$ for the form norm of $ -\Delta^D_{\mathcal O_R}\otimes I +  I\otimes \dG(\omega(K))$, i.e. the  $H^1_0(\mathcal O_R,\Fs)\cap L^2\hspace{-2pt}\left(\mathcal O_R,\D\big(\dG(\omega(K))^{1/2} \,\big)\right)$ norm. But the latter is the $H^1(\R^2,\Fs) \cap L^2\hspace{-2pt}\left(\R^2,\D\big(\dG(\omega(K))^{1/2}\big) \right) = \Q(\Ha)$ norm restricted to the $H_0^1(\mathcal{O}_R,\Fs)$ functions, which is equivalent to the norm associated to $\q_\Ha$ by the KLMN argument used in the proof of Proposition~\ref{prelim:prop:general_quadratic_form}.
\end{proof}

%==========================================================
%       A first result using the literature
%==========================================================

\subsection{A first result using the literature}

As explained above, the existing results for decaying potentials cannot be applied directly here: we therefore use \cite[Theorem 2]{AC_rayleigh_scatt_1}. In the notation of that theorem, hypothesis (H1) is satisfied for $\omega(k)=\sqrt{k^2+m^2}$, while hypothesis (H2) requires exponential decay below the ionization threshold $\Sigma$. We first prove this decay estimate and to do so, we need the IMS localization formula, which we state in the next lemma.

\begin{lemma}[IMS localization formula]\label{ess_spec:lemma:IMS}
    Let $f$ be a smooth real-valued function with $f$, $\nabla f \in L^\infty(\R^2)$.
    One has, for any $\psi \in \Q(\Ha)$, the following double commutation property:
    \begin{align*}
         \q_\Ha[f(X)^2\psi, \psi] + \q_\Ha[\psi, f(X)^2\psi] - 2 \q_\Ha[f(X)\psi, f(X)\psi] = - 2 \ipldfs{\psi}{\abs{(\nabla f)(X)}^2 \psi }.
    \end{align*}
    Let $(j_i)_{i \in \llbracket 1, n \rrbracket}$ be a partition of unity (see \cite[\textit{Definition 3.1}]{SchroOPGlobalGeom}). Applying the previous equality to each $j_i$ and summing over $i$ yields:
    \begin{align*}
         \q_\Ha[\psi, \psi]  =   \sum_{i=1}^n \q_\Ha[j_i(X)\psi, j_i(X)\psi] - \sum_{i=1}^n \ipldfs{\psi}{\abs{(\nabla j_i)(X)}^2 \psi }.
    \end{align*}
    This is the so-called IMS localization formula. 
\end{lemma}
\begin{proof}
    Since $f(X)$ and $(\nabla f)(X)$ are bounded, multiplication by $f$ preserves the form domain $\Q(\Ha)$. The terms $V(X)$, $\Ha_f$ and $\Ha_I$ commute with $f(X)$ at the quadratic form level: $V(X)$ and $f(X)$ are multiplication operators, $\Ha_f$ acts only on the Fock variable, and $\Ha_I$ is a direct integral over the particle position. Hence only the kinetic form contributes to the commutator. The usual IMS identity for the kinetic energy gives the first formula, and summing it over a partition of unity gives the second one.
\end{proof}

\begin{proposition}[Exponential decay]\label{sec:ess_spec:prop_exp_decay}
    Assume $\Ha$ is as in the previous lemma. Let $\mathcal{J}$ be a closed interval in $]-\infty,\Sigma[$. Then for all $0 < \alpha \displaystyle < \sqrt{\Sigma-\sup\mathcal{J}}$: 
    \begin{align*}
        \nop{e^{\alpha \abs{X}} E^\Ha(\mathcal{J})} < + \infty,
    \end{align*}
    where $E^\Ha$ is the projection-valued measure of $\Ha$.
\end{proposition}
\begin{proof}
    We apply \cite[Theorem~1]{editnrqed}. In the notation of that theorem, it suffices to verify hypothesis (iii) on $\Q(\Ha)$. Indeed, hypotheses (i) and (ii) are only used to establish hypothesis (iii) after closure, whereas $\q_\Ha$ is already closed. The radiality assumption on $f$ is not needed in our setting: in \cite{editnrqed}, it is used to preserve the symmetry or anti-symmetry of the $N$-electron subspace, while our small system carries no such constraint. Hence the previous lemma gives hypothesis (iii), and the conclusion follows with every $\alpha < \sqrt{\Sigma-\sup\mathcal{J}}$.
\end{proof}

\begin{proposition}\label{sec:ess:first_inclusion_ess_spec}
    One has: $\inf \sess(\Ha) \geq \min(\Sigma,\inf\sigma(\Ha)+m)$.
\end{proposition}
\begin{proof}
    Assume first that $\Sigma>\inf\sigma(\Ha)$. By the preceding results, hypotheses (H1) and (H2) of \cite[Theorem~2]{AC_rayleigh_scatt_1} are satisfied, so the theorem yields $\inf\sess(\Ha)\geq\min(\Sigma,\inf\sigma(\Ha)+m)$. If $\Sigma=\inf\sigma(\Ha)$, the claimed inequality reduces to $\inf \sess(\Ha) \geq \inf\sigma(\Ha)$, which is true by definition.
\end{proof}

%==========================================================
%       Spectrum of the straight waveguide
%==========================================================

\subsection{Spectrum of the straight waveguide}

As explained in Subsection~\ref{model:subsec_general_ideas}, we would like to show that the spectrum of $\Ha_S$ has no gaps. A priori, there is no simple way to determine this spectrum. However, the straight waveguide is translation invariant along the longitudinal direction $\xu$, so that the corresponding component of the total momentum is conserved. This allows us to decompose $\Ha_S$ into fiber Hamiltonians, each corresponding to a fixed value $\eta$ of this momentum. Since $\eta$ ranges continuously over $\R$, we expect the spectra of these fibers to fill the possible gaps and to yield a half-line, so that the spectrum of $\Ha_S$ is entirely essential. We thus begin by proving the translation invariance of $\Ha_S$.

\begin{notation}
    Let $-i \mspace{1mu}\partial_{x_1}$ and $ \dG(K_1)$ be the infinitesimal generators of translations along $x_1$ on $L^2(\Rx^2)$ and $\Fs\big(L^2(\Rk^2)\big)$ respectively. We denote $p_1 = -i \mspace{1mu}\partial_{x_1} \otimes I$ and $P_1 = I \otimes \dG(K_1)$ the operators corresponding to the momentum along $\xu$ for the small system and the field respectively. 
\end{notation}

\begin{lemma}[Translation invariance along $x_1$ of $\Ha_S$]\label{sec:ess_spec:lemma_inv_translation}  The operator $P^{\,\mathrm{tot}}_1 := \overline{p_1 + P_1}$ is the infinitesimal generator of translations on  $L^2(\Rx^2)\otimes\Fs\big(L^2(\Rk^2)\big)$ and one has:
\begin{align*}
    \forall \xi \in \R, \; e^{ - \mspace{1mu }i \mspace{2mu}P^{\,\mathrm{tot}}_1 \mspace{2mu} \xi } \, \Ha_S \,  e^{\mspace{2mu} i \mspace{2mu}P^{\,\mathrm{tot}}_1 \mspace{2mu} \xi }  = \Ha_S,
\end{align*}
where for all $\xi \in \R$, denoting $\tau_1^\xi = e^{\mspace{2mu} - i \mspace{2mu} \xi \mspace{2mu} p_1}$, one has $e^{ - \mspace{1mu }i \mspace{2mu}P^{\,\mathrm{tot}}_1 \mspace{2mu} \xi } = \tau_1^\xi \otimes \Gamma\big(e^{-iK_1\xi}\big)$. \\ 
In particular, $P^{\,\mathrm{tot}}_1$ and $\Ha_S$ strongly commute.
\end{lemma}

\begin{proof}
    If the claimed equality is proven, the strong commutativity comes from \cite[Proposition 1.48]{Arai}. \cite[Theorem 3.8(v), Theorem 4.15(i)]{Arai} give that for all $\xi \in \R$, $\tau_1^\xi \otimes \Gamma\big(e^{-iK_1\xi}\big) = e^{-\mspace{2mu} i \mspace{2mu}P^{\,\mathrm{tot}}_1 \mspace{2mu} \xi }$. The left-hand side corresponds to translations on $L^2(\Rx^2)\otimes\Fs\big(L^2(\Rk^2)\big)$, so $P^{\,\mathrm{tot}}_1$ is indeed their infinitesimal generator. Let us then prove the equality. \\
    The particle Hamiltonian $-\Delta+V_S(X)$ is invariant under longitudinal translations because $V_S$ depends only on $\xd$, while the invariance of $\dG\bigl(\omega(K)\bigr)$ follows from \cite[Theorem 4.16]{Arai}. It therefore remains to consider $\Ha_I$. One gets, using \cite[Theorem 5.32]{Arai}, that $\Gamma\big(e^{-iK_1\xi}\big) \phi(v_x) \Gamma\big(e^{iK_1\xi}\big) = \phi(e^{-i K_1 \xi} \, v_x) = \phi\big(v_{x+(\xi,0)}\big)$:
    \begin{align*}
            e^{ - \mspace{1mu }i \mspace{2mu}P^{\,\mathrm{tot}}_1 \mspace{2mu} \xi } \, \Ha_I \,  e^{\mspace{2mu} i \mspace{2mu}P^{\,\mathrm{tot}}_1 \mspace{2mu} \xi } = \left(I \otimes \Gamma\big(e^{-iK_1\xi}\big) \right) \left(\intd{2}{ \phi\left(v_{x-(\xi,0)}\right)} \right) \left(I \otimes \Gamma\big(e^{iK_1\xi}\big) \right) = \Ha_I.
    \end{align*}
\end{proof}

By Proposition~\ref{app_op:simultaneous_diagonalize_direct_integral}, the strong commutativity established above reduces the direct integral decomposition of $\Ha_S$ to the diagonalization of $P^{\mathrm{tot}}_1$. Hence we now construct an explicit unitary operator diagonalizing $P^{\mathrm{tot}}_1$, using the transformation introduced in \cite{LeePines1953}, sometimes called the Lee-Low-Pines transformation.

\begin{notation}
    We denote $L^2(\R_{\xd}, \Fs) := L^2(\R,d\xd,\Fs)$.
\end{notation}

\begin{lemma}[Diagonalization of $P^{\,\mathrm{tot}}_1$]\label{sec:ess_spec:lemma_diagonalization_ptot}
    The unitary operator $\F_{\xu} \int_{\R^2}^\oplus \Gamma\big(e^{i K_1 \xu}\big) d\xu d\xd$, where $\F_{\xu}$ is the partial Fourier transform along $\xu$ for vector-valued functions, diagonalizes $P^{\,\mathrm{tot}}_1$ on the Hilbert space $\int_{\R}^{\oplus} L^2(\R_{\xd}, \Fs) \, d\eta$, where $\eta \in \R$ is the value of the total momentum of the system along $\xu$. One has:
    \begin{align*}
        \left(\F_{\xu} \int_{\R^2}^\oplus \Gamma\big(e^{i K_1 \xu}\big) d\xu d\xd \right)  P^{\,\mathrm{tot}}_1   \left(\F_{\xu} \int_{\R^2}^\oplus \Gamma\big(e^{i K_1 \xu}\big) d\xu d\xd \right)^\ast     = \int_{\R}^{\oplus} \eta \, d\eta \hspace{0.5cm} \mathrm{on } \hspace{0.2cm} \int_{\R}^{\oplus} L^2(\R_{\xd}, \Fs) \, d\eta.
    \end{align*}
\end{lemma}
\begin{proof}
    For the definition of the Fourier transform for vector-valued functions, see, e.g., \cite[Section \cRM{3}.4.2]{QuasiLinearParabolicPb}: we denote it similarly in the scalar-valued and the vector-valued case. \\
    Let $W := \int_{\R^2}^\oplus \Gamma\big(e^{i K_1 \xu}\big) d\xu d\xd$ and let us begin by computing $W P^{\,\mathrm{tot}}_1 W^\ast $. We observe that: 
    \begin{align*}
        W \left( \D(p_1) \cap \D(P_1)\right) = \D(p_1) \cap \D(P_1). 
    \end{align*}
    Indeed, $W \D(p_1) \subset \D(p_1)$ follows from $\Gamma\big(e^{i K_1 \xu}\big) = e^{ i \mspace{2mu} \dG(K_1)\mspace{1mu} \xu} $ and Proposition~\ref{app:func:derivation_unitary_direct_int}, while $W \D(P_1) \subset \D(P_1)$ follows by property of infinitesimal generators; this gives the $\subset$ inclusion. The $\supset $ inclusion follows similarly, by proving $W^\ast \hspace{-2pt}\left( \D(p_1) \cap \D(P_1)\right) \subset \D(p_1) \cap \D(P_1) $. So let us work on $\D(p_1) \cap \D(P_1)$, which is a core for  $P^{\,\mathrm{tot}}_1$.   Combining Proposition \ref{app:func:derivation_unitary_direct_int} and \cite[Theorem 4.16]{Arai}, we get $W p_1 W^\ast = (p_1-P_1)$ and $W P_1 W^\ast = P_1$ on $\D(p_1) \cap \D(P_1)$, so that:
    \begin{align*}
        W P^{\,\mathrm{tot}}_1 W^\ast = p_1 \hspace{0.5cm} \text{on } \D(p_1) \cap \D(P_1). 
    \end{align*}
    Since $\D(p_1)\cap\D(P_1)$ is a core for $P^{\,\mathrm{tot}}_1$, the domain invariance above implies that it is a core for $W P^{\,\mathrm{tot}}_1 W^\ast$. It is also a core for $p_1$, since it contains the standard core $\D(p_1)\alten\D(P_1)$ of $p_1$. Thus the two closed operators $WP^{\,\mathrm{tot}}_1W^\ast$ and $p_1$ agree on a common core, hence they are equal. \\
    As in the scalar case, the Fourier transform diagonalizes the momentum operator, which becomes multiplication by $\eta$ on $L^2(\R^2, d\eta \, d\xd; \Fs) = \int_{\R}^{\oplus} L^2(\R_{\xd}, \Fs) \, d\eta$, where $\eta$ is the Fourier variable.
\end{proof}

\noindent Consequently, we call the unitary operator constructed above the fibering operator:

\begin{definition}[Fibering operator]
    The \textit{fibering operator} is $U_\fib := \F_{\xu} \int_{\R^2}^\oplus \Gamma\big(e^{i K_1 \xu}\big) d\xu d\xd$, which is unitary.
\end{definition}

\noindent Since the direct integral representation will be used repeatedly below, we state it as a theorem.

\begin{notation}
    We denote $\ha_{\el,0} := -\partial_{x_2}^2 + V(X_2)$, $\Ha_{\el,0} := \ha_{\el,0}\otimes I $, $\Ha_{I,0} :=  \int_{\R}^\oplus \phi\left(v_{(0,x_2)}\right) d\xd$ and $\dG_\eta :=  \left( \dG(\omega(K)) + \left(\eta - \dG(K_1)\right)^2 \right)$.
\end{notation} 

\begin{theorem}\label{sec:ess:prop_fiber_straight_hamiltonian}
    One has:
    \begin{align*}
        U_\fib \, \Ha_S \, U_\fib^\ast = \int_{\R}^\oplus  \Ha_S(\eta) \, d\eta \hspace{0.5cm} \mathrm{on} \hspace{5pt} \int_{\R}^{\oplus} L^2(\R_{\xd}, \Fs) \, d\eta.
    \end{align*}
    For all $\eta \in \R$, $\Ha_S(\eta)$ is a self-adjoint operator on $L^2(\R_{\xd}, \Fs) = \int_{\R}^\oplus \Fs \, d\xd $ given by: 
    \begin{align*}
        \Ha_S(\eta) =& \hspace{5pt} (\eta - P_1)^2 + (-\partial_{x_2}^2 + V(X_2))\otimes I + \Ha_f  +  g\int_{\R}^\oplus \phi\left(v_{(0,x_2)}\right) d\xd \\
        =& \hspace{5pt} \Ha_{\el,0} + I \otimes \dG_\eta + g\mspace{1mu}\Ha_{I,0}.
    \end{align*}
    Physically, $\eta$ represents the total momentum of the system along $\xu$. Its domain is:
    \begin{align*}
        \D(\Ha_S(\eta)) = \D(\Ha_{\el,0}) \cap \D(\Ha_f) \cap \D \hspace{-2pt} \left(  \left( \eta - P_1\right)^2 \right) = \D(\Ha_{\el,0}) \cap \D \hspace{-1pt}\left(I \otimes \dG_\eta\right). 
    \end{align*}
\end{theorem}
\begin{proof}
    Thanks to the two previous lemmas, we can apply Proposition~\ref{app_op:simultaneous_diagonalize_direct_integral} with $U = U_\fib$ and $\Hi = L^2(\R_{\xd}, \Fs)$ to get the existence of such $\Ha_S(\eta)$. It remains to compute it. \\
    Thanks to Proposition~\ref{sec:prelim:domain_splitting_and_cores}(i), for the straight waveguide one has $\Ha_\el=-\partial_{\xu}^2\otimes I+\Ha_{\el,0}$. By Proposition~\ref{app:op:standard_properties}(iii), $-\partial_{\xu}^2\otimes I=p_1^2$. Using the previous lemma, we get:
    \begin{align*}
        U_\fib \, p_1^2 \, U_\fib^\ast = (U_\fib \, p_1 \, U_\fib^\ast)^2 = \F_{\xu}(p_1-P_1)^2 \F_{\xu}^\ast =  \int_{\R}^{\oplus}(\eta-P_1)^2 d\eta. 
    \end{align*}
    Moreover, $U_\fib$ leaves $\Ha_{\el,0}$ unchanged, since this operator only acts on the transverse variable $\xd$, and it leaves $\Ha_f$ unchanged by \cite[Theorem 4.16]{Arai}. Thus it remains to compute the interaction. As in the proof of Lemma~\ref{sec:ess_spec:lemma_inv_translation}, for every $\xi\in\R$, $\Gamma(e^{iK_1\xi})\phi(v_x)\Gamma(e^{-iK_1\xi}) =
    \phi(v_{x-(\xi,0)})$ so taking $\xi=\xu$ gives $\phi\big(v_{(0,\xd)}\big)$. Hence, putting the four contributions together, we obtain the announced expression for $\Ha_S(\eta)$:
    \begin{align*}
        U_\fib \mspace{2mu} \Ha_S \mspace{2mu} U_\fib^\ast = \int_{\R}^{\oplus} \left( (\eta-P_1)^2 + \Ha_{\el,0} + \Ha_f + g\Ha_{I,0} \right) d\eta.
    \end{align*}
    It remains to justify the domain statement. By Proposition~\ref{app:op:standard_properties}(iii), $(\eta-P_1)^2=I\otimes(\eta-\dG(K_1))^2$. Moreover, by \cite[Proposition 5.5, Theorem 1.36(ii)]{Arai} and Propositions~\ref{app:op:standard_properties}(v),(iii), the positive operators $\Ha_f=I\otimes\dG(\omega(K))$ and $(\eta-P_1)^2$ strongly commute. Hence:
    \begin{align*}
        \Ha_f+(\eta-P_1)^2 = I\otimes \left(\dG(\omega(K))+(\eta-\dG(K_1))^2\right) := I\otimes\dG_\eta.
    \end{align*}
    Thus the proofs of Lemma~\ref{sec:prelim:hfree_self_adjoint} and Theorem~\ref{prelim:th:self-adjointness_Ha} apply verbatim, with $\Ha_{\el,0}$ instead of $\Ha_\el$, $I\otimes\dG_\eta$ instead of $\Ha_f$ and $\Ha_{I,0}$ instead of $\Ha_I$. This gives the announced self-adjointness and domain of $\Ha_S(\eta)$.
\end{proof}

\begin{remark}
The expression of $\Ha_S(\eta)$ reflects the fixing of the total momentum along $\xu$: formally, $\eta=p_1+P_1$, hence $p_1^2=(\eta-P_1)^2$. Since the straight waveguide is translation invariant along $\xu$, the interaction only depends on the transverse variable (hence index $x=(0,\xd)$ in the field operator), although the field still lives in the two-dimensional momentum space. Finally, by abuse of notation, we still write $\Ha_f$ both in the full Hamiltonian, where the identity acts on $L^2(\Rx^2)$, and in the fiber Hamiltonian, where it acts on $L^2(\R_{\xd})$.
\end{remark}

Having written $\Ha_S$ as a direct integral, we can get some information about its spectrum thanks to Proposition~\ref{app:op:prop_spectrum_direct_integrals}:

\begin{proposition}
    One has:
    \begin{enumerate}[label=(\roman*)]
        \item The domain of $\Ha_S(\eta)$ is independent of $\eta$: $\D(\Ha_S(\eta)) = \D(\Ha_S(0))$. Moreover, the map $\R \owns \eta \mapsto \Ha_S(\eta)$ is strongly continuous. 
        \item The spectrum of $\Ha_S$ is given by: 
        \begin{align*}
        \sigma(\Ha_S) = \overline{\bigcup_{\eta \in \R} \sigma\hspace{-1pt}\left(\Ha_S(\eta)\right)}.
    \end{align*}
    \end{enumerate}
\end{proposition}
\begin{proof}
    We begin by proving $(i)$. One sees that $\D \hspace{-2pt} \left(  \left( \eta - P_1 \right)^2 \right) = \D\hspace{-2pt}\left(P_1^2\right)$ so the previous theorem yields $\D\big(\Ha_S(\eta)\big) = \D\big(\Ha_S(0)\big)$ and the strong continuity can be checked on this common domain. Let $\psi \in \D(\Ha_S(0))$, $\eta$, $\eta_0 \in \R$. We compute: 
    \begin{align*}
        \nldfs{\left(\Ha_S(\eta)-\Ha_S(\eta_0)\right)\psi} \leq \abs{\eta^2-\eta_0^2} \nldfs{\psi}+ 2 \abs{\eta_0 - \eta} \nldfs{P_1 \mspace{2mu}\psi}. 
    \end{align*}
    And one sees that $\D(\Ha_S(0)) \subset   \D \hspace{-2pt} \left(P_1^2 \right) \subset \D(P_1)$, so the norm is finite and taking the limit $\eta \to \eta_0$ gives the result.  \\
    Now we prove $(ii)$. Since $U_\fib$ is unitary, $\sigma\left(U_\fib \Ha_S U_\fib^\ast\right) = \sigma(\Ha_S)$. Combining Theorem \ref{sec:ess:prop_fiber_straight_hamiltonian}, what we proved above in $(i)$ and Proposition \ref{app:op:prop_spectrum_direct_integrals} yields the result.
\end{proof}

With this property, we expect that there is no gap in the spectrum of $\Ha_S$. Indeed, a large value of the total momentum along $\xu$ should entail a large energy, while its continuous variation should fill the possible gaps. We now make this precise.

\begin{proposition}\label{sec:ess:prop_no_gap_straight_spec}
    There is no gap in $\sigma(\Ha_S)$, i.e. $\sigma(\Ha_S) = \left[\underset{\eta\in\R}{\inf} \inf \sigma\hspace{-1pt}\left(\Ha_S(\eta)\right),+\infty \right[$.
\end{proposition}
\begin{proof}
    Thanks to the previous proposition, we see that since $\sigma(\Ha_S)$ is lower semi-bounded, so is each $\Ha_S(\eta)$. Consequently, the function $E : \R \owns \eta\mapsto \inf\sigma\left(\Ha_S(\eta)\right) \in \R$ is well defined. One sees that if $E$ varies continuously with $\eta \in \R$ and goes to $+\infty$ when $\eta \to \pm \infty$, it will yield:
    \begin{equation}\label{sec:ess:eq_prop_no_gap_straight_spec}
    \left[\underset{\eta \in \R}{\inf}E(\eta), + \infty \right[ =  \overline{\bigcup_{\eta \in \R} \sigma\hspace{-1pt}\left(\Ha_S(\eta)\right)}.
    \end{equation}
    So let us prove it. First we prove that $\underset{\eta \to \pm \infty}{\lim}E(\eta) = + \infty$. Since $\omega(K) \geq \abs{K_1} \geq 0$, one gets $\dG(\omega(K)) \geq \dG\left(\abs{K_1}\right) \geq 0$. Combining \cite[Proposition 3.4(ii)]{lemma_power_dG_gerar_moller} (with $b = \abs{K_1}$, $a=K_1$ in the notation of that proposition) with \cite[Proposition 10.14]{Sch} (with $\alpha =1/2$ in the notation of that proposition), we get $\dG\big(\abs{K_1}\big) \geq \abs{\dG(K_1)}$ and we compute: 
     \begin{align*}
         \left(\eta-\dG(K_1)\right)^2 + \dG(\omega(K)) \geq \left(\eta-\dG(K_1)\right)^2 + \dG \big(\abs{K_1}\big) \geq \left(\eta-\dG(K_1)\right)^2 + \abs{\dG\left(K_1\right)}. 
     \end{align*}
     By the functional calculus and since $\lambda \mapsto (\eta-\lambda)^2+\abs{\lambda}$ is continuous and lower bounded, we have, using $\sigma(\dG(K_1)) = \R$:
     \begin{align*}
         \inf \sigma\hspace{-2pt}\left(\left(\eta-\dG(K_1)\right)^2 + \abs{\dG(K_1)}\right) = \underset{\lambda\in\R}{\inf}(\eta-\lambda)^2+\abs{\lambda} = 
         \left\{ \begin{aligned}
             \eta-1/4 \hspace{0.5cm}  &\text{ if } \eta \geq 1/2\\
             -\eta-1/4 \hspace{0.5cm} &\text{ if } \eta \leq -1/2 \\
             \eta^2 \hspace{0.5cm} &\text{ otherwise}
            \end{aligned} \right. .
     \end{align*}
     Combining this with \cite[Theorem 3.8(iv)]{Arai}, we find:
     \begin{align*}
        f(\eta) := \inf \sigma(\Ha_{\el,0} + I \otimes \dG_\eta) \geq \inf \sigma(\Ha_{\el,0})  + \underset{\lambda\in\R}{\inf}(\eta-\lambda)^2+\abs{\lambda} \underset{\eta \to \pm \infty}{\longrightarrow} + \infty.
     \end{align*}
     It remains to add $\Ha_{I,0}$. We know it is infinitesimally $\Ha_{\el,0} + I \otimes \dG_\eta$-bounded (see the proof of Theorem \ref{sec:ess:prop_fiber_straight_hamiltonian}), so let $\varepsilon \in \mspace{3mu}]0,1[$, $b(\varepsilon) > 0$ be relative bounds. The Kato-Rellich theorem tells us (see, e.g., \cite[Theorem \cRM{10}.12]{RS}) that: 
     \begin{align*}
         E(\eta)  \geq f(\eta) - \max\left( \frac{b(\varepsilon)}{1-\varepsilon}, \varepsilon \mspace{1mu}f(\eta) + b(\varepsilon)\right). 
     \end{align*}
     For $\abs{\eta}$ large enough, so is $f(\eta)$, and one gets $E(\eta) \geq f(\eta)(1-\varepsilon)-b(\varepsilon) \underset{\eta \to \pm \infty}{\longrightarrow} + \infty$ since $\varepsilon \in \mspace{3mu}]0,1[$. \\
     Now we prove that $E$ is continuous. Combining \cite[Lemma 1.6 and Corollary 1.9]{Arai}, it suffices to show that $\Ha_S(\eta)$ converges in norm resolvent sense to $\Ha_S(\eta_0)$ for all $\eta_0 \in \R$ as $\eta \to \eta_0$. Since, for all $\eta\in\R$, $\D(\Ha_S(\eta))  \subset \D(P_1)$ (see the previous proposition), \cite[Lemma 1.9]{Arai} gives us that for all $z \in \C\setminus\R$ and all $\eta \in \R$, $P_1 \left(\Ha_S(\eta)-z\right)^{-1}$ is bounded. Then the second resolvent identity yields:
    \begin{align*}
    \nop{\left(\Ha_S(\eta_0)-z\right)^{-1} -  \left(\Ha_S(\eta)-z\right)^{-1}} 
    &= \nop{\left(\Ha_S(\eta)-z\right)^{-1}\left(\Ha_S(\eta_0)-\Ha_S(\eta)\right) \left(\Ha_S(\eta_0)-z\right)^{-1}} \\
    &\leq \abs{\eta_0^2 - \eta^2}\nop{\left(\Ha_S(\eta)-z\right)^{-1}\left(\Ha_S(\eta_0)-z\right)^{-1}}  \\
    &\phantom{\leq} + 2\abs{\eta_0-\eta}\nop{\left(\Ha_S(\eta)-z\right)^{-1}P_1 \left(\Ha_S(\eta_0)-z\right)^{-1}}.
\end{align*}
All operator norms are bounded so taking $\eta \to \eta_0$ yields the result. 
\end{proof}

This result is important: the spectrum of the Hamiltonian associated with the straight waveguide is a closed half-line and is therefore purely essential. Thus, despite the coupling with the quantum field, we recover the spectral picture known for straight waveguides without a field; see Subsection~\ref{model:subsec_general_ideas} and Figure~\ref{fig:straight_waveguide_and_spectrum}. To completely determine the spectrum, it remains to identify its infimum. We expect it to coincide with the ionization threshold, as in the uncoupled case where Persson's formula guarantees that it is true for the straight waveguide (see Subsection \ref{prelim:subsec_persson}).

\begin{notation}
    We denote $\Sigma_S$ (resp. $\Sigma_C$) the ionization threshold (see Definition~\ref{sec:ess:def_ionization_threshold}) for $\Ha_S$ (resp. $\Ha_C$).
\end{notation}

\begin{proposition}\label{sec:ess_spec:inf_straight_equal_sigma}
    One has $\inf \sigma(\Ha_S) = \Sigma_S$.
\end{proposition}

\begin{proof}
    Using the translation invariance of $\Ha_S$ (see Lemma~\ref{sec:ess_spec:lemma_inv_translation}, from which we take the notations) and denoting $U_\xi := e^{-\mspace{2mu} i \mspace{2mu}P^{\,\mathrm{tot}}_1 \mspace{2mu} \xi }$, we can write, for all $\xi \in \R$:
    \begin{equation}\label{sigma_droit_inv_translation}
        \Sigma_R^S := \underset{\substack{\rule{0pt}{7pt} \psi \in \D(\Ha_S) \\[1pt] \psi = 0 \text{ on } B(0,R) \\[1pt] \|\psi\| = 1}}{\inf} \ipldfs{U_\xi \mspace{2mu} \psi}{\Ha_S  \mspace{1mu} U_\xi \mspace{2mu} \psi} 
        = \underset{\substack{\rule{0pt}{7pt} \psi \in \D(\Ha_S) \\[1pt] \psi = 0 \text{ on } B((\xi,0),R) \\[1pt] \|\psi\| = 1}}{\inf}\ipldfs{\psi}{\Ha_S \mspace{2mu} \psi}.
    \end{equation}
    Consequently, translation invariance implies that the computation of $\Sigma_R^S$ is invariant under longitudinal translations of the excluded ball. Taking the limit as this translation goes to infinity essentially removes the spatial constraint, yielding exactly $\inf \sigma(\Ha_S)$. So let $\chi_R \in \ci(\R^2,[0,1])$ such that $\chi_R = 0$ on $B(0,R)$ and $\chi_R = 1$ on $\R^2 \setminus B(0,2R)$, and let $\chi_R^\xi = \tau^\xi_1 \chi_R$, which is $0$ on $B((\xi,0),R)$. \\
    In order to compute the infimum, we build a minimizing sequence $(\psi_n)_n$ for $\inf\sigma(\Ha_S)$, i.e. $\psi_n\in\D(\Ha_S)$, $\nldfs{\psi_n}=1$, and
    $\ipldfs{\psi_n}{\Ha_S \mspace{2mu}\psi_n}\to\inf\sigma(\Ha_S)$. As in the proof of the IMS localization formula, multiplication by $\chi_R^\xi$ preserves $\D(\Ha_S)$ so that $\chi_R^\xi \psi_n \in \D(\Ha_S)$ for all $n\in\N$. We find that, for every fixed $n\in\N$:
    \begin{align*}
        \chi_R^\xi\psi_n \underset{\xi\to+\infty}{\longrightarrow} \psi_n \quad \text{in the graph norm of } \Ha_S .
    \end{align*}
    Indeed, $\chi_R^\xi\to 1$ pointwise as $\xi \to +\infty$, and its derivatives are uniformly bounded and supported in the annulus $B((\xi,0),2R)\setminus B((\xi,0),R)$ (which escapes to infinity), hence they converge to $0$ pointwise as $\xi \to +\infty$. Thus, the dominated convergence theorem yields the strong convergence of the kinetic part. The free-field term converges because the bounded multiplier $\chi_R^\xi$ acts solely on the particle variable, while the interaction term converges by the infinitesimal relative boundedness of $\Ha_I$ with respect to $\Ha_\free$. Together, this yields the claimed convergence. \\
    Now, consider the normalized sequence  $\Psi^\xi_n = \displaystyle \frac{\chi_R^\xi \psi_n}{\norm{\chi_R^\xi \psi_n}}$ for all $n\in \N$ and all $\xi$ large enough so that $\norm{\chi_R^\xi \psi_n} \neq 0$. One computes:
    \begin{align*}
        \ipldfs{\Psi^\xi_n}{\Ha_S \mspace{2mu}\Psi^\xi_n} = \frac{1}{\norm{\chi_R^\xi \mspace{2mu}\psi_n}^2}\ipldfs{\chi_R^\xi \mspace{2mu}\psi_n}{\Ha_S \mspace{2mu} \chi_R^\xi \mspace{2mu} \psi_n}  \underset{\xi \to +\infty}{\longrightarrow} \ipldfs{\psi_n}{\Ha_S\mspace{2mu}\psi_n} \underset{n \to +\infty}{\longrightarrow} \inf \sigma(\Ha_S).
    \end{align*}
    But $\Psi^\xi_n \in \D(\Ha_S)$, $\nldfs{\Psi^\xi_n} = 1$ and $\Psi^\xi_n = 0$ on $B((\xi,0),R)$ for all $n \in \N$ and $\xi$ large enough; hence, using \eqref{sigma_droit_inv_translation} and by definition of $\Sigma_R^S$:
    \begin{align*}
         \forall \xi \in \R,\, n\in \N, \: \inf \sigma(\Ha_S) \leq \Sigma_R^S = \underset{\substack{\rule{0pt}{8pt} \psi \in \D(\Ha_S) \\[1pt] \psi = 0 \text{ on } B((\xi,0),R) \\[1pt] \|\psi\| = 1}}{\inf}\ipldfs{\psi}{\Ha_S \mspace{2mu} \psi} \leq \ipldfs{\Psi^\xi_n \mspace{2mu} }{\Ha_S \mspace{2mu} \Psi^\xi_n}. 
    \end{align*}
    Since it is true for all $\xi \in \R$, $n \in \N$, we can take the limit in $\xi$ then in $n$ in the right-hand side to get $\inf \sigma(\Ha_S) \leq \Sigma_R^S \leq \inf \sigma(\Ha_S)$ and then: 
    \begin{align*}
        \forall R > 0, \: \Sigma_R^S = \inf \sigma(\Ha_S) \Longrightarrow \Sigma_S = \underset{R \to +\infty}{\lim} \Sigma_R^S = \inf\sigma(\Ha_S).
    \end{align*} 
\end{proof}

We hence have the following first important result of this work: 
\begin{theorem}\label{sec:ess_spec:theorem_spectrum_straight_waveguide}
    Let $\Sigma_S$ be the ionization threshold of $\Ha_S$ and let $(\Ha_S(\eta))_{\eta \in \R}$ be its fibers (Theorem~\ref{sec:ess:prop_fiber_straight_hamiltonian}). The spectrum of $\Ha_S$ is the half-line given by:
    \begin{align*}
        \sigma(\Ha_S) = [\Sigma_S,+\infty[ = \left[\underset{\eta\in\R}{\inf} \inf \sigma\hspace{-1pt}\left(\Ha_S(\eta)\right),+\infty\right[.
    \end{align*}
    In particular, $\Ha_S$ only has essential spectrum: $\sess(\Ha_S) = \sigma(\Ha_S)$.
\end{theorem}
\begin{proof}
    One combines the two previous propositions. There is no discrete spectrum because there are no isolated points in the spectrum.
\end{proof}

%==========================================================
%       Essential spectrum of the curved waveguide
%==========================================================

\subsection{Essential spectrum of the curved waveguide}

As discussed in Subsection~\ref{prelim:subsec_persson}, we expect the ionization thresholds $\Sigma_S$ and $\Sigma_C$ to coincide because of Assumption~\ref{h2_geom_compact_curvature}. This equality will provide a first link between the spectral properties of $\Ha_S$ and $\Ha_C$. \\ 
Unlike in the uncoupled case, however, it does not by itself determine the bottom of the essential spectrum of $\Ha_C$. Indeed, Proposition~\ref{sec:ess:first_inclusion_ess_spec} only yields $\inf\sess(\Ha_C)\geq\min(\Sigma_C,\inf\sigma(\Ha_C)+m)$, where the second threshold reflects the minimal energy required to create one boson (the mass $m>0$).

\smallbreak
 
To prove the equality of the ionization thresholds, we adapt the localization argument of Subsection~\ref{prelim:subsec_persson}. The new point is that the field must also be aligned with each straight branch of the curved waveguide, which requires Hypothesis~\ref{h3_field_rotation_invariance}.

\begin{proposition}\label{sec:ess_spec:prop_equality_ionization_thresholds}
    Assume \textup{\ref{h3_field_rotation_invariance}}. Then the ionization threshold for $\Ha_S$ and $\Ha_C$ is the same, namely $\Sigma_C = \Sigma_S$.
\end{proposition}
\begin{proof}
    We prove $\Sigma_C \geq \Sigma_S$ and $\Sigma_S \geq \Sigma_C$. To retain the structure of the proof of Proposition~\ref{sec:prelim:prop_persson_small_system} and carry out the computations at the operator level on functions that are smooth and compactly supported in the particle variable, we compute $\Sigma$ on $\D_R = \ciz\big(\R^2 \setminus B(0,R)\big) \alten \D\bigl(\dG(\omega(K))\bigr)$ (see Proposition~\ref{sec:ess:prop_computation_sigma_D_R}).
    \begin{itemize}
        \vspace{0.3cm}
        \item $\Sigma_C \geq \Sigma_S$: the proof is the same as the one of Proposition \ref{sec:prelim:prop_persson_small_system}, up to a few changes. \\
        We take the same partition of unity and apply again the IMS localization formula (which holds for $\Ha_C$ as proved in Lemma \ref{ess_spec:lemma:IMS}). In the region where $V_C = 0$, we find again $\Ha_C \geq \Ha_S$ since $V_S \leq 0$, so this argument is not changed. The negligible term is also the same, so this is not changed either. It remains to check the claim on the two straight branches, and it suffices to look at one of them since both cases are similar, so let us consider the strip $\Omega_{R,+}$. The main difference with the proof of Proposition \ref{sec:prelim:prop_persson_small_system} is that even in $(s,u)$ coordinates, we do not have $\Ha_S = \Ha_C$ outside the support of the curvature, because the interaction term depends on $\varphi(s,u)$ (see Subsection~\ref{sec:mode:subsec_the_model}), which is not equal to $(s,u)$ if the straight parts of the curved waveguide are not aligned with the straight waveguide.
        \\
        However, since we are outside of the support of the curvature, the waveguide is straight and $\varphi$ is a Euclidean transformation (rotation and translation), namely:
        \begin{align*}
            \varphi(s,u) = \mathcal{R}\mspace{1mu}(s,u) + (s_0,u_0) \iff (s,u) = \mathcal{R}^{-1}\left(\varphi(s,u) -(s_0,u_0)\right),
        \end{align*}
        where $\mathcal{R}$ is a rotation matrix in $\R^2$. Applying then, only on the Fock space, a rotation $\mathcal{R}^{-1}$ and a translation by $-(s_0,u_0)$ through the adjoint of the operator 
        \begin{align*}
           W_{\mspace{2mu} \mathcal{R}, \mspace{2mu} (s_0,u_0)} = I \otimes \left( \Gamma\mspace{-3mu}\left(e^{- i \mspace{1mu} s_0 \mspace{2mu} K_1}\right) \Gamma\mspace{-3mu}\left(e^{- i \mspace{1mu} u_0 \mspace{2mu} K_2}\right) \Gamma\left(\mathcal{R}\right) \right) 
        \end{align*}
        (where, by abuse of notation, $\mathcal{R}$ still denotes the rotation operator in $\ldd$, i.e. $\mathcal{R} \mspace{2mu} \psi = \psi(\mathcal{R}^{-1}\cdot)$), we get back $(s,u)$ rather than $\varphi(s,u)$ in the interaction Hamiltonian, which, upon passing to the $(s,u)$ coordinates, gives the following identity:
        \begin{align*}
            \ip{\tilde{j}_{R,+}\mspace{1mu}\psi}{\Ha_C\mspace{1mu}\tilde{j}_{R,+}\mspace{1mu}\psi} &=  \ip{W_{\mspace{2mu} \mathcal{R}, \mspace{2mu} (s_0,u_0)}^\ast \mspace{1mu} \tilde{j}_{R,+}\mspace{1mu}\psi \mspace{2mu}}{\mspace{2mu}\Ha_S \mspace{1mu} W_{\mspace{2mu} \mathcal{R}, \mspace{2mu} (s_0,u_0)}^\ast \mspace{1mu}\tilde{j}_{R,+}\mspace{1mu}\psi}.
        \end{align*}
        Indeed, since $\psi\in\D_R$ and $\tilde{j}_{R,+}$ is smooth, one has $\tilde{j}_{R,+}\mspace{1mu}\psi\in\D_R\subset\D(\Ha_C)$ in the original Cartesian coordinates. To see that $W_{\mathcal R,(s_0,u_0)}^\ast \tilde{j}_{R,+}\mspace{1mu}\psi \in \D(\Ha_S)$, we pass to the $(s,u)$ coordinates, keeping the same notation for the transformed vector. As in the proof of Proposition~\ref{sec:prelim:prop_persson_small_system}, in $(s,u)$ coordinates, for some fixed $b>0$ independent of $R$, one has $\tilde{j}_{R,+} \mspace{1mu} \psi \in \D_{R-b}$. Then, since $W_{\mathcal R,(s_0,u_0)}^\ast$ preserves $\D\bigl(\dG(\omega(K))\bigr)$ (\cite[Theorem 4.16]{Arai} combined with the fact that $\omega(K)$ commutes with translations and rotations) and since it only acts on the Fock variable, it does not change the regularity or the support in the particle variable, and therefore $W_{\mathcal R,(s_0,u_0)}^\ast \tilde{j}_{R,+}\mspace{1mu}\psi \in \D_{R-b}\subset\D(\Ha_S)$. Afterwards, it suffices to compute $W_{\mspace{2mu} \mathcal{R}, \mspace{2mu} (s_0,u_0)} \mspace{1mu} \Ha_S \mspace{1mu} W_{\mspace{2mu} \mathcal{R}, \mspace{2mu} (s_0,u_0)}^\ast$ to get the claimed equality on $\supp(\tilde{j}_{R,+}\mspace{1mu}\psi)$. For the translation part, the computation closely follows that in Lemma~\ref{sec:ess_spec:lemma_inv_translation}; the rotation case follows analogously, since rotations in configuration space correspond to rotations in momentum space and since the form factor $v$ is rotation-invariant by Assumption~\ref{h3_field_rotation_invariance}.
        Consequently, after aligning the field with the straight waveguide through $W_{\mspace{2mu} \mathcal{R}, \mspace{2mu} (s_0,u_0)}$, we obtain the needed inequality:
        \begin{align*}
            \ip{\tilde{j}_{R,+}\mspace{1mu}\psi}{\Ha_C\mspace{1mu}\tilde{j}_{R,+}\mspace{1mu}\psi} &=  \ip{W_{\mspace{2mu} \mathcal{R}, \mspace{2mu} (s_0,u_0)}^\ast \mspace{1mu} \tilde{j}_{R,+}\mspace{1mu}\psi\mspace{2mu}}{\mspace{2mu}\Ha_S \mspace{1mu} W_{\mspace{2mu} \mathcal{R}, \mspace{2mu} (s_0,u_0)}^\ast \mspace{1mu}\tilde{j}_{R,+}\mspace{1mu}\psi}   \\
            &\geq \Sigma^S_{R-b} \, \norm{W_{\mspace{2mu} \mathcal{R}, \mspace{2mu} (s_0,u_0)}^\ast \mspace{1mu}\tilde{j}_{R,+}\mspace{1mu}\psi}^2 = \Sigma^S_{R-b} \, \norm{\tilde{j}_{R,+}\mspace{1mu}\psi}^2         
        \end{align*}
        Everything works similarly for $\Omega_{R,-}$, but we had to separate the cases because the matrix $\mathcal{R}$ and the translation vector $(s_0,u_0)$ might not be the same. Once these estimates are established on both branches, the remainder of the argument is exactly the same as in Proposition~\ref{sec:prelim:prop_persson_small_system} and yields $\Sigma_C \geq \Sigma_S$.
        \vspace{0.5cm}
        
        \item $\Sigma_S \geq \Sigma_C$: it is explained in Proposition \ref{sec:prelim:prop_persson_small_system} how to adapt the previous case to this one. Again, it remains to apply the corresponding unitary transformation to the interaction term on $\Omega_{R,+}$ and $\Omega_{R,-}$, and the required argument has already been given above.
    \end{itemize}
\end{proof}

\begin{notation}
    From now on, we simply denote $\Sigma$ for $\Sigma_S$ or $\Sigma_C$, as they are equal.
\end{notation}

We have thus obtained a first link between the bottoms of the essential spectra of $\Ha_S$ and $\Ha_C$, through $\Sigma_S=\Sigma_C$, $\Sigma_S=\inf\sess(\Ha_S)$ (Theorem \ref{sec:ess_spec:theorem_spectrum_straight_waveguide}) and Proposition \ref{sec:ess:first_inclusion_ess_spec}. \\
To complete the spectral picture of $\Ha_C$, we now seek to show, as anticipated in Subsection~\ref{model:subsec_general_ideas}, that its essential spectrum has no gaps above its bottom. We first prove that it contains the half-line $[\Sigma,+\infty[$. The argument is similar to that of Proposition \ref{sec:ess_spec:inf_straight_equal_sigma}, where the ball defining the ionization threshold was moved away at infinity. Here, after aligning the branch corresponding to $s \to -\infty$ with the straight waveguide, we translate the curved part to infinity along the longitudinal direction. Since this branch is infinitely long and straight, the translated curved waveguide should converge, in an appropriate sense, to the straight waveguide up to the corresponding Euclidean transformation. This is made rigorous in the next lemma:

\begin{lemma}
    Assume \textup{\ref{h3_field_rotation_invariance}}. Consider the straight branch of the waveguide corresponding to $s \to -\infty$ and let $\mathcal{R}$, $x_0$ be the rotation matrix and the translation giving this straight branch when applied to the straight waveguide $\R \times ]-a,a[$. By abuse of notation, we still denote by $\mathcal{R}$ the rotation operator on $\ldd$. Let 
    \begin{align*}
        U_{\mspace{2mu} \mathcal{R}, \mspace{2mu} x_0} =  e^{- i \mspace{3mu} x_0 \mspace{2mu}\cdot \mspace{2mu} p} \: \mathcal{R}  \otimes \left(  \Gamma\mspace{-3mu}\left(e^{- i \mspace{3mu} x_0 \mspace{2mu} \cdot \mspace{2mu} K}\right) \Gamma\left(\mathcal{R}\right) \right) 
    \end{align*}
    be the unitary operator realizing the Euclidean transformation aligning the straight and the curved waveguide.
    For $\xi \in \R$, let again $U_\xi := e^{-\mspace{2mu} i \mspace{1mu}P^{\,\mathrm{tot}}_1 \mspace{2mu} \xi }$ the translation along $\xu$ operator. Then: 
    \begin{align*}
         U_\xi \mspace{2mu} U_{\mspace{2mu} \mathcal{R}, \mspace{2mu} x_0}^\ast   \mspace{2mu} \Ha_C \mspace{2mu} U_{\mspace{2mu} \mathcal{R}, \mspace{2mu} x_0} \mspace{2mu} U_\xi^\ast \mspace{3mu} \underset{\xi \to +\infty}{\longrightarrow} \mspace{3mu}  \Ha_S   \hspace{0.4cm} \text{in strong resolvent sense}.
    \end{align*}
\end{lemma}
\begin{proof}
    By \cite[Theorem \cRM{8}.25]{RS}, it is enough to show strong convergence on a common core. By Proposition \ref{sec:prelim:domain_splitting_and_cores}(iii), $\ciz(\R^2) \alten \D\big(\dG(\omega(K))\big)$ is a core for $\Ha_S$ and $\Ha_C$. One sees that this core is invariant under rotations and translations, so that for all $\xi \in \R$, we have a common core for $ U_\xi \mspace{2mu} U_{\mspace{2mu} \mathcal{R}, \mspace{2mu} x_0}^\ast  \mspace{2mu}\Ha_C \mspace{2mu} U_{\mspace{2mu} \mathcal{R}, \mspace{2mu} x_0} \mspace{2mu} \mspace{2mu} U_\xi^\ast$ and $\Ha_S$. By linearity, it is enough to look at a product state $f \otimes g \in \ciz(\R^2) \alten \D\big(\dG(\omega(K))\big)$. \\
    As everything except the potential is invariant under translation and rotation (see the proof of Lemma \ref{sec:ess_spec:lemma_inv_translation}, which extends easily to rotations thanks to Hypothesis~\ref{h3_field_rotation_invariance}), one gets: 
    \begin{align*}
         U_\xi \mspace{2mu} U_{\mspace{2mu} \mathcal{R}, \mspace{2mu} x_0}^\ast   \mspace{2mu} \Ha_C \mspace{2mu} U_{\mspace{2mu} \mathcal{R}, \mspace{2mu} x_0} \mspace{2mu} U_\xi^\ast \mspace{3mu} - \Ha_S = 
         V_C\big(\mathcal{R}\left(X - (\xi,0)\right) + x_0\big) \otimes I  - V_S(X)   \otimes I. 
    \end{align*}
    Applying this to $f \otimes g$ yields that it suffices to compute: 
    \begin{align*}
        \nld{V_C\big(\mathcal{R}\left(X - (\xi,0)\right) + x_0 \big) f  - V_S(X)  f }^2.
    \end{align*}
    The rotation $\mathcal{R}$ and the translation $x_0$ define the Euclidean transformation mapping the straight waveguide onto the asymptotic branch corresponding to $s \to -\infty$. Using the compact curvature Assumption~\ref{h2_geom_compact_curvature}, there exists some $d<0$ such that for almost all $x$ with $x_1 < d$, only the straight branch is left and we have $V_C(x ) = V_S\big(\mathcal{R}^{-1}(x -  x_0)\big)$ or, equivalently, $V_C\big(\mathcal{R}\mspace{1mu}x+ x_0 \big) = V_S(x)$. Then, for almost every $x \in \R^2$, one sees that it suffices to take $\xi$ large enough to have $x_1 - \xi < d$ and then $V_C\big(\mathcal{R}\mspace{1mu}\left(x - (\xi,0)\right)+ x_0 \big) = V_S(x)$, giving us pointwise convergence as $\xi \to +\infty$. Since $V_C$, $V_S$ are bounded, the dominated convergence theorem yields: 
    \begin{align*}
        \nld{V_C\big(\mathcal{R}\left(X - (\xi,0)\right) + x_0 \big) f  - V_S(X)  f }^2 \underset{\xi \to +\infty}{\longrightarrow}  0.
    \end{align*}
    Hence the result.
\end{proof}

\noindent The strong resolvent convergence of the unitarily equivalent Hamiltonians yields the following spectral inclusion:

\begin{proposition}\label{sec:ess_spec:inclusion_sigma_curved_spec}
    Assume \textup{\ref{h3_field_rotation_invariance}}. One has $[\Sigma,+\infty[ \subset \sess(\Ha_C)$.
\end{proposition}
\begin{proof}
     By \cite[Theorem \cRM{8}.24(a)]{RS} and the previous proposition, for all $\lambda \in \sigma(\Ha_S) = [\Sigma,+\infty[$, since $U_{\mspace{2mu} \mathcal{R}, \mspace{2mu} x_0}$, $U_\xi$ are unitary:
     \begin{align*}
         \exists (\lambda_\xi)_{\xi \in \R} \in \sigma\big(U_\xi \mspace{2mu} U_{\mspace{2mu} \mathcal{R}, \mspace{2mu} x_0}^\ast   \mspace{2mu} \Ha_C \mspace{2mu} U_{\mspace{2mu} \mathcal{R}, \mspace{2mu} x_0} \mspace{2mu} U_\xi^\ast \mspace{3mu}\big) = \sigma(\Ha_C), \: \lambda_\xi \underset{\xi \to +\infty}{\longrightarrow} \lambda.
     \end{align*}
     Since the spectrum is a closed set, $\lambda \in \sigma(\Ha_C)$. As a consequence, $[\Sigma,+\infty[ \subset \sigma(\Ha_C)$ and in particular, since being in the discrete spectrum means being isolated, no point of a half-line can be in the discrete spectrum, so $[\Sigma,+\infty[ \subset  \sess(\Ha_C)$.
\end{proof}

Denote $E_C = \inf\sigma(\Ha_C)$. Under \ref{h3_field_rotation_invariance}, we get $[\Sigma,+\infty[\subset\sess(\Ha_C)\subset[\min(\Sigma,E_C+m),+\infty[$ from Proposition \ref{sec:ess:first_inclusion_ess_spec} and the preceding one. If $\min(\Sigma,E_C+m)=\Sigma$, these inclusions coincide and $\sess(\Ha_C)=[\min(\Sigma,E_C+m),+\infty[$. It remains to consider the case $E_C+m<\Sigma$. For this, we use the $\supset$ part of Theorem~4.1 of \cite{acqftmpfh}, which states:

\begin{proposition}[{\cite[Theorem 4.1]{acqftmpfh}}]\label{sec:ess_spec:theorem_gerard}
    Assume $\Ha$ has a ground state. Then:
    \begin{align*}
        \sess(\Ha) \supset [\inf\sigma(\Ha)+m, +\infty[.
    \end{align*} 
\end{proposition}
\begin{proof}
The $\supset$ part of \cite[Theorem~4.1]{acqftmpfh} only uses, in the notation of that theorem, assumptions (H1), (I1), which are satisfied in our setting, and the existence of a ground state. However, their argument for the existence of a ground state relies on the compact resolvent of the small-system, which is not available here. We therefore retain the existence of a ground state as an assumption.
\end{proof}

Combining the two previous propositions and Proposition \ref{sec:ess:first_inclusion_ess_spec}, we establish the following HVZ-type theorem for the curved waveguide:

\begin{theorem}[HVZ theorem]\label{sec:ess_spec:HVZ_theorem_curved_waveguide}
    Assume \textup{\ref{h3_field_rotation_invariance}} and let $\Sigma$ be the ionization threshold of $\Ha_C$. One has: 
    \begin{align*}
        \sess(\Ha_C) = \left[\min\big(\Sigma, \inf\sigma(\Ha_C)+m\big), +\infty\right[, 
    \end{align*}
    where $\Sigma$ is also the ionization threshold of $\Ha_S$, with $\Sigma = \inf \sigma(\Ha_S) = \inf \sess (\Ha_S)$.
\end{theorem}
\begin{proof}
    Denote $E_C := \inf\sigma(\Ha_C)$. If $\min(\Sigma, E_C+m) = E_C + m > E_C$, then $\Ha_C$ has a ground state by Proposition \ref{sec:ess:first_inclusion_ess_spec}, and combining Propositions \ref{sec:ess_spec:theorem_gerard} and \ref{sec:ess:first_inclusion_ess_spec} yields the result. If $\min(\Sigma, E_C+m) = \Sigma$, Propositions \ref{sec:ess_spec:inclusion_sigma_curved_spec} and \ref{sec:ess:first_inclusion_ess_spec} yield the result.
\end{proof}

The objectives of this section were to establish the gapless structure of the essential spectra associated with the straight and curved waveguides, and to relate their bottoms. Theorems \ref{sec:ess_spec:theorem_spectrum_straight_waveguide} and \ref{sec:ess_spec:HVZ_theorem_curved_waveguide} achieve these objectives through the common ionization threshold $\Sigma$. Unlike for waveguides without a quantum field, however, we do not obtain an equality between the essential spectra of the straight and curved systems. Instead, we have an interdependence, since the bottom of $\sess(\Ha_C)$ is $\min(\Sigma,E_C+m)$, where the additional threshold $E_C+m$ reflects the creation of one boson.

\smallbreak

This difference does not alter the criterion relevant to the existence of discrete spectrum. Since $E_C<E_C+m$, one has $E_C<\inf\sess(\Ha_C)$ if and only if $E_C<\Sigma=\inf\sigma(\Ha_S)$. Thus, as in the uncoupled case (see Figures \ref{fig:straight_waveguide_and_spectrum} and \ref{fig:curved_waveguide_and_spectrum}), \textit{the question is whether the geometry of the curved waveguide produces an energy below the bottom of the spectrum of the Hamiltonian associated with the straight waveguide}. This is consistent with Persson's formula: the essential spectrum is governed by what happens at infinity, whereas discrete spectrum may arise from a local geometric effect, such as curvature.

% À FAIRE DANS CETTE SECTION

%%% 

%%%%%% 

% 5) Binding condition
\section{Binding condition and discrete spectrum}
\label{sec:binding}

We now derive a sufficient condition for $\Ha_C$ to have non-empty discrete spectrum, and hence a ground state. In the framework of NRQED, and particularly for atomic models, the inequality $\Sigma>\inf\sigma(\mathcal{H})$ is often imposed as an additional hypothesis. This requirement is known as the \textbf{binding condition} (see, e.g., \cite[(3)]{gsnrqed}, \cite[(b)]{analisa} or \cite[Assumption 3.1]{hiroshima2019ground}). In the massive setting considered here, it ensures the existence of a ground state via Proposition \ref{sec:ess:first_inclusion_ess_spec}; in massless models, it instead controls the escape of the matter subsystem, while the infrared behavior of the field requires separate estimates. Making such an assumption is not appropriate here, as we aim to discuss the existence of a ground state in relation to geometric conditions. Thus, rather than postulating $\Sigma>\inf\sigma(\mathcal{H})$ directly, we derive below a sufficient condition for this inequality to hold. Since
this condition depends explicitly on both $V_S$ and $V_C$, we expect it
to be satisfied under suitable geometric assumptions; establishing such
assumptions is left for future work.

\smallbreak

A case in which the binding condition can be proved is the single-particle case; see, e.g.,  \cite[Proposition 3.11]{hiroshima2019ground}, \cite{takaesu_binding} or \cite[Theorem 3.1]{gsnrqed} combined with \cite[Theorem 3]{editnrqed}. In these works, the Hamiltonian $\Ha$ is split into its translation-invariant part $\Ha^0$ and an external potential which vanishes at infinity. It is this decay of the potential which allows one to prove $\inf\sigma(\Ha^0)\leq\Sigma$, and then $\inf\sigma(\Ha)<\Sigma$. In our setting, the waveguide potential does not vanish at infinity, since the waveguide is infinite. Consequently, this argument does not apply directly.
 
\smallbreak

However, under Assumption~\ref{h3_field_rotation_invariance}, we have shown that $\Sigma_C = \Sigma_S = \Sigma$ (Proposition \ref{sec:ess_spec:prop_equality_ionization_thresholds}) and that $\Sigma = \inf \sigma(\Ha_S) = \underset{\eta \in \R}{\inf} \inf \sigma(\Ha_S(\eta))$ (Theorem \ref{sec:ess_spec:theorem_spectrum_straight_waveguide}). Here, $\Ha_S$ can be seen as the $\xu$-translation-invariant part of the Hamiltonian, which can then play the role of $\Ha^0$ if we write:
\begin{align*}
    \Ha_C = \underbrace{\Ha_S}_{\text{$\xu$-translation-invariant part}} + \hspace{0.2cm} V_C(X)-V_S(X).
\end{align*}
Since $\Ha_S$ is translation invariant only in the $\xu$-direction, we adapt the trial state used in \cite{takaesu_binding} by separating the longitudinal and transverse variables: instead of considering $u(x)\Psi(x)$ with $x\in\R^2$, we take $u\otimes\Psi$, with $u\in L^2(\R_{\xu})$ and $\Psi\in L^2(\R_{\xd},\Fs)$. In order to extract $\Sigma$ from the computations below, we first select an appropriate fiber. Indeed, as the function $\eta\mapsto\inf\sigma(\Ha_S(\eta))$ is continuous and goes to $+\infty$ when $\eta\to\pm\infty$ (see the proof of Proposition \ref{sec:ess:prop_no_gap_straight_spec}), it has a global minimum, i.e. there exists $\eta_0\in\R$ such that $\Sigma=\inf\sigma(\Ha_S(\eta_0))$. Assuming for now that $\Ha_S(\eta_0)$ has a ground state, and letting $\Psi$ be such a ground state, we can consider the following suitably transformed version of $u\otimes\Psi$:
\begin{align*}
    \Upsilon := \left( \int_{\R^2}^\oplus e^{\mspace{1mu} i \mspace{2mu} \eta_0 \mspace{2mu} \xu }  \, e^{-\mspace{2mu} i \mspace{3mu} \dG(K_1) \mspace{3mu} \xu } \: d\xu \mspace{2mu} d\xd  \right) u \otimes \Psi.
\end{align*}
Formally, if $u$, $\Psi$ are normalized and $u$ is real-valued, one gets (see proof of Theorem~\ref{sec:binding:theorem_binding_condition} below)
\begin{align*}
    \inf\sigma(\Ha_C) \leq \ip{\Upsilon}{\Ha_C \Upsilon} &= \ip{\Upsilon}{\bigl(\Ha_S + V_C(X) - V_S(X)\bigr) \Upsilon}  \\
    &= \ipld{u}{\left(-\partial_{\xu}^2+V_{\eff,\Psi}(X_1)\right)u} + \Sigma, 
\end{align*}
where $V_{\eff,\Psi}$ is an \textit{effective potential} given by: 
\begin{align*}
    V_{\eff,\Psi} := \int_{\R}(V_C-V_S)(\, \cdot \, ,\xd)\snfs{\Psi(\xd)} d\xd.
\end{align*}
In this situation, it suffices to find $u \in \D(-\partial_{\xu}^2+V_{\eff,\Psi}(X_1))$ so that $\ipld{u}{\hspace{-2pt}\left(-\partial_{\xu}^2+V_{\eff,\Psi}(X_1)\right)\hspace{-2pt}u} < 0$ to obtain $\inf\sigma(\Ha_C) < \Sigma$. In the remainder of this work, \textbf{we shall refer to the existence of such a negative-energy state as the binding condition}.

\medbreak

The above formal computations will be justified in Subsection \ref{sec:binding:rigorous_computations_binding}. We have thus reduced the existence of discrete spectrum for $\Ha_C$ to the following two statements:
\vspace{7pt}
\begin{mdframed}
\begin{enumerate}[label=(\alph*)]
\item\label{sec:binding:statement_a} The Hamiltonian $\Ha_S(\eta_0)$ has a ground state $\Psi$, where $\eta_0 \in \underset{\eta \, \in \, \R}{\argmin}\, \inf\sigma(\Ha_S(\eta))$.
\item\label{sec:binding:statement_b} $\sigma \big(-\partial_{\xu}^2+V_{\eff,\Psi}(X_1)\big) \: \cap \hspace{3pt}  ]-\infty, 0[ \hspace{4pt}  \neq \, \emptyset$.
\end{enumerate} 
\end{mdframed}

\vspace{7pt}
\textit{In what follows, we  first prove statement \ref{sec:binding:statement_a}, and then use it to show rigorously that condition \ref{sec:binding:statement_b} implies $\inf\sigma(\Ha_C)<\Sigma$, i.e. implies binding.} The verification of condition \ref{sec:binding:statement_b} under suitable geometric assumptions is left for future work.

%===================================================================
%       Existence of a ground state for the fibered Hamiltonian
%===================================================================

\subsection{Existence of a ground state for the fibered Hamiltonian}\label{sec:binding:ground_state_fibered}

In this subsection, we prove statement \ref{sec:binding:statement_a} by studying each fiber of $\Ha_S$. To this end, we introduce:
\begin{notation}
    Let $\Sigma(\eta)$ denote the ionization threshold associated with $\Ha_S(\eta)$ (see Definition \ref{sec:ess:def_ionization_threshold}), and let $E(\eta) := \inf\sigma(\Ha_S(\eta))$ be the notation from the proof of Proposition \ref{sec:ess:prop_no_gap_straight_spec}.
\end{notation}

\noindent To control the estimates used below and clarify the overall proof strategy, we first establish exponential decay below the ionization threshold:

\begin{proposition}[Exponential decay for $\Ha_S(\eta)$]\label{sec:binding:exp_decay_hamiltonian}
     Let $\eta \in \R$ and $\mathcal{J}$ be a closed interval in $]-\infty,\Sigma(\eta)[$. Then for all $0 < \alpha \displaystyle < \sqrt{\Sigma(\eta)-\sup\mathcal{J}}$: 
    \begin{align*}
        \nop{e^{\, \alpha \mspace{2mu} \abs{X_2}} \, E^{\Ha_S(\eta)}(\mathcal{J})} < + \infty,
    \end{align*}
    where $E^{\Ha_S(\eta)}$ is the projection-valued measure of $\Ha_S(\eta)$.
\end{proposition}
\begin{proof}
    As in the proofs of Lemma \ref{ess_spec:lemma:IMS}, Proposition \ref{sec:ess_spec:prop_exp_decay} and Proposition \ref{prelim:prop:general_quadratic_form}, but with the additional $(\eta-P_1)^2$ term.
\end{proof}

We aim to prove an analog of Proposition \ref{sec:ess:first_inclusion_ess_spec} for $\Ha_S(\eta)$. Unlike in the previous case, this cannot be achieved simply by invoking the existing literature, due to the presence of the $(\eta - P_1)^2$ term in $\Ha_S(\eta)$ that falls outside the standard non-fibered Nelson framework. The fibers of translation-invariant Nelson Hamiltonians have been previously investigated, especially in \cite{moller_translation_invariant_nelson_model} and subsequent work, but always in the fully translation-invariant setting. In contrast, our Hamiltonian is only translation invariant along the longitudinal $\xu$ axis. A consequence of this partial invariance is that the transverse degree of freedom of the small system remains present alongside the $(\eta-I\otimes\dG(K_1))^2$ term; thus, unlike in the standard one-particle fully translation-invariant Nelson model, the fibered Hamiltonian cannot be reduced to an operator defined exclusively on the Fock space\footnote{\, This comment is specific to translation-invariant Nelson models. 
In Pauli--Fierz models, fiber Hamiltonians retaining matter degrees of freedom have been studied, for instance in mobile atom or ion models (see, e.g., \cite{amour_grebert_guillot_dressed_electron_magnetic_field}).}.
\\

Consequently, we need to adapt the proofs of the existence of a ground state for Nelson Hamiltonians to our case, with the additional $(\eta - I \otimes \dG(K_1))^2$ term. The key observation to achieve this goal is that, at the level of the fibered Hamiltonian, the potential vanishes at infinity, since we only consider the transverse variable on the fibers. As explained at the beginning of Section \ref{sec:binding}, this is a crucial hypothesis to get $\Sigma(\eta) > \inf \sigma(\Ha_S(\eta))$. If we had this, it would be enough to prove an analog of Proposition \ref{sec:ess:first_inclusion_ess_spec} to determine whether $\Ha_S(\eta)$ has a ground state. These considerations lead to the following theorem:

\begin{theorem}\label{sec:binding:theorem_grond_state_fibered_hamiltonian}
    Let $\eta \in \R$, $\eta_0 \in \underset{\eta \, \in \, \R}{\argmin}\, E(\eta)$ and $E_0^{(1)}(\eta) = \underset{k_1 \in \R}{\inf} \, E(\eta - k_1) + \omega(k_1,0)$. Then: 
    \begin{enumerate}[label=(\roman*), itemsep=10pt]
        \item\label{theorem_ground_state_fibered_assertion_i} $\sess(\Ha_S(\eta)) \subset \left[\min\hspace{-1pt}\left(\Sigma(\eta), E_0^{(1)}(\eta) \right),+\infty  \right[$.
        
        \item\label{theorem_ground_state_fibered_assertion_ii} $\Sigma(\eta) > E(\eta)$.
        
        \item\label{theorem_ground_state_fibered_assertion_iii} $E_0^{(1)}(\eta_0) = E(\eta_0) + m > E(\eta_0)$ and $E(\eta_0) = \Sigma_S$. 
        
        \vspace{0.2cm}
    \end{enumerate}
    Hence $\Ha_S(\eta_0)$ has a ground state $\Psi$ such that $\Ha_S(\eta_0) \Psi = \Sigma_S \mspace{1mu}\Psi$. 
\end{theorem}
\vspace{0.5cm}

For the proof of assertion (i), fix $\eta\in\R$. To establish it, we adapt the techniques of the proofs of \cite[Theorem 2]{AC_rayleigh_scatt_1} and \cite[Theorem 4.1]{acqftmpfh} (see also \cite[Theorem 6.1]{analisa}). We also take inspiration from \cite{moller_translation_invariant_nelson_model}, which treats translation-invariant Nelson Hamiltonians. These works use the same standard notations for the partition of unity, the extended Hilbert space and the $\check{\Gamma}$ operators, which we adopt here. These constructions and notations are recalled in Appendix~\ref{app:fock},  Subsection~\ref{app:fock:extended}. The arguments of these works must nevertheless be adapted to account for the additional $(\eta-I\otimes\dG(K_1))^2$ term in $\Ha_S(\eta)$. The proof of assertion (i) will be completed after establishing a sequence of auxiliary lemmas.

\begin{notation}
For the sake of readability, we set: $\Gc_R:=I\otimes\Gc(j_R)$.
\end{notation}

The idea of the proof is that the operators $\Gc_R$ separate the bosons localized in a bounded region from those which are localized far away. Below the ionization threshold, the particle is exponentially localized in the transverse direction, so that, for large $R$, the bounded region introduced above contains the area where the particle is concentrated. Consequently, the bosons localized far away interact only negligibly with the particle and are therefore represented as asymptotically free. This leads us to consider the extended Hilbert space
\begin{align*}
    \Hi^{\, \ext} := L^2(\R_{\xd}) \otimes \Fs \otimes \Fs = L^2(\R_{\xd},\Fs)\otimes\Fs,    
\end{align*}
where the first subsystem $L^2(\R_{\xd},\Fs)$ describes the matter degrees of freedom together with the bosons localized in the bounded region, while the second Fock factor describes the free bosons at infinity. On this space, we introduce the following extended Hamiltonian:
\begin{align*}
\Ha_S^\ext(\eta)
:={}& \bigl(-\partial^2_{\xd}+V(X_2)\bigr)\otimes I \otimes I
+ I\otimes\dG^\ext(\omega(K))
+ \bigl(\eta-I\otimes\dG^\ext(K_1)\bigr)^2
+ g\mspace{1mu}\Ha_{I,0}\otimes I \\
:={}& \bigl(-\partial^2_{\xd}+V(X_2)\bigr)\otimes I \otimes I
+ I\otimes\dG^\ext(\omega(K))
+ \bigl(\eta-P_1^\ext\bigr)^2
+ g\mspace{1mu}\Ha_{I,0}\otimes I.
\end{align*}
It describes the system after this separation: the interaction is retained in the first subsystem, whereas the second Fock factor carries only the free field terms associated with the bosons localized far away. In view of the preceding discussion, the action of $\Ha_S(\eta)$ after the separation induced by $\Gc_R$ becomes asymptotically that of the extended Hamiltonian as $R \to +\infty$, provided we are below the ionization threshold. Spectral information on $\Ha_S^\ext(\eta)$ can therefore be transferred back to $\Ha_S(\eta)$, as made precise in Proposition~\ref{sec:binding:prop_comparison_with_extended_hamiltonian}. In particular, the extended Hamiltonian allows us to identify the threshold $E_0^{(1)}(\eta)$ as the lowest energy of configurations with at least one asymptotically free boson. Since the matter subsystem cannot escape in the transverse direction below $\Sigma(\eta)$, neither the matter subsystem nor the field can escape to infinity below $\min\bigl(\Sigma(\eta),E_0^{(1)}(\eta)\bigr)$, so there should not be any essential spectrum below this threshold.

\medbreak

We now turn to the estimates needed to implement this comparison. We will use the following dense subspace of $\Fs$, on which all compositions of field operators used below are well defined:
\begin{notation}
We set $\Fs^\infty:=\Fs^\fin\big(\ciz(\Rk^2)\big)$, where $\Fsf(\cdot)$ is the notation from Subsection~\ref{app:fock:field}.
\end{notation}

\begin{lemma}\label{sec:binding:lemma_error_field}
Let $\alpha > 0$ and $\Nb = \dG(I)$ the number operator. Then the operator
    \begin{align*}
        \left( \Gc_R\big( I \otimes \dG(\omega(K)) + g\mspace{2mu}\Ha_{I,0}\big) - ( I \otimes \dG^\ext(\omega(K)) + g\mspace{2mu}\Ha_{I,0} \otimes I\big)\Gc_R \right)(I \otimes \Nb + 1)^{-1} 
    \end{align*}
    is bounded. Moreover, the following estimate holds in operator norm:
    \begin{align*}
        e^{-\alpha\abs{X_2}} \hspace{-2pt} \left( \Gc_R\big( I \otimes \dG(\omega(K)) + g \Ha_{I,0}\big) - ( I \otimes \dG^\ext(\omega(K)) + g \Ha_{I,0} \otimes I\big)\Gc_R \right)(I \otimes \Nb + 1)^{-1} \underset{R\to+\infty}{=} \hspace{-2pt} o(1).
    \end{align*}
\end{lemma}
\begin{proof}
    Throughout this proof, we use the notation of \cite{AC_rayleigh_scatt_1}. We adapt the argument from \cite[Lemma 32(i)]{AC_rayleigh_scatt_1} evaluated at $m=0$ (where $m$ denotes the exponent parameter of that reference), with the only difference lying in our treatment of the $\Ha_{I,0}$ term. Since we do not assume hypothesis (H5) of that reference, the convergence of this term is not covered there. Its boundedness nevertheless follows from the relative boundedness of the Segal operator with respect to the number operator (see, e.g., \cite[Corollary 5.9]{Arai}). It remains to establish the required convergence:
    \begin{align*}
    \underset{\xd \,\in\, \R}{\sup} \; \, e^{-\mspace{1mu}\alpha \mspace{2mu}\abs{\xd}}\left(\norm{(j_{0,R}-1)\, v_{(0,\xd)}}
    + \norm{j_{\infty,R}\,v_{(0,\xd)}} \right) \underset{R \, \to \, + \mspace{1mu} \infty}{\longrightarrow}0.
    \end{align*}
    Let $y = i \nabla_k$ be the configuration variable for the field. Recall that in configuration space, $v_x(y) = v(y-x)$ (see Hypothesis \ref{h2_field_interaction}). One sees that the previous convergence is true if:
    \begin{align*}
    \underset{\xd \,\in\, \R}{\sup} \; \, e^{-\mspace{1mu}\alpha \mspace{2mu}\abs{\xd}}\norm{\ind_{\left\{\mspace{1mu} \abs{y} \, \geq \, R\right\}}\, v_{(0,\xd)}}
     \underset{R \, \to \, + \mspace{1mu} \infty}{\longrightarrow} 0.      
    \end{align*}
    If $\abs{\xd} \geq R/2$, then we just write $e^{-\mspace{1mu}\alpha \mspace{2mu}\abs{\xd}} \leq e^{-\mspace{1mu}\alpha \mspace{2mu}R/2}$ and $\norm{\ind_{\left\{\mspace{1mu} \abs{y} \, \geq \, R\right\}}\, v_{(0,\xd)}} \leq \norm{v}$. If $\abs{\xd} \leq R/2$, then $\abs{y} \geq R$ means $\abs{y-(0,\xd)} \geq R/2$ and a change of variables yields $\norm{\ind_{\left\{\mspace{1mu} \abs{y} \, \geq \, R\right\}}\, v(\cdot - (0,\xd))} \leq \norm{\ind_{\left\{\mspace{1mu} \abs{z} \, \geq \, R/2 \right\}}\, v} $. Consequently: 
    \begin{align*}
    \underset{\xd \,\in\, \R}{\sup} \; \, e^{-\mspace{1mu}\alpha \mspace{2mu}\abs{\mspace{1mu}\xd}}\norm{\ind_{\left\{\mspace{1mu} \abs{y} \, \geq \, R\right\}}\, v_{(0,\xd)}} \leq e^{-\mspace{1mu}\alpha \mspace{2mu}R/2} \norm{v} + \norm{\ind_{\left\{\mspace{1mu} \abs{z} \, \geq \, R/2 \right\}}\, v}
     \underset{R \, \to \, + \mspace{1mu} \infty}{\longrightarrow} 0.      
    \end{align*}
\end{proof}

We need to prove that this holds for $(\eta - I \otimes \dG(K_1))^2$ as well. First, we observe that:

\begin{lemma}\label{sec:binding:lemma_N_properties_and_bounds}
The following properties hold:
\begin{enumerate}[label=(\roman*)]
\item The number operators $\Nb$ and $\Nb^\ext$ strongly commute with all the relevant second quantization operators. Moreover, for all $p\in\R$, $(\Nb+1)^p\Fs^\infty=\Fs^\infty$, and $\Gc(j_R) f(\Nb)=f(\Nb^\ext)\Gc(j_R)$ for all measurable functions $f:\R \to \C$.

\item The operators $(\eta-P_1)(\Ha_S(\eta)-z)^{-1}$ and $(I\otimes\Nb)(\Ha_S(\eta)-z)^{-1}$ are bounded for each $z\in\rho(\Ha_S(\eta))$. Similarly, the operators $(\eta-P_1^\ext)(\Ha_S^\ext(\eta)-z)^{-1}$ and $(I\otimes\Nb^\ext)(\Ha_S^\ext(\eta)-z)^{-1}$ are bounded for each $z\in\rho(\Ha_S^\ext(\eta))$.
\end{enumerate}
\end{lemma}

\begin{proof}
\textit{(i)} The strong commutativity follows from \cite[Section 4.13]{Arai} and the fact that $\Nb$ is the second quantization of the identity. The claim $(\Nb+1)^p\Fs^\infty=\Fs^\infty$ is immediate (it suffices to decompose on $n$-particles sectors). The equality $\Gc(j_R) \Nb = \Nb^\ext \Gc(j_R)$ follows from \cite[(21)]{AC_rayleigh_scatt_1} and combining it with \cite[Proposition 5.15]{Sch} and functional calculus yields the result. \\
\textit{(ii)} We only prove it for $(\eta-P_1)$ and $I \otimes \Nb$, the extended case being identical. Using \cite[Lemma 1.9]{Arai}, it suffices to have $\D(\Ha_S(\eta)) \subset \D(P_1) \cap \D(I\otimes \Nb)$. Theorem \ref{sec:ess:prop_fiber_straight_hamiltonian} gives $\D(\Ha_S(\eta)) \subset \D(P_1) \cap \D(\Ha_f)$. \\
Now we have $0 < m \leq \omega(K)$ hence $\Nb^2 \leq m^{-2} \, \dG(\omega(K))^2$ by \cite[Proposition 3.4(ii)]{lemma_power_dG_gerar_moller}, which means that $\Nb$ is $\dG(\omega(K))$-bounded in the operator sense. Consequently, $\D(\Ha_f) \subset \D(I\otimes\Nb)$.
\end{proof}

\noindent Then we compute: 
\begin{lemma}
    Let $H$ denote the restriction of $\Ha_S(\eta)+i$ to $\D(\ha_{\el,0}) \alten \Fs^\infty$. On $\Ran(H)$, one gets, in operator norm: 
    \begin{align*}
       (\Ha_S^\ext(\eta)+i)^{-1} \left( \Gc_R\big( \eta - P_1\big)^2 - ( \eta - P^\ext_1\big)^2 \, \Gc_R \right)(\Ha_S(\eta)+i)^{-1} \underset{R\to+\infty}{=} O\hspace{-2pt}\left(R^{-1}\right).
    \end{align*}
\end{lemma}
\begin{proof}
    By the computation at the beginning of the proof of \cite[Lemma 3.1]{moller_translation_invariant_nelson_model}, combined with the estimate of \cite[Section 2.4]{AC_rayleigh_scatt_1}, we obtain, in operator norm:
    \begin{equation}\label{sec:binding:estimate_P_1}
         \left( \mspace{1mu} \Gc(j_R) \big( \eta - \dG(K_1)\big) - ( \eta -  \dG^\ext(K_1)\big) \, \Gc(j_R) \right)(\Nb + 1)^{-1} \underset{R\to+\infty}{=} O\hspace{-2pt}\left(R^{-1}\right).
    \end{equation}
     Let $A_R := \Gc(j_R) \big( \eta - \dG(K_1)\big) - ( \eta -  \dG^\ext(K_1)\big) \, \Gc(j_R)$. Using Lemma \ref{sec:binding:lemma_N_properties_and_bounds}(i), we write that, on $\Fs^\infty$: 
    \begin{align*}
        A_R \, (\Nb + 1)^{-1} = (\Nb^\ext + 1)^{-1}A_R. 
    \end{align*}
    Still on $\Fs^\infty$, we compute: 
    \begin{align*}
        \Gc(j_R) \big( \eta - \dG(K_1)\big)^2 - ( \eta -  \dG^\ext(K_1)\big)^2 \, \Gc(j_R) = \big(\eta - \dG^\ext(K_1)\big)A_R + A_R \big(\eta - \dG(K_1)\big). 
    \end{align*}
    On $\D(\ha_{\el,0}) \alten \Fs^\infty$, it becomes: 
    \begin{equation}\label{sec:binding:estimate_A_R}
         \Gc_R\big( \eta - P_1\big)^2 - \big( \eta -  P_1^\ext\big)^2 \, \Gc_R = \big(\eta - P_1^\ext\big)(I \otimes A_R) + (I \otimes A_R) \big(\eta - P_1\big). 
    \end{equation}
    The operator $(I \otimes A_R)(\Ha_S(\eta)+i)^{-1}$ is bounded because $(I \otimes A_R)(I \otimes \Nb +1)^{-1}$ and $(I \otimes \Nb+1)(\Ha_S(\eta)+i)^{-1}$ are (the former because of \eqref{sec:binding:estimate_P_1}, the latter by the previous lemma). The boundedness of $(\Ha_S^\ext(\eta)+i)^{-1}(I \otimes A_R)$ follows analogously, with the additional use of Proposition~\ref{r_B A bounded if A B bounded}. Combining these observations with \eqref{sec:binding:estimate_P_1}, we get: 
    \begin{align*}
        (I \otimes A_R)(\Ha_S(\eta)+i)^{-1}, \hspace{0.2cm}  (\Ha_S^\ext(\eta)+i)^{-1}(I \otimes A_R) \underset{R\to+\infty}{=} O\hspace{-2pt}\left(R^{-1}\right).
    \end{align*}
    In order to right-multiply \eqref{sec:binding:estimate_A_R} by $(\Ha_S(\eta)+i)^{-1}$, we must ensure that the range of this resolvent is contained in $\D(\ha_{\el,0}) \alten \Fs^\infty$: this is why we consider $\Ran(H)$. Using Lemma \ref{sec:binding:lemma_N_properties_and_bounds}(ii) and Proposition~\ref{r_B A bounded if A B bounded} to get the boundedness of $(\eta-P_1)(\Ha_S(\eta)+i)^{-1}$, $(\Ha_S^\ext(\eta)+i)^{-1}(\eta-P^\ext_1)$, we obtain, on $\Ran(H)$:
    \begin{align*}
        (\Ha_S^\ext(\eta)+i)^{-1} \left(\big(\eta - P_1^\ext\big)(I \otimes A_R) + (I \otimes A_R) \big(\eta - P_1\big)\right)(\Ha_S(\eta)+i)^{-1}  \underset{R\to+\infty}{=} O\hspace{-2pt}\left(R^{-1}\right).
    \end{align*}
    Equality \eqref{sec:binding:estimate_A_R} then yields the claim.
\end{proof}

Gathering all this, we get:

\begin{lemma}
    Let $z \in \C\setminus\R$ and $\chi_1 \in \ciz\big(\mspace{2mu}]-\infty,\Sigma(\eta)[\mspace{2mu}\big)$. One has, on $L^2(\R_{\xd},\Fs)$ and in operator norm:
    \begin{align*}
        \left(\Gc_R\,(\Ha_S(\eta)-z)^{-1} - (\Ha_S^\ext(\eta)-z)^{-1}\Gc_R\right) \chi_1(\Ha_S(\eta)) \underset{R \, \to \, +\infty}{=}  o(1) \hspace{4pt} \frac{(1+\abs{z})^2}{\abs{\Im z}^{2}}.
    \end{align*}
\end{lemma}
\begin{proof}
    On $\D(\ha_{\el,0})\alten \Fs^\infty$, we can write: 
    \begin{align*}
         \Gc_R \, \Ha_S(\eta) - \Ha_S^\ext(\eta) \mspace{1mu} \Gc_R  = \mathcal{E}^{\mathrm{field}}_R + \mathcal{E}^{\mathrm{kin}}_R, 
    \end{align*}
    where: 
    \begin{align*}
        \mathcal{E}^{\mathrm{field}}_R &= \Gc_R\big( I \otimes \dG(\omega(K)) + g\mspace{1mu}\Ha_{I,0}\big) - ( I \otimes \dG^\ext(\omega(K)) + g\mspace{1mu}\Ha_{I,0} \otimes I\big)\Gc_R \\
        \mathcal{E}^{\mathrm{kin}}_R &= \Gc_R\big( \eta - P_1\big)^2 - ( \eta - P^\ext_1\big)^2 \, \Gc_R.
    \end{align*}
    On $\Ran(H)$ (defined as in the previous lemma), this yields: 
    \begin{equation}\label{sec:binding:eq_lemma_estimate}
        \Gc_R\,(\Ha_S(\eta)+i)^{-1} - (\Ha_S^\ext(\eta)+i)^{-1}\Gc_R = - \mspace{2mu}(\Ha_S^\ext(\eta)+i)^{-1}\left(\mathcal{E}^{\mathrm{field}}_R + \mathcal{E}^{\mathrm{kin}}_R \right)(\Ha_S(\eta)+i)^{-1}.
    \end{equation}
    The $(-\partial_{\xd}^2  + V(X_2)) \otimes I$ term does not appear in $\mathcal{E}^{\mathrm{field}}_R$ because $\Gc_R$ does not act on the $\xd$ variable, so the Schrödinger terms cancel out. Using Lemmas \ref{sec:binding:lemma_error_field} and \ref{sec:binding:lemma_N_properties_and_bounds}(ii), we get that $(\Ha_S^\ext(\eta)+i)^{-1} \mathcal{E}^{\mathrm{field}}_R (\Ha_S(\eta)+i)^{-1}$ is a bounded operator. For the other term in \eqref{sec:binding:eq_lemma_estimate}, the previous lemma shows that the operator, defined on $\Ran(H)$, admits a bounded extension satisfying the following estimate: 
    \begin{align*}
        (\Ha_S^\ext(\eta)+i)^{-1} \mathcal{E}^{\mathrm{kin}}_R (\Ha_S(\eta)+i)^{-1}  \underset{R\to+\infty}{=} O\hspace{-2pt}\left(R^{-1}\right).
    \end{align*}
    Hence both sides in \eqref{sec:binding:eq_lemma_estimate} are bounded operators agreeing on $\Ran(H)$. In addition, as in Proposition \ref{sec:prelim:domain_splitting_and_cores}(iii), we get that $\D(\ha_{\el,0}) \alten \Fs^\infty$ is a core for $\Ha_S(\eta)$, so by the characterization of essential self-adjointness (see, e.g., \cite[Proposition 1.30]{Arai}), $\Ran(H)$ is a dense set. Two bounded operators agreeing on a dense set are equal, so \eqref{sec:binding:eq_lemma_estimate} holds on the whole Hilbert space. \\
    It remains to derive an estimate for the field term. Everything in $\mathcal E_R^{\mathrm{field}}$ is decomposable in the $\xd$ variable, hence $\mathcal E_R^{\mathrm{field}}$ commutes with $e^{-\alpha\abs{X_2}}$ (upon identifying multiplication operators on $L^2(\R_{\xd},\Fs)$ and $L^2(\R_{\xd},\Fs\otimes\Fs)$). Then, we observe that the operator $B := (I \otimes \Nb+1) \, e^{\alpha \abs{X_2}} \chi_1(\Ha_S(\eta))$ is bounded for $\alpha > 0$ small enough. Indeed, the proof of \cite[Lemma 31(iii)]{AC_rayleigh_scatt_1} applies almost verbatim: the only issue is that, when conjugating $-\partial_{\xd}^2$ by the weight $e^{\alpha\abs{X_2}}$, the second derivative of $\abs{X_2}$ produces a Dirac term at $x_2=0$. We circumvent this difficulty by replacing $\abs{X_2}$ in that argument by the smooth weight $g(X_2)$, where $g(t):=\sqrt{1+t^2}$. Since $0\leq g(t)-\abs{t}\leq1$, the boundedness of $e^{\, \alpha \mspace{2mu}\abs{X_2}}\chi_1(\Ha_S(\eta))$ is equivalent to that of $e^{\, \alpha \mspace{2mu} g(X_2)}\chi_1(\Ha_S(\eta))$. Thus, Proposition~\ref{sec:binding:exp_decay_hamiltonian} provides the required exponential localization for the weight $g(X_2)$. Moreover, $g'$ and $g''$ are bounded, the bosons are massive, and the additional term $(\eta-I\otimes\dG(K_1))^2$ strongly commutes with both $g(X_2)$ and $I\otimes\Nb$. Hence the argument of \cite[Lemma 31(iii)]{AC_rayleigh_scatt_1} yields the boundedness of $(I\otimes\Nb+1)e^{\alpha g(X_2)}\chi_1(\Ha_S(\eta))$, and therefore that of $B$. Gathering all this, we compute: 
    \begin{align*}
        e^{-\alpha \abs{X_2}} \mathcal{E}^{\mathrm{field}}_R (I \otimes \Nb +1)^{-1} B = \mathcal{E}^{\mathrm{field}}_R \chi_1(\Ha_S(\eta)), 
    \end{align*}
    and get that the norm of this operator is $o(1)$ as $R\to+\infty$ (Lemma \ref{sec:binding:lemma_error_field}). Hence one has, in operator norm: 
    \begin{align*}
       (\Ha_S^\ext(\eta)+i)^{-1} \mathcal{E}^{\mathrm{field}}_R (\Ha_S(\eta)+i)^{-1} \chi_1(\Ha_S(\eta)) \underset{R\to+\infty}{=} o(1).
    \end{align*}   
    We multiply \eqref{sec:binding:eq_lemma_estimate} by $\chi_1(\Ha_S(\eta))$ and combine the estimates to get, in operator norm: 
    \begin{align*}
        \left(\Gc_R\,(\Ha_S(\eta)+i)^{-1} - (\Ha_S^\ext(\eta)+i)^{-1}\Gc_R\right) \chi_1(\Ha_S(\eta)) \underset{R\to+\infty}{=} o(1).
    \end{align*}
    Since $\nop{(A+i)(A-z)^{-1}} \leq 2 \displaystyle \:  \frac{1+\abs{z}}{\abs{\Im z}}$ for any self-adjoint operator $A$, the previous estimate becomes: 
    \begin{align*}
        \left(\Gc_R\,(\Ha_S(\eta)-z)^{-1} - (\Ha_S^\ext(\eta)-z)^{-1}\Gc_R\right) \chi_1(\Ha_S(\eta)) \underset{R \, \to \, +\infty} {=}  o(1) \hspace{4pt}\frac{(1+\abs{z})^2}{\abs{\Im z}^{2}}.
    \end{align*}
\end{proof}

To prove assertion (i) of Theorem~\ref{sec:binding:theorem_grond_state_fibered_hamiltonian}, we shall ultimately show that smooth spectral cutoffs of $\Ha_S(\eta)$ below a suitable energy threshold are compact; the Helffer-Sjöstrand formula (see, e.g., \cite[Appendix A.2]{AC_rayleigh_scatt_1}) allows us to deduce from the preceding resolvent comparison the corresponding comparison for smooth spectral cutoffs:

\begin{proposition}\label{sec:binding:prop_comparison_with_extended_hamiltonian}
    Let $\chi\in\ciz(\R)$, $\chi_1 \in \ciz\big(\mspace{2mu}]-\infty,\Sigma(\eta)[\mspace{2mu}\big)$. Then one has, in operator norm:
    \begin{align*}
        \left(\Gc_R\,\chi(\Ha_S(\eta)) - \chi(\Ha_S^\ext(\eta))\Gc_R  \right)\chi_1(\Ha_S(\eta)) \underset{R\to+\infty}{=}  o(1).
    \end{align*}
\end{proposition}
\begin{proof}
    Let $\tilde{\chi}\in\ciz(\C)$ be an almost analytic extension of $\chi$ such that $\pder{\tilde{\chi}}{\cc{z}} = 0$ on $\R$ and $\abs{\partial_{\cc{z}} \tilde{\chi}} \leq c \abs{\Im(z)}^3$. Let $\mathcal{E}_R(z) := \Gc_R\,(\Ha_S(\eta)-z)^{-1} - (\Ha_S^\ext(\eta)-z)^{-1}\Gc_R$. One has:
    \begin{align*}
        \left(\Gc_R\,\chi(\Ha_S(\eta)) - \chi(\Ha_S^\ext(\eta))\Gc_R\right)\chi_1(\Ha_S(\eta)) = \frac{1}{2 \mspace{2mu} i \mspace{1mu}\pi}\int_\C \pder{\tilde{\chi}}{\cc{z}}(z) \: \mathcal{E}_R(z) \, \chi_1(\Ha_S(\eta)) \, d\cc{z} \wedge dz. 
    \end{align*} 
    Combining the previous lemma with the bound on $\partial_z\widetilde\chi$ and the compactness of $\supp\widetilde\chi$ in $\C$ (so that $|z|$ is bounded), we get: 
    \begin{align*}
        \nop{\pder{\tilde{\chi}}{\cc{z}}(z)\:\mathcal{E}_R(z) \: \chi_1\big(\Ha_S(\eta)\big)} \hspace{5pt} \underset{\substack{R\to+\infty \\ \Im z \to 0}}{=} o(1) \, O\big(\abs{\Im z}\big). 
    \end{align*}
    Using once more the compactness of $\supp\widetilde\chi$, and noting that the $o(1)$ term does not depend on $z$, we get the desired result:
    \begin{align*}
        \left(\Gc_R\,\chi(\Ha_S(\eta)) - \chi(\Ha_S^\ext(\eta))\Gc_R\right)\chi_1(\Ha_S(\eta)) \underset{R\to+\infty}{=}  o(1).
    \end{align*}
\end{proof}

The next step in the computations is to determine the threshold $E_0^{(1)}(\eta)$, which is such that for $\chi$ supported in $]-\infty,E_0^{(1)}(\eta)[$, one has $\chi(\Ha_S^\ext(\eta)) \ind_{[1,+\infty[}(I \otimes \Nb_\infty) = 0$, i.e. such that no asymptotically free boson can occur below it. This threshold is simply equal to $\inf \sigma(\Ha)+m$ in the non-fibered case (see \cite[Theorem 4.1]{acqftmpfh}) but this might not be true in the fibered case, so one has to decompose along each $I \otimes  \Nb_\infty = \ell$ sector to compute it. Indeed, the additional difficulty in the fibered case comes from the fact that the $(\eta - I \otimes \dG^\ext(K_1))^2$ term does not split linearly, whereas the $I \otimes \dG^\ext(\omega(K))$ term does. However, by the same method as in \cite[Section 2.5 and proof of Theorem 1.2]{moller_translation_invariant_nelson_model}, we find, denoting $P_0 := \ind_{\{0\}}(I \otimes \Nb_\infty)$ and $P_0^\perp := \ind_{[1,+\infty[}(I \otimes \Nb_\infty)$: 
\begin{equation}\label{sec:binding:lower_bound_ext_hamiltonian}
    \begin{aligned}
        \inf \sigma \hspace{-2pt}\left( \Ha_S^\ext(\eta)\big|_{\Ran\left(P_0^\perp\right)}\right) &= \underset{n \in \N^\ast}{\inf} \hspace{5pt} \underset{k \in \R^{2n}}{\inf} \hspace{5pt} E \hspace{-2pt}\left( \eta  - \sum_{j=1}^{n} k_1^j\right) + \sum_{j=1}^n \omega \hspace{-2pt}\left(k^j\right) \\
        &= \underset{k \in \R^2}{\inf} \, E(\eta - k_1) + \omega(k) \\
        &= \underset{k_1 \in \R}{\inf} \, E(\eta - k_1) + \omega(k_1,0) := E_0^{(1)}(\eta),
    \end{aligned}
\end{equation}
where $k \in \R^{2n}$ is written as $(k^1, \cdots, k^n)$. The first equality follows from the direct-integral decomposition on each $I\otimes\Nb_\infty=n$ sector \cite[Section 2.5]{moller_translation_invariant_nelson_model}; the second equality comes from the strict subadditivity of $\omega$, which implies that the infimum over $n$ is attained at $n=1$ (see \cite[(1.23)]{moller_translation_invariant_nelson_model}); the third one comes from the fact that the infimum of $\omega(k_1,k_2)$ over $k_2 \in \R$ is $\omega(k_1,0)$. We deduce from this that $E_0^{(1)}(\eta)$ is the desired threshold, and this yields the first claim of Theorem \ref{sec:binding:theorem_grond_state_fibered_hamiltonian}:

\begin{proof}[\fbox{Proof of Theorem \ref{sec:binding:theorem_grond_state_fibered_hamiltonian}(i)}]
    Let $\eta \in \R$ and let us denote $\tau(\eta) = \min\hspace{-2pt}\left(\Sigma(\eta), E_0^{(1)}(\eta)\right)$. 
    Let $\chi \in \ciz\big(\mspace{2mu}]-\infty,\tau(\eta)[\mspace{2mu}\big)$ and $\chi_1 \in \ciz\big(\mspace{2mu}]-\infty,\tau(\eta)[\mspace{2mu}\big)$ be real-valued, with $\chi_1 = 1$ on $\supp \mspace{1mu}\chi$. The previous proposition enables us to write, since $\tau(\eta) \leq \Sigma(\eta)$: 
    \begin{align*}
        \left(\Gc_R\,\chi(\Ha_S(\eta)) - \chi(\Ha_S^\ext(\eta))\Gc_R  \right)\chi_1(\Ha_S(\eta)) \underset{R\to+\infty}{=}  o(1).
    \end{align*}
    Since $\chi(\Ha_S(\eta)) \chi_1(\Ha_S(\eta)) = \chi(\Ha_S(\eta))$ and $\Gc_R^\ast \Gc_R = I$, one gets:
    \begin{equation}\label{sec:binding:first_eq_limit_R_localization}
        \chi(\Ha_S(\eta)) \underset{R\to+\infty}{=} \Gc_R^\ast \, \chi(\Ha_S^\ext(\eta)) \: \Gc_R  \, \chi_1(\Ha_S(\eta))  +  o(1).
    \end{equation}
    In order to write $\chi(\Ha_S^\ext(\eta))$ in a more convenient form, we
    observe that $\Ha_S^\ext(\eta)$ and $I\otimes\Nb_\infty$ strongly commute (by \cite[Theorem 3.5(iv) and Proposition 5.5(ii)]{Arai}), meaning that every spectral subspace of $I\otimes\Nb_\infty$ reduces
    $\Ha_S^\ext(\eta)$ by \cite[Proposition 1.47 and Section 1.9]{Arai} and
    therefore also reduces $\chi(\Ha_S^\ext(\eta))$ by
    \cite[Theorem 1.49]{Arai}. Since $\Ha_S^\ext(\eta)\big|_{\Ran\left(P_0^\perp\right)} \geq E_0^{(1)}(\eta)$ (see \eqref{sec:binding:lower_bound_ext_hamiltonian}) and $\supp \, \chi \subset ]-\infty,E_0^{(1)}(\eta)[$, this yields: 
    \begin{equation}\label{sec:binding:decomposition_chi_ext_hamiltonian}
        \chi(\Ha_S^\ext(\eta)) =  \chi(\Ha_S^\ext(\eta))\big|_{\Ran\left(P_0\right)} \oplus \: \chi\hspace{-2pt}\left(\Ha_S^\ext(\eta)\big|_{\Ran\left(P_0^\perp\right)}\right) 
        = \chi\hspace{-2pt}\left(\Ha_S^\ext(\eta)\big|_{\Ran\left(P_0\right)}\right) \oplus \: 0.
    \end{equation}
    Using $P_0 = I \otimes I \otimes P_\Omega$, where $P_\Omega$ is the orthogonal projection onto the vacuum, we compute $\Ha_S^\ext(\eta) \big|_{\Ran\left(P_0\right)} = \Ha_S(\eta) \otimes I_{\,\C \mspace{2mu}\Omega}$ and hence $\chi\hspace{-2pt}\left(\Ha_S^\ext(\eta)\big|_{\Ran\left(P_0\right)}\right) = \chi(\Ha_S(\eta)) \otimes I_{\,\C \mspace{2mu}\Omega}$ by Proposition \ref{app:op:standard_properties}(iii). Thanks to the canonical identification of Hilbert spaces $\mathscr{K} \otimes( \Hi_1 \oplus \Hi_2) \cong (\mathscr{K} \otimes \Hi_1) \oplus (\mathscr{K} \otimes \Hi_2) $ (\cite[Theorem 2.12]{Arai}), one computes: \begin{equation}\label{sec:binding:split_localization_extended_hamiltonian}
        \chi(\Ha_S^\ext(\eta)) = \left( \chi(\Ha_S(\eta)) \otimes I_{\,\C \mspace{2mu}\Omega} \right) \oplus 0 = \chi(\Ha_S(\eta)) \otimes \left( I_{\,\C \mspace{2mu}\Omega} \oplus 0_{(\C \mspace{2mu}\Omega)^\perp} \right) = \chi(\Ha_S(\eta)) \otimes P_\Omega.
    \end{equation}
    Let us define $\iota : L^2(\R_{\xd},\Fs) \to L^2(\R_{\xd}, \Fs \otimes \Fs)$ by $\iota \psi = \psi \otimes \Omega$. By the formula preceding \cite[Lemma 2.14]{acqftmpfh}, we obtain $P_0 \Gc_R = \iota \hspace{-2pt}\left(I \otimes \Gamma\big(j_{0,R}\big)\right)$, and taking the adjoint yields $\Gc_R^\ast P_0 = \left(I \otimes \Gamma\big(j_{0,R}\big)\right)^\ast \iota^\ast$. Furthermore, since $(I \otimes I \otimes P_\Omega)\iota = \iota$, we have $\iota^\ast (I \otimes I \otimes P_\Omega) \iota = \iota^\ast \iota = I$; combining this with the identity $(\chi(\Ha_S(\eta)) \otimes I) \mspace{2mu} \iota = \iota\, \chi(\Ha_S(\eta))$ derived from the definition of $\iota$, we deduce that $\iota^\ast \big(\chi(\Ha_S(\eta)) \otimes P_\Omega \big) \iota = \chi(\Ha_S(\eta))$. Together with \eqref{sec:binding:split_localization_extended_hamiltonian} and the identity $P_0 \, \chi(\Ha_S^\ext(\eta)) P_0 = \chi(\Ha_S^\ext(\eta))$, derived from \eqref{sec:binding:decomposition_chi_ext_hamiltonian}, the equality \eqref{sec:binding:first_eq_limit_R_localization} becomes:
    \begin{equation}\label{sec:binding:eq_localization_limit_compact_op}
    \begin{aligned}
        \chi(\Ha_S(\eta)) &\underset{R\to+\infty}{=}  \Gc_R^\ast \mspace{1mu} P_0\, \chi(\Ha_S^\ext(\eta)) P_0 \mspace{2mu} \Gc_R  \, \chi_1(\Ha_S(\eta))  +  o(1) \\
        &\underset{R\to+\infty}{=}\left(I \otimes \Gamma\big(j_{0,R}\big)\right)^\ast \chi\big(\Ha_S(\eta)\big) \left(I \otimes \Gamma\big(j_{0,R}\big)\right) \chi_1(\Ha_S(\eta))  +  o(1).      
    \end{aligned}
    \end{equation}
    Let $\Gamma_{0,R} := I \otimes \Gamma\big(j_{0,R}\big)$ and let us show that $\chi\big(\Ha_S(\eta)\big) \Gamma_{0,R} \, \chi_1\big(\Ha_S(\eta)\big)$ is compact. To see this, we factorize the operator as follows: 
    \begin{align*}
        \chi\big(\Ha_S(\eta)\big) \Gamma_{0,R} \chi_1\big(\Ha_S(\eta)\big) = \left(e^{\alpha \abs{X_2}} \chi\big(\Ha_S(\eta)\big)\right)^\ast e^{-\alpha \abs{X_2}}\, \Gamma_{0,R} \mspace{2mu} (\Ha_S(\eta)+i)^{-1/2} (\Ha_S(\eta)+i)^{1/2} \chi_1\big(\Ha_S(\eta)\big).
    \end{align*}
     We can write this factorization because $(\Ha_S(\eta)+i)^{1/2} \chi_1\big(\Ha_S(\eta)\big)$ is bounded by functional calculus, $\chi\big(\Ha_S(\eta)\big)e^{\alpha \abs{X_2}} \subset \left(e^{\alpha \abs{X_2}} \chi\big(\Ha_S(\eta)\big)\right)^\ast$ (\cite[Proposition 1.11]{Arai}), and $\Ran\hspace{-2pt}\left(e^{-\alpha \abs{X_2}}\right) = \D\hspace{-2pt}\left(e^{\alpha \abs{X_2}}\right)$. By Proposition \ref{sec:binding:exp_decay_hamiltonian}, we can choose $\alpha > 0$ such that $e^{\alpha \abs{X_2}} \chi\big(\Ha_S(\eta)\big)$, and thus its adjoint, are bounded. For this choice of $\alpha$, \cite[Lemma 34]{AC_rayleigh_scatt_1} ensures that $e^{-\alpha \abs{X_2}}\, \Gamma_{0,R} \mspace{2mu} (\Ha_S(\eta)+i)^{-1/2} = e^{-\alpha \abs{X_2}} \otimes \Gamma\bigl(j_{0,R}\bigr) \mspace{2mu} (\Ha_S(\eta)+i)^{-1/2} $ is compact, directly yielding the compactness of the full operator $\chi\big(\Ha_S(\eta)\big) \Gamma_{0,R} \, \chi_1\big(\Ha_S(\eta)\big)$. Since $\Gamma_{0,R}$ is bounded, \eqref{sec:binding:eq_localization_limit_compact_op} gives the compactness of $\chi\big(\Ha_S(\eta)\big)$ as an operator norm limit of compact operators. \cite[Theorem \cRM{13}.77]{RS} then gives the claim about the essential spectrum.
   \end{proof}

\bigbreak 

The second claim of Theorem \ref{sec:binding:theorem_grond_state_fibered_hamiltonian} follows by adapting
the proof of \cite{takaesu_binding} to the additional translation-invariant term $(\eta-P_1)^2$. As explained above, on each fiber only the transverse variable remains, and the potential $V(X_2)$ vanishes at infinity. This allows us to reproduce the usual single-particle binding argument.

\begin{proof}[\fbox{Proof of Theorem \ref{sec:binding:theorem_grond_state_fibered_hamiltonian}(ii)}]
We indicate how \cite{takaesu_binding} applies, using its notation. Let $\eta\in\R$ and define the translation-invariant part of
$\Ha_S(\eta)$ by:
\begin{align*}
\Ha_S^0(\eta) := -\partial_{\xd}^2\otimes I + I \otimes \left(\dG(\omega(K))+\big(\eta-\dG(K_1)\big)^2\right) + g \mspace{2mu}\Ha_{I,0} = -\partial_{\xd}^2\otimes I + I \otimes \dG_\eta + g \mspace{2mu}\Ha_{I,0}.
\end{align*} 
Both $\Ha_S^0(\eta)$ and the decomposition $\Ha_S(\eta) = \Ha_S^0(\eta) + V(X_2) \otimes I$ are well defined thanks to Theorem \ref{sec:ess:prop_fiber_straight_hamiltonian} and Proposition \ref{sec:prelim:domain_splitting_and_cores}(i). The operator $\Ha_S^0(\eta)$ is indeed translation invariant: the proof is the same as the one of Lemma \ref{sec:ess_spec:lemma_inv_translation}, except for the additional $\left(\eta-P_1\right)^2$ term which is translation invariant because $P_1$ is (thanks to the covariance of the functional calculus, \cite[Theorem 1.32]{Arai}). Having established this, we may apply the argument of
\cite{takaesu_binding}. In his notation, we set $\Ha_b=\dG_\eta$; with this choice, the proofs of \cite[Lemmas 2.1 to 2.4]{takaesu_binding} apply to
$\Ha_S(\eta)$ and yield:
\begin{align*}
    \inf\sigma(\Ha_S(\eta)) \leq \ipld{u}{\bigl(-\partial_{\xd}^2 + V(X_2)\bigr)u} + \inf\sigma\hspace{-2pt}\left(\Ha_S^0(\eta)\right). 
\end{align*}
We now take $u \in H^2(\R)$ to be a normalized real-valued ground state of
$-\partial_{\xd}^2+V(X_2)$ with eigenvalue $\varepsilon_0<0$, which exists because $V$ is compactly supported, non-positive and not identically zero (Assumption~\ref{hyp:potential}). As an $H^2(\R)$ function, it is bounded and so is its derivative because in one dimension, one has $H^2(\R) \hookrightarrow   H^1(\R) \hookrightarrow L^{\infty}(\R)$. As a consequence, all the hypotheses of
\cite[Remark 2.1]{takaesu_binding} are satisfied: $u$, $\partial_{\xd}u \in H^1(\R)$ are bounded, $\partial_{\xd}^2u=(V-\varepsilon_0)u$ is bounded since $V$ is, and so is $Vu$. Hence the conclusion of \cite[Corollary 2.5 and Remark 2.1]{takaesu_binding} applies and yields:
\begin{align*}
\inf\sigma(\Ha_S(\eta)) \leq \varepsilon_0 + \inf\sigma\hspace{-2pt}\left(\Ha_S^0(\eta)\right). 
\end{align*}
And \cite[Lemma 2.7]{takaesu_binding} holds because $V$ has compact support and yields $\inf\sigma\big(\Ha_S^0(\eta)\big)\leq\Sigma(\eta)$, so $\inf\sigma(\Ha_S(\eta)) < \Sigma(\eta)$ follows from the fact that $\varepsilon_0<0$. 
\end{proof} 

\bigbreak

Combining statements (i) and (ii) of Theorem \ref{sec:binding:theorem_grond_state_fibered_hamiltonian}, we deduce that a sufficient condition for the existence of a ground state for $\Ha_S(\eta)$ is $E_0^{(1)}(\eta) > E(\eta)$. Since $E_0^{(1)}(\eta)$ represents the lowest energy among all configurations with at least one asymptotically free boson, one expects this inequality to hold near the minimum of $E(\eta)$. Indeed, near the momentum $\eta_0$ at which the energy is minimal, emitting a boson costs at least its mass, whereas reducing the momentum of the bound subsystem yields little energy gain. At large total longitudinal momentum $\eta$, however, the energy decrease from $E(\eta)$ to $E(\eta-k_1)$ obtained by letting a freely propagating boson carry a momentum $k_1$ may compensate for the energy cost $\omega(k_1,0)$ of that boson. It is therefore physically conceivable that no completely bound configuration lies below $E_0^{(1)}(\eta)$. Since all configurations with one or more asymptotically free bosons have energy at least $E_0^{(1)}(\eta)$, this would imply $E(\eta)\geq E_0^{(1)}(\eta)$ for large longitudinal momentum $\eta$. \\
Mathematically, this can be seen as follows: in the proof of Proposition \ref{sec:ess:prop_no_gap_straight_spec}, we showed that $E(\eta)$ grows at least linearly when $\eta \to \pm \infty$, and at the beginning of Section \ref{sec:binding}, we explained that it has a global minimum at $\eta_0 \in \R$. Letting $k_1 = \eta-\eta_0$, we find $E_0^{(1)}(\eta) \leq E(\eta_0) + \sqrt{(\eta-\eta_0)^2+m^2}$, so that $E_0^{(1)}(\eta)$ has at most linear growth. Hence $E(\eta)$ and $E_0^{(1)}(\eta)$ seem to be of the same order of magnitude as $\eta \to \pm \infty$, so we cannot say a priori that $E(\eta) < E_0^{(1)}(\eta)$ for all $\eta \in \R$.  \\

However, we can prove that this inequality holds at $\eta_0$:
\begin{proof}[\fbox{Proof of Theorem \ref{sec:binding:theorem_grond_state_fibered_hamiltonian}(iii)}]
    Recall that $\eta_0$ is well defined because we proved that $\eta \mapsto E(\eta)$ is continuous and goes to $+\infty$ as $\abs{\eta} \to +\infty$ (see the proof of Proposition \ref{sec:ess:prop_no_gap_straight_spec}). By definition of $\eta_0$, for any $k_1 \in \R$, $E(\eta_0 - k_1) \geq E(\eta_0)$. Consequently, the infimum of $E(\eta_0 - k_1)$ over $k_1 \in \R$ is attained at $k_1 = 0$, as is the infimum of $\omega(k_1,0)$ over $k_1 \in \R$. Thus, we find $E_0^{(1)}(\eta_0) = E(\eta_0) + m$, as claimed. Finally, the equality $E(\eta_0) = \Sigma_S$ follows from Theorem \ref{sec:ess_spec:theorem_spectrum_straight_waveguide}.
\end{proof}

Together, assertions (i), (ii), and (iii) of Theorem \ref{sec:binding:theorem_grond_state_fibered_hamiltonian} yield the desired conclusion: $\Ha_S(\eta_0)$ admits a ground state. This is statement \ref{sec:binding:statement_a}, which is needed to prove that condition \ref{sec:binding:statement_b} is sufficient for binding.  \\
Before doing so in the next subsection, we establish the additional property $E(0) \leq E(\eta)$ for all $\eta \in \R$. This has been proved for the translation-invariant Nelson model and seems logical: the energy is minimized when the total momentum is. In our case, $\eta$ is only the total momentum along the $\xu$ axis, but, as will be proved below, the result still holds: as the direction of the waveguide is the natural direction of propagation, having no momentum along it should correspond to a low-energy configuration. \\
The proof essentially follows that of \cite[Theorem 3.2(5)]{thomas_thesis} (itself inspired by that of Leonard Gross \cite[Theorem 8]{eupgs}), with the difference that our Hamiltonian is not fibered along all directions but only along the longitudinal one, leaving an additional $-\partial^2_{\xd}+V(X_2)$ term that must be handled. The argument relies on positivity-preserving and positivity-improving semigroups (see Subsection~\ref{app:operator:positivity} for the relevant
background and auxiliary results) together with the $Q$-space
isomorphism (Proposition~\ref{app:fock:Q-space_isomorphism}) but, unlike alternative proofs (see, e.g., \cite[Corollary 2.148] {hiroshima2020feynman2}), does not use a functional-integral representation of the semigroup.\\

We begin by noticing that, using functional calculus (\cite[Lemma 1.6(i)]{Arai}):
\begin{equation}\label{sec:binding:inequality_op_norm_semigroup}
    \forall \eta \in \R, \: E(0) \leq E(\eta) \iff \forall \eta \in \R, \: \forall t > 0, \: \nop{e^{-t \mspace{2mu}\Ha_S(\eta)}} \leq \nop{e^{-t\mspace{2mu}\Ha_S(0)}}.
\end{equation}
Let $\eta \in \R$, $t>0$ be fixed in what follows. Using Theorem \ref{sec:ess:prop_fiber_straight_hamiltonian}, we split the Hamiltonian into the part that depends on $\eta$ and the one that does not: 
\begin{align*}
    \Ha_S(\eta) = \underbrace{(\eta-P_1)^2}_{:=M(\eta)} + \underbrace{\Ha_{\el,0} + \Ha_f   + g \mspace{2mu}\Ha_{I,0}}_{:= L} := M(\eta)+L.
\end{align*}
To exploit this splitting of the Hamiltonian, we can rely on the Trotter-Kato formula (\cite[Theorem \cRM{8}.31]{RS}). Indeed, $M(\eta) \geq 0$ is self-adjoint, $L$ is self-adjoint and lower semi-bounded (same argument as in the proof of Theorem \ref{prelim:th:self-adjointness_Ha}), and $M(\eta) + L = \Ha_S(\eta)$ has domain $\D(M(\eta)) \cap \D(L)$ (see again Theorem \ref{sec:ess:prop_fiber_straight_hamiltonian}); hence:
\begin{align*}
 e^{-t\mspace{2mu}\Ha_S(\eta)} = \slim{n\to+\infty} \left(e^{-t \mspace{1mu}M(\eta)/n} \: e^{-t \mspace{1mu} L/n}\right)^n,
\end{align*}
where $\mathrm{s-lim}$ denotes the limit in the strong operator topology. Consequently, in order to prove \eqref{sec:binding:inequality_op_norm_semigroup}, it suffices to prove that:
\begin{align*}
\forall n \in \N^\ast, \: \abs{\left(e^{-t \mspace{1mu}M(\eta)/n} \: e^{-t \mspace{1mu} L/n}\right)^n\psi} \leq \left(e^{-t \mspace{1mu}M(0)/n} \: e^{-t \mspace{1mu} L/n}\right)^n \abs{\psi}.
\end{align*}
Since $t>0$ is arbitrary, the previous inequality follows by induction from the one below. Iterating the bound is justified because the inequality itself implies that $e^{-t \mspace{1mu}M(0)} \: e^{-t \mspace{1mu} L}$ is positivity preserving: 
\begin{equation}\label{sec:binding:inequality_abs_value_semigroup}
\abs{\mspace{2mu} e^{-t \mspace{1mu}M(\eta)} \: e^{-t \mspace{1mu} L} \, \psi} \leq e^{-t \mspace{1mu}M(0)} \: e^{-t \mspace{1mu} L} \, \abs{\psi}.
\end{equation}
Such estimates are naturally related to positivity preserving semigroups, see, e.g., \cite[Section 5, Lemma 5.2(6)]{thomas_thesis}. To make sense of $\abs{\psi}$, we need the $Q$-space isomorphism (Proposition \ref{app:fock:Q-space_isomorphism} with $L^2\hspace{-2pt}\left(\Ry^2, \C\right) = L^2\hspace{-2pt}\left(\Ry^2,\R\right) + i \hspace{2pt} L^2\hspace{-2pt}\left(\Ry^2,\R\right)$) which enables us to view elements of the Fock space as functions on a suitable $L^2$ space. Recall that $y=i\mspace{1mu}\nabla_k$ denotes the one-boson configuration variable: in this representation, $K=-i\mspace{1mu}\nabla_y$ on $L^2(\R_y^2)$. \\
By covariance of the functional calculus (\cite[Theorem 1.32]{Arai}), the above computations may be carried out on the Hilbert space
\begin{align*}
   \Hi' := L^2\hspace{-2pt}\left(\R_{\xd} \times Q, d\xd \otimes \dP \right) = \int_{\R}^\oplus L^2(Q, \dP) \hspace{3pt} d\xd = \int_{\R}^\oplus U \Fs \hspace{3pt} d\xd ,
\end{align*}
and we omit the unitary operator $U$ in the following. As $\Hi'$ is well suited to discuss positivity improving operators, we can proceed to computations. As in \cite[Lemma A.7]{thomas_thesis}, one gets:
\begin{equation}\label{sec:binding:eq_expression_semigroup_m(eta)}
    e^{-t \mspace{1mu}M(\eta)} = \int_{\R} N_t(\xi) \hspace{2pt}  e^{\mspace{1mu}i \mspace{2mu}\eta \mspace{3mu}  \xi} \hspace{2pt} e^{ - \mspace{1mu }i \mspace{2mu}P_1 \mspace{2mu} \xi } \hspace{2pt} d\xi = \int_{\R} N_t(\xi) \hspace{2pt}  e^{\mspace{1mu}i \mspace{2mu}\eta \mspace{3mu}  \xi}  \left(I \otimes \Gamma\big(e^{-iK_1\xi}\big)\right) d\xi,
\end{equation}
where $N_t(\xi) = \frac{1}{\sqrt{4\pi t}}e^{\, -\, \xi^2/4t}$ and the integral is to be understood in the sense of a Bochner integral. For all $\psi \in \Hi'$, \cite[Lemma 5.2(7)]{thomas_thesis}, Proposition \ref{app:fock:Q-space_isomorphism}(i) (translations are unitary and preserve realness), \cite[Lemma 5.8 and 5.2(6)]{thomas_thesis} yield: 
\begin{equation}\label{sec:binding:eq_m(eta)_leq_M(0)}
    \begin{aligned}
        \abs{e^{-t \mspace{1mu}M(\eta)} \psi } \leq  \int_{\R} \abs{N_t(\xi) \hspace{2pt}  e^{\mspace{1mu}i \mspace{2mu}\eta \mspace{3mu}  \xi} \hspace{2pt} e^{ - \mspace{1mu }i \mspace{2mu}P_1 \mspace{2mu} \xi }\hspace{2pt} \psi} \hspace{2pt}  d\xi = \int_{\R} N_t(\xi) \hspace{1pt} \abs{\hspace{1pt} e^{ - \mspace{1mu }i \mspace{2mu}P_1 \mspace{2mu} \xi }\hspace{2pt} \psi} \hspace{2pt}  d\xi &\leq \int_{\R} N_t(\xi) \hspace{1pt} e^{ - \mspace{1mu }i \mspace{2mu}P_1 \mspace{2mu} \xi } \hspace{2pt} \abs{\psi} \hspace{2pt}  d\xi \\
        &= e^{-t \mspace{1mu}M(0)} \, \abs{\psi }.
    \end{aligned}
\end{equation}
Consequently, to prove \eqref{sec:binding:inequality_abs_value_semigroup}, it remains to prove that $e^{-t \mspace{1mu} L}$ is positivity preserving and to apply again \cite[Lemma 5.2(6)]{thomas_thesis}.

\begin{lemma}\label{sec:binding:L-pos-improv-semigroup}
Assume \textup{\ref{h4_field_real_valued_configuration_space}}. Then, on the configuration space, $L$ generates a positivity improving semigroup.
\end{lemma}
\begin{proof}
    We divide the proof into two steps: the analysis of the free Hamiltonian and the addition of the perturbation. The result follows from a combination of Proposition \ref{app:vector_valued:tensor_sum_pos_improving} and Corollary \ref{app:vector_valued:sum_pos_improving}.

\begin{enumerate}[label=(\roman*)]

\item It is well known that $-\partial_{\xd}^2$ generates the heat semigroup, which acts as a convolution by a strictly positive function (see, e.g., \cite[Section \cRM{10}.8, Example 5]{RS}), and is therefore positivity improving. Since $V(X_2)$ is a bounded multiplication operator, we have $\Q(-\partial_{\xd}^2) \subset \Q(V(X_2))$. Moreover, $V(X_2)$ is infinitesimally $-\partial_{\xd}^2$-bounded in the operator sense, and hence in the quadratic form sense (see \cite[Theorem \cRM{10}.18]{RS}). According to Corollary \ref{app:vector_valued:sum_pos_improving}, $-\partial_{\xd}^2+V(X_2)$ generates a positivity improving semigroup. \\
The operator $\omega(K)$ is self-adjoint, and since $\omega \geq m > 0$, it is injective. Moreover, on the configuration space, $\omega(K)=\omega(-i \mspace{1mu}\nabla_y)=\sqrt{-\Delta_y+m^2}$: since $-\Delta_y+m^2$ commutes with complex conjugation, so do its square root and the semigroup $e^{-t\mspace{1mu}\omega(-i \mspace{1mu}\nabla_y)}$. Thus, $e^{-t\mspace{2mu}\omega(-i \mspace{1mu}\nabla_y)} L^2(\Ry^2,\R)\subset L^2(\Ry^2,\R)$ and $\dG(\omega(-i \mspace{1mu}\nabla_y))$ generates a positivity improving semigroup by Proposition \ref{app:fock:Q-space_isomorphism}(ii). \\
By Proposition \ref{app:vector_valued:tensor_sum_pos_improving}, the operator 
\begin{align*}
\Ha_{\el,0}+\Ha_f := \big(-\partial_{\xd}^2 + V(X_2)\big) \otimes I + I \otimes \dG(\omega(-i \mspace{1mu}\nabla_y))
\end{align*}
generates a positivity improving semigroup. 
 
\item  It remains to add the perturbation $g \mspace{2mu} \Ha_{I,0}$. Let us denote $\phi_{\xd} := \phi\left(v_{(0,x_2)}\right)$. On the configuration space, $v_{(0,\xd)} = v\left(\cdot-(0,\xd)\right)$, which is real-valued by assumption. By Proposition~\ref{app:fock:Q-space_isomorphism}(iii), $(\phi_{\xd})_{\xd \in \R}$ defines a Gaussian process with mean $0$ and $\mathrm{Cov}\, (\phi_{\xd^1},\phi_{\xd^2}) = \frac{1}{2}\ipld{v_{\xd^1}}{v_{\xd^2}}$ for $\xd^1, \xd^2 \in \R$; hence it acts as a family of multiplication operators. Since a direct integral of multiplication operators is again a multiplication operator (see \cite[Lemma 6.2]{thomas_thesis}), $g \mspace{2mu}\Ha_{I,0}$ is a multiplication operator, which is known to be infinitesimally $(\Ha_{\el,0}+\Ha_f)$-bounded. Consequently, \cite[Theorem \cRM{10}.18]{RS} yields an infinitesimal form-bound, allowing us to apply Corollary \ref{app:vector_valued:sum_pos_improving} to conclude that $L$ generates a positivity improving semigroup.
\end{enumerate}
\end{proof}

Combining \eqref{sec:binding:eq_m(eta)_leq_M(0)} and the above lemma, one gets \eqref{sec:binding:inequality_abs_value_semigroup} and even the following result:

\begin{theorem}\label{sec:binding:eta_0_equal_0}
    Assume \textup{\ref{h4_field_real_valued_configuration_space}}. Then one has: 
    \begin{itemize}
        \item For all $\eta \in \R$, $E(0) \leq E(\eta)$, i.e. $0 \in \underset{\eta \, \in \, \R}{\argmin} \, E(\eta)$.
        \vspace{5pt}
        \item The Hamiltonian $\Ha_S(0)$ has a unique and almost everywhere strictly positive ground state on $\, \R_{\xd} \times Q$.
    \end{itemize}
\end{theorem}
\begin{proof}
    From the preceding discussion, we obtain \eqref{sec:binding:inequality_abs_value_semigroup} and, consequently, \eqref{sec:binding:inequality_op_norm_semigroup} on the $Q$-space associated with the configuration Fock space. By unitary equivalence and the covariance of the functional calculus \cite[Theorem 1.32]{Arai}, the bound \eqref{sec:binding:inequality_op_norm_semigroup} carries over to our Hamiltonians on the original Fock space, which establishes the first claim. \\
    To establish the second claim, it suffices to show that $\Ha_S(0)$ generates a positivity improving semigroup; the result then follows from the existence of a ground state (Theorem \ref{sec:binding:theorem_grond_state_fibered_hamiltonian}) combined with \cite[Theorem \cRM{13}.44]{RS}. The proof that $e^{-t\Ha_S(0)}$ is positivity improving closely follows that of Lemma \ref{sec:binding:L-pos-improv-semigroup}. By setting $\dG_0 := \dG(\omega(-i \mspace{1mu}\nabla_y)) + \dG(-i \mspace{2mu} \partial_{y_1})^2$, and substituting it for $\dG(\omega(-i \mspace{1mu}\nabla_y))$ in the argument, one verifies, just as in Theorem \ref{sec:ess:prop_fiber_straight_hamiltonian}, that $\Ha_{I,0}$ is infinitesimally $\dG_0$-bounded. Consequently, the problem is reduced to proving that $\dG_0$ itself generates a positivity improving semigroup. \\
    Thanks to \eqref{sec:binding:eq_m(eta)_leq_M(0)}, we find that $\dG(-i \mspace{2mu} \partial_{y_1})^2$ generates a positivity preserving semigroup: 
    \begin{align*}
    \forall t > 0,\, \Hi' \owns \psi \geq 0, \: e^{-t \mspace{1mu} M(0)} \,\psi \geq \abs{\mspace{1mu} e^{-t \mspace{1mu} M(\eta)} \, \psi}  \geq 0. 
    \end{align*}
    Again, as in Theorem \ref{sec:ess:prop_fiber_straight_hamiltonian}, $\dG(\omega(-i \mspace{1mu}\nabla_y))$ and $\dG(-i \mspace{2mu} \partial_{y_1})^2$ strongly commute so using \cite[Corollary 1.7(ii)]{Arai} we find that $\dG_0$ is self-adjoint and lower semi-bounded (in particular positive), so that the associated semigroup makes sense. The joint functional calculus for strongly commuting operators (\cite[Corollary 1.7]{Arai}) yields 
    \begin{align*}
        e^{\mspace{1mu} - \mspace{1mu} t \mspace{3mu}  \dG_0} = e^{- \mspace{1mu} t \mspace{1mu}  \left( \dG(\omega(-i \mspace{1mu}\nabla_y)) \mspace{3mu}  + \mspace{3mu}   \dG(-i \mspace{2mu} \partial_{y_1})^2 \right) } = e^{\mspace{1mu} -\mspace{1mu} t \mspace{3mu}   \dG(\omega(-i \mspace{1mu}\nabla_y))} e^{\mspace{1mu} -\mspace{1mu} t\mspace{3mu}  \dG(-i \mspace{2mu} \partial_{y_1})^2  },
    \end{align*}
    and the composition of two positivity preserving semigroups is positivity preserving. To show that the semigroup is positivity improving, we use again \cite[Theorem \cRM{13}.44]{RS} and show that the vacuum $\Omega$ is a strictly positive and non-degenerate ground state. Since $\dG_0 \geq 0$ and $\dG_0 \, \Omega = 0$, we find that $\Omega$ is a ground state. Let $\psi$ be any eigenvector of $\dG_0$ associated to $0$. One has: 
    \begin{align*}
        0 = \ip{\psi}{\dG(\omega(-i \mspace{1mu}\nabla_y))\psi} + \norm{\dG(-i \mspace{2mu} \partial_{y_1})\psi}^2 \geq \ip{\psi}{\dG(\omega(-i \mspace{1mu}\nabla_y))\psi} = \norm{\dG(\omega(-i \mspace{1mu}\nabla_y))^{1/2}\psi}^2. 
    \end{align*}
    Thus, $\dG(\omega(-i\mspace{1mu}\nabla_y))^{1/2}\psi=0$, and hence $\dG(\omega(-i\mspace{1mu}\nabla_y))\psi=0$. By \cite[Theorem 5.4] {Arai}, this implies that $\psi=\mu\mspace{2mu}\Omega$ for some $\mu\in\C$, so $\Omega$ is non-degenerate. The vacuum is strictly positive in the $Q$-space representation (Proposition~\ref{app:fock:Q-space_isomorphism}(iv)), so $\dG_0$ and hence $\Ha_S(0)$ generate positivity improving semigroups (in the $Q$-space representation). The conclusion follows again from \cite[Theorem \cRM{13}.44]{RS}.
\end{proof}

\begin{remark}
    Observe that the assumption that the form factor $v$ is real-valued on the configuration space is only needed to show that $\eta_0$ may be chosen equal to $0$, but is not required to prove the existence of a minimizer, which is the only ingredient needed in the next subsection to prove that condition~\ref{sec:binding:statement_b} implies binding. However, if one wants to gain information on the ground state of $\Ha_S(\eta_0)$ (such as uniqueness), one may choose $\eta_0 = 0$ so that $\dG(\omega(-i \mspace{1mu}\nabla_y)) + (\eta_0-\dG(-i \mspace{2mu} \partial_{y_1}))^2$ and hence $\Ha_S(\eta_0)$ generate positivity improving semigroups. This is not possible if $\eta_0 \neq 0$, because of the complex term $e^{\mspace{2mu} i \mspace{2mu} \eta_0 \mspace{2mu} \xi}$ in \eqref{sec:binding:eq_expression_semigroup_m(eta)} which prevents the positivity preserving argument. 
\end{remark}

%===================================================================
%        Rigorous justification of the binding condition
%===================================================================

\subsection{Rigorous justification of the binding condition}\label{sec:binding:rigorous_computations_binding}

As explained at the beginning of the section, having established assertion \ref{sec:binding:statement_a}—which corresponds to Theorem~\ref{sec:binding:theorem_grond_state_fibered_hamiltonian}—we are now in a position to prove that condition \ref{sec:binding:statement_b} implies binding. This implication is the main result of this section.

\begin{theorem}[Binding condition]\label{sec:binding:theorem_binding_condition} 
    Assume \textup{\ref{h3_field_rotation_invariance}} and let $\Sigma := \Sigma_C = \Sigma_S$. Let $\eta_0$ be a minimizer of $\eta \mapsto E(\eta)$ and $\Psi \in L^2(\R_{\xd},\Fs)$ be a normalized ground state of $\Ha_S(\eta_0)$. Let $V_{\eff,\Psi}$ be the effective potential defined by:
    \begin{align*}
        V_{\eff,\Psi} := \int_{\R}(V_C-V_S)(\, \cdot \,  ,\xd)\snfs{\Psi(\xd)} d\xd.
    \end{align*}
    For all real-valued and normalized $u \in \D\big(-\partial_{\xu}^2+V_{\eff,\Psi}(X_1)\big)$, one has the inequality: 
    \begin{align*}
          \inf\sigma(\Ha_C) \leq \ipld{u}{\left(-\partial_{\xu}^2+V_{\eff,\Psi}(X_1)\right)u} + \Sigma. 
    \end{align*}
    Consequently, if $-\partial_{\xu}^2+V_{\eff,\Psi}(X_1)$ has a negative-energy state, then $\Ha_C$ has discrete spectrum below the essential one, and hence has a ground state. Equivalently: 
    \begin{align*}
        \sigma \big(-\partial_{\xu}^2+V_{\eff,\Psi}(X_1)\big) \: \cap \hspace{3pt}  ]-\infty, 0[ \hspace{4pt}  \neq \, \emptyset \hspace{3pt} \Longrightarrow \hspace{3pt} \sdisc(\Ha_C) \neq \emptyset. 
    \end{align*}
    Hence the condition $\sigma \big(-\partial_{\xu}^2+V_{\eff,\Psi}(X_1)\big) \: \cap \hspace{3pt}  ]-\infty, 0[ \hspace{4pt}  \neq \, \emptyset$ ensures binding, i.e. ensures $\Sigma > \inf \sigma(\Ha_C)$.
\end{theorem}

\begin{proof}
    The idea behind the proof was explained at the beginning of Section~\ref{sec:binding}, so we focus here on the computations. Thanks to Propositions \ref{sec:prelim:domain_splitting_and_cores}(i) and \ref{sec:prelim:domain_splitting_and_cores}(ii), we can write:
    \begin{align*}
        \Ha_C = \Ha_S +  (V_C-V_S)(X) \hspace{0.5cm} \text{and } \:\: \D(\Ha_S) = \D(\Ha_C).
    \end{align*}
    We know that $\eta_0$ and $\Psi$ exist as defined thanks to Theorem \ref{sec:binding:theorem_grond_state_fibered_hamiltonian} and we know that $\Sigma_C = \Sigma_S := \Sigma$ from Proposition \ref{sec:ess_spec:prop_equality_ionization_thresholds}. Observe that $V_{\eff,\Psi}$ is bounded by $2 \norm{V}_\infty \norm{\Psi}^2$, so that it is infinitesimally $-\partial_{\xu}^2$-bounded. Then $\D\big(-\partial_{\xu}^2+V_{\eff,\Psi}(X_1)\big) = \D(-\partial_{\xu}^2) =  H^2(\R_{\xu})$, so let $u \in H^2(\R_{\xu})$ be real-valued and normalized; we introduce the following normalized trial state:
    \begin{align*}
        \Upsilon := \left( \int_{\R^2}^\oplus   \, e^{ \mspace{2mu} i \mspace{2mu} (\eta_0 - \dG(K_1)) \mspace{3mu} \xu } \: d\xu \mspace{1mu} d\xd  \right) u \otimes \Psi =  W^\ast \mspace{2mu} e^{\mspace{2mu} i \mspace{2mu} \eta_0 \mspace{1mu} X_1}  \mspace{2mu} ( u \otimes \Psi),
    \end{align*}
    where $W = \int_{\R^2}^\oplus \Gamma\big(e^{i K_1 \xu}\big) d\xu d\xd$ is the operator from the proof of Lemma \ref{sec:ess_spec:lemma_diagonalization_ptot}. Let us justify that $\Upsilon \in \D(\Ha_S)$, so let $\Phi = e^{\mspace{2mu} i \mspace{2mu} \eta_0 \mspace{1mu} X_1}  \mspace{2mu} ( u \otimes \Psi)$ and let us show that $ \Phi \in \D(W \mspace{2mu} \Ha_S \mspace{3mu} W^\ast)$ so that $W^\ast \Phi \in \D(\Ha_S)$. As in the proof of Theorem \ref{sec:ess:prop_fiber_straight_hamiltonian}, we find:
    \begin{align*}
        W \mspace{2mu} \Ha_S \mspace{3mu} W^\ast = (p_1 - P_1)^2 + (-\partial_{x_2}^2 + V(X_2))\otimes I + \Ha_f  +  g\int_{\R^2}^\oplus \phi\left(v_{(0,x_2)}\right) d\xu \, d\xd.
    \end{align*}
    Because $e^{\mspace{2mu} i \mspace{2mu} \eta_0 \mspace{1mu} X_1} u $ does not live in the Fock space and only depends on $\xu$, $\Psi \in \D(\Ha_S(\eta_0)) \subset H^2(\R_{\xd},\Fs) \cap L^2\hspace{-2pt}\left(\R_{\xd},\D\big(\dG(\omega(K))\big) \right)$ (by Theorem \ref{sec:ess:prop_fiber_straight_hamiltonian}) and does not depend on $\xu$, we have
    \begin{align*}
        \Phi \in \D\hspace{-2pt}\left((-\partial_{\xd}^2 + V(X_2)) \otimes I\right) \cap \D(\Ha_f).
    \end{align*}
   As in Theorem \ref{sec:ess:prop_fiber_straight_hamiltonian}, the interaction does not alter the domain, so the non-trivial point is to show $\Phi \in \D\hspace{-2pt}\left((p_1 - P_1)^2\right)$. One sees that the unitary operator $e^{\mspace{2mu} i \mspace{2mu} \eta_0 \mspace{1mu} X_1}$ only acts on $L^2(\R_{\xu})$ so that: 
    \begin{align*}
         e^{-\mspace{2mu} i \mspace{2mu} \eta_0 \mspace{1mu} X_1} \mspace{2mu}(p_1-P_1)^2 \mspace{3mu} e^{\mspace{2mu} i \mspace{2mu} \eta_0 \mspace{1mu} X_1} = (p_1+\eta_0-P_1)^2.
    \end{align*}
     Observe that $u \mspace{2mu} \Psi \in \D\hspace{-2pt}\left(p_1^2\right) \cap \D\hspace{-2pt}\left((\eta_0-P_1)^2\right)$ since $u \in H^2(\R_{\xu})$, $\Psi$ does not depend on $\xu$ and $ \Psi \in \D(\Ha_S(\eta_0)) \subset \D\hspace{-2pt}\left( (\eta_0-P_1)^2 \right)$ (Theorem \ref{sec:ess:prop_fiber_straight_hamiltonian}). Combining these observations with Proposition \ref{app:op:standard_properties}(ii) ($p_1$ and $P_1$ strongly commute) yields $u\mspace{1mu}\Psi \in  \D\hspace{-2pt}\left( e^{-\mspace{2mu} i \mspace{2mu} \eta_0 \mspace{1mu} X_1} \mspace{2mu}(p_1-P_1)^2 \mspace{3mu} e^{\mspace{2mu} i \mspace{2mu} \eta_0 \mspace{1mu} X_1} \right)$ and hence:
    \begin{align*}
        u\mspace{1mu}\Psi \in e^{\mspace{1mu} - \mspace{2mu} i \mspace{2mu} \eta_0 \mspace{1mu} X_1} \D\hspace{-2pt}\left(  \left(p_1-P_1\right)^2    \right). 
    \end{align*}
    So we have $\Phi \in  \D\left( W \mspace{1mu} \Ha_S \mspace{3mu} W^\ast  \right)$ and hence $\Upsilon \in \D(\Ha_S)$. We can now compute, using again Proposition \ref{app:op:standard_properties}(ii):
    \begin{align*}
        \ipldfs{u\mspace{2mu}\Psi}{(p_1 + \eta_0 - P_1)^2 \mspace{1mu} u \mspace{2mu} \Psi} &= \ipldfs{u\mspace{2mu}\Psi}{\left(p_1^2 + (\eta_0 - P_1)^2 + 2 p_1 (\eta_0-P_1)\right) \mspace{1mu} u \mspace{2mu} \Psi} \\
        &= \ipld{u}{-\partial_{\xu}^2 u} \snldfs{\Psi} + \nld{u}^2 \ipldfs{\Psi}{(\eta_0-P_1)^2\mspace{1mu}\Psi} \\
        &\phantom{=} + 2\ipld{u}{-i \mspace{2mu} \partial_{\xu} u }\ipldfs{\Psi}{(\eta_0-P_1)\Psi}.
    \end{align*}
    For $u \in H^1(\R_{\xu})$, one has $\Re \left(\int_\R \cc{u} \, u'\right)  = 0$. If $u$ is additionally  real-valued, as assumed, it means that $\ipld{u}{-i \mspace{2mu} \partial_{\xu} u } = 0$. For the other terms, we use that the variables separate to finally get: 
    \begin{align*}
        \ipldfs{\Upsilon}{\Ha_S \Upsilon} &= \ipldfs{u \mspace{2mu} \Psi}{e^{-\mspace{2mu} i \mspace{2mu} \eta_0 \mspace{1mu} X_1} \mspace{2mu} W \mspace{2mu} \Ha_S \mspace{3mu} W^\ast \mspace{3mu} e^{\mspace{2mu} i \mspace{2mu} \eta_0 \mspace{1mu} X_1} \mspace{2mu} u \mspace{2mu}\Psi} \\
        &= \ipld{u}{-\partial_{\xu}^2 u} \snldfs{\Psi} + \nld{u}^2 \ipldfs{\Psi}{\Ha_S(\eta_0)\mspace{1mu}\Psi}.
    \end{align*}
    It is easily seen that $e^{-\mspace{2mu} i \mspace{2mu} \eta_0 \mspace{1mu} X_1} \mspace{2mu} W \mspace{1mu} ( V_C - V_S)(X) \mspace{3mu} W^\ast \mspace{3mu} e^{\mspace{2mu} i \mspace{2mu} \eta_0 \mspace{1mu} X_1} \mspace{2mu} = ( V_C - V_S)(X)$, so that:
    \begin{align*}
        \ipldfs{\Upsilon}{( V_C - V_S)(X) \Upsilon } &= \ipldfs{\hspace{1pt} u \mspace{1mu} \Psi\hspace{1pt}}{\hspace{1pt}( V_C - V_S)(X) \hspace{2pt} u \mspace{1mu} \Psi } \\
        &=  \int_{\R^2}  \abs{\mspace{2mu} u(\xu)}^2 \mspace{2mu}(V_C-V_S) ( \xu  ,\xd)\snfs{\Psi(\xd)} \, d\xu \, d\xd \\
        &=  \int_{\R}  \abs{\mspace{2mu} u(\xu)}^2 \mspace{4mu}V_{\eff,\Psi}(\xu) \: d\xu  = \ipld{u\mspace{1mu}}{\mspace{1mu} V_{\eff,\Psi}(X_1) \: u }.
    \end{align*}
    Finally, it means: 
    \begin{align*}
        \ipldfs{\Upsilon}{\Ha_C \Upsilon} &= \ipldfs{\Upsilon}{\Ha_S \Upsilon} + \ipldfs{\Upsilon}{( V_C - V_S)(X) \, \Upsilon} \\
        &= \ipld{u}{-\partial_{\xu}^2 u} \snldfs{\Psi} + \nld{u}^2 \ipldfs{\Psi}{\Ha_S(\eta_0)\mspace{1mu}\Psi} +  \ipld{u\mspace{1mu}}{\mspace{1mu} V_{\eff,\Psi}(X_1) \: u } \\
        &= \ipld{u}{\left(-\partial_{\xu}^2+V_{\eff,\Psi}(X_1)\right)u} + \ipldfs{\Psi}{\Ha_S(\eta_0)\mspace{1mu}\Psi},
    \end{align*}
    since $u$, $\Psi$ are assumed to be normalized. Moreover, by assumption on $\eta_0$ and using that $\underset{\eta \, \in \, \R}{\min}\: E(\eta) = \Sigma$ (see Theorem \ref{sec:ess_spec:theorem_spectrum_straight_waveguide}) and that $\Psi$ is a ground state: 
    \begin{align*}
        \inf \sigma(\Ha_C) \leq \ipldfs{\Upsilon}{\Ha_C \Upsilon}  = \ipld{u}{\left(-\partial_{\xu}^2+V_{\eff,\Psi}(X_1)\right)u} + \Sigma.
    \end{align*}
    This is the claimed inequality. It is seen that if $\ipld{u}{\left(-\partial_{\xu}^2+V_{\eff,\Psi}(X_1)\right)u} < 0$, then $\inf\sigma(\Ha_C) <  \Sigma$ and $\Ha_C$ has non-empty discrete spectrum (thanks to the HVZ Theorem \ref{sec:ess_spec:HVZ_theorem_curved_waveguide}) and hence a ground state.  However, notice that we assumed that $u$ is real-valued, which seems to restrain the possibilities for the trial states. This is not a problem because the operator $-\partial_{\xu}^2+V_{\eff,\Psi}(X_1)$ commutes with the complex conjugation since the potential is real-valued, yielding, for $f = u + i v$:
    \begin{align*}
        \ipld{f}{(-\partial_{\xu}^2+V_{\eff,\Psi}(X_1))f} =  \ipld{u}{(-\partial_{\xu}^2+V_{\eff,\Psi}(X_1))u} + \ipld{v}{(-\partial_{\xu}^2+V_{\eff,\Psi}(X_1))v}.     
    \end{align*}
    Consequently, for any $f$ complex-valued such that $\ipld{f}{(-\partial_{\xu}^2+V_{\eff,\Psi}(X_1))f}<0$, we have either $\ipld{u}{(-\partial_{\xu}^2+V_{\eff,\Psi}(X_1))u} < 0$ or $\ipld{v}{(-\partial_{\xu}^2+V_{\eff,\Psi}(X_1))v} < 0$, with $u$, $v$ real-valued. So the existence of a negative-energy state for $-\partial_{\xu}^2+V_{\eff,\Psi}(X_1)$ ensures the existence of a real-valued negative-energy state for $-\partial_{\xu}^2+V_{\eff,\Psi}(X_1)$. This concludes the proof.
\end{proof}

% À FAIRE DANS CETTE SECTION

%%%%%

%%%%%% 

%================================================================
% --- APPENDICES ---
%================================================================
\appendix
\addtocontents{toc}{\protect\setcounter{tocdepth}{1}} %pour ne pas afficher les sous-sections des annexes 

% Annexe A
\section{Curvilinear coordinates and cut locus}
\label{app:curve}

This appendix is devoted to detailing the geometry of the waveguide, with an emphasis on the switch to curvilinear coordinates and the conditions ensuring the bijectivity of the latter. In particular, this leads us to introduce the notion of cut locus. \textit{Throughout this appendix, $\gamma$ denotes the $\mathscr{C}^4$, unit-speed, non-self-intersecting curve introduced in Subsection~\ref{subsec:pres_waveguide}.}

%===============================================================
%======== Curvilinear coordinates in the waveguide
%===============================================================

\subsection{Curvilinear coordinates in the waveguide}\label{app:curve:curv_coord}

Let $\gamma : \R \to \R^2$ be the curve described in Subsection \ref{subsec:pres_waveguide} and parameterized by its arc length $s$. The signed curvature $\kappa(s)$ governs the angular variation of the unit tangent vector $\dot{\gamma}(s)$: denoting by $\beta$ the angle between $\dot{\gamma}$ and the $x_1$-axis (clockwise-oriented), we have $\kappa(s) = \dot{\beta}(s)$. 

\begin{figure}[H]
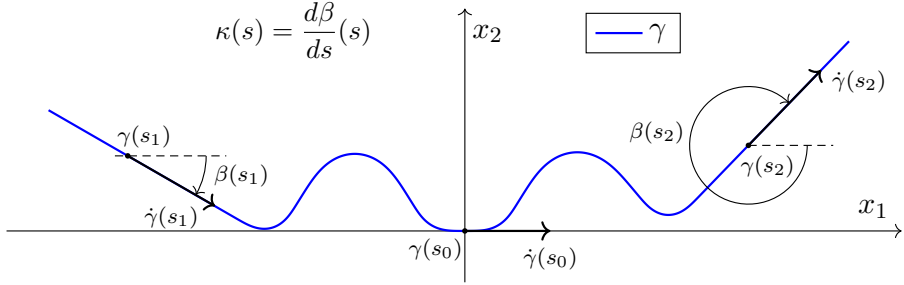

    \centering
    \resizebox{0.72\textwidth}{!}{
 }
    \caption{Definition of the signed curvature $\kappa$ of a curve $\gamma$.}
    \label{signed_curvature}
\end{figure}

\noindent Up to Euclidean transformations, $\gamma$ is uniquely determined by $\kappa$. Indeed, for all $s_0$, $s \in \R$, one sets $\beta(s,s_0) := \int_{s_0}^s \kappa(t) \, dt$ and $\gamma(s) = \gamma(s_0) + \left(\int_{s_0}^{s}\cos\beta(t,s_0) dt, - \int_{s_0}^{s}\sin\beta(t,s_0)dt\right) $.

\smallbreak

To parameterize the strip $\Omega^a$, we introduce the curvilinear coordinates $(s,u) \mapsto \gamma(s) + u N(s)$, where $s$ is the arc length, $u$ the normal distance, and $N(s) = (-\dot{\gamma}_2(s), \dot{\gamma}_1(s))$ the unit normal vector:

\begin{figure}[H]
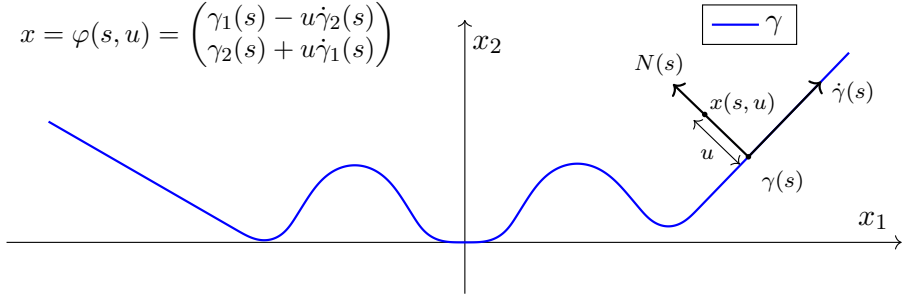

    \centering
    \resizebox{0.72\textwidth}{!}{
 }
    \caption{Local curvilinear coordinates.}
    \label{curvilinear_coordinates}
\end{figure}

With these variables, the bent strip is defined as $\Omega^a = \ens{\varphi(s,u)}{s\in \R,\, u \in ]-a,a[\,}$. To prevent self-intersections, each point in $\Omega^a$ must correspond to a unique pair $(s,u)$, meaning $\varphi$ must be a global diffeomorphism from $\R \times ]-a,a[$ to $\Omega^a$. While the condition $a \norm{\kappa}_{\infty} < 1$ ensures $\varphi$ is a local diffeomorphism (which can be verified using its Jacobian), global injectivity is enforced by Hypothesis~\ref{h3_geom_no_self_intersection}. Under this assumption, we can formally identify the curved waveguide $\Omega^a$ with the straight strip $\Omega^a_0 := \R\times]-a,a[$.\\
As an illustration, we construct the strip associated with the curve of the previous figures:

\begin{figure}[H]
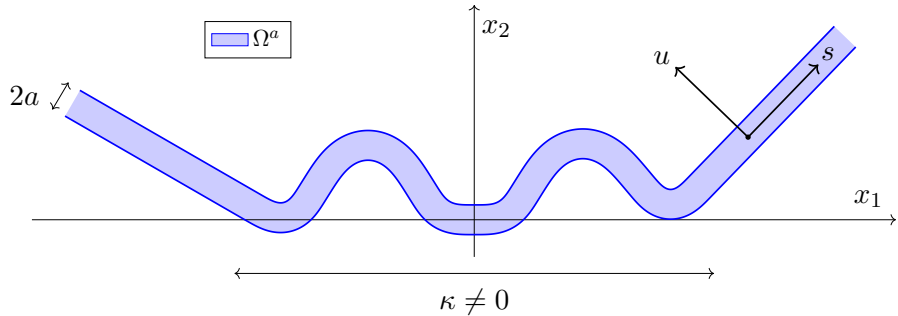

    \centering
    \resizebox{0.72\textwidth}{!}{
 }
    \caption{Infinite strip $\Omega^a$ of width $2a$ with a compactly supported curvature $\kappa$.}
    \label{curved_waveguide}
\end{figure}

%===============================================================
%======== Cut locus
%===============================================================

\subsection{Cut locus of a planar curve}\label{app:curve:cutlocus}

The finite-strip diffeomorphism $\varphi:\Omega_0^a\to\Omega^a$ is sufficient for hard waveguides, where the particle is strictly confined to $\Omega^a$. In our soft waveguide model, however, quantum tunneling allows the particle to explore the whole plane, requiring the free Laplacian $-\Delta$ to be defined on $L^2(\R^2)$. Consequently, we need curvilinear coordinates not only in the strip but on all of $\R^2$ (possibly up to a negligible set). The purpose of this subsection is precisely to construct the maximal coordinate domain $\mathcal{O}_\gamma$ introduced in Subsection~\ref{subsec:pres_waveguide}, and to justify the global straightening $U_{\mathrm{straight}} : L^2(\R^2)\longrightarrow L^2(\mathcal{O}_\gamma)$.

\smallbreak

The geometric obstruction to this global extension is that normal coordinates are not everywhere one-to-one: a point in the plane may have several shortest normal projections onto the curve. This leads naturally to the cut locus of $\gamma$. By removing this negligible set, one obtains a domain $\mathcal{O}_\gamma$ on which the normal-coordinate map $\varphi$ remains a diffeomorphism onto $\R^2\setminus\Cut(\gamma)$. We recall below the construction in the planar case, referring the reader to \cite{cutlocus} for a clear account, and to \cite{riemm_geo_diff, thesis_riemm_geom} for the general Riemannian framework.

\smallbreak 

The intuition behind this construction is simple. A point $x\in\R^2$ has
normal coordinates $(s,u)$ if it lies on the normal line to $\gamma$ at
$\gamma(s)$, namely if $x=\gamma(s)+uN(s)=\varphi(s,u)$.
The difficulty is that the same point may be reached from several normal lines, so that $\varphi$ may fail to be injective. We therefore follow each half-normal line only as long as the starting point $\gamma(s)$ remains a closest point of the curve to the point reached. Since $N(s)$ has norm $1$, this is equivalent to $\dist(\varphi(s,\pm u),\gamma)=u$. When this identity ceases to hold, the point reached from $\gamma(s)$ is closer to another part of the curve, and the normal coordinates have reached their obstruction. This leads to the definition of the cut radii:
\begin{equation}\label{cut_radius}
    c_\pm(s)
    := \sup\ens{u>0}{\dist\left(\varphi(s,\pm u),\gamma\right)=u}.
\end{equation}
For $0<u<c_\pm(s)$, the point $\varphi(s,\pm u)$ is still described by the
half-normal line issued from $\gamma(s)$. At $u=c_\pm(s)$, this description reaches its endpoint: keeping the half-normal line beyond this value would no longer give the desired one-to-one normal coordinates. If $c_\pm(s)=+\infty$, the corresponding half-normal line never meets such an obstruction.

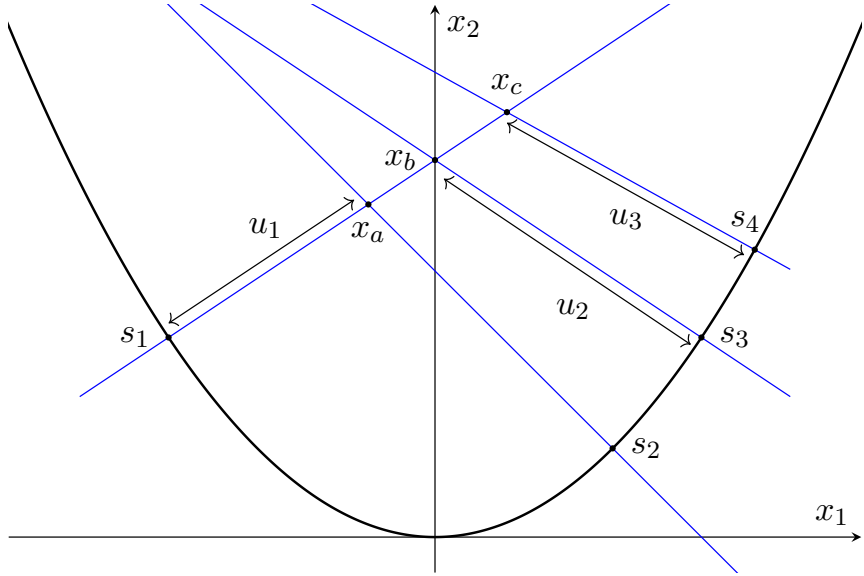
\begin{figure}[H]
    \centering
    \resizebox{0.68\textwidth}{!}{\begin{tikzpicture}
\begin{axis}[
  axis equal image,
  axis lines=middle,
  xmin=-1.2,xmax=1.2, ymin=-0.1,ymax=1.5,
  xlabel={$x_1$}, ylabel={$x_2$},
  domain=-1.5:1.5, samples=400,
  width=12cm, height=8cm,
  ticks=none
]
% Courbe y = x^2
\addplot[thick]{x^2};

% Normales en plusieurs points
\foreach \a in {-0.75,0.5,0.75, 0.9}{
  \pgfmathsetmacro{\ya}{\a*\a}
  \pgfmathsetmacro{\mn}{-1/(2*\a)}
  \addplot[blue,thin,domain=-1:1]{\mn*(x-\a)+\ya};
}

% Points d'intersection
\node[below, yshift = -2pt] at (-0.1875,0.9375) {$x_a$} ;
\fill (-0.1875,0.9375) circle[radius=1pt];
\node[left, xshift = -2pt] at (0,1.0625) {$x_b$};
\fill (0,1.0625) circle[radius=1pt];
\node[above, yshift = 2pt] at (0.2025,1.1975) {$x_c$};
\fill (0.2025,1.1975) circle[radius=1pt];

% Points de la courbe
\node[left, xshift = -2pt] at (-0.75,0.5625) {$s_1$};
\fill (-0.75,0.5625) circle[radius=1pt];
\node[right, xshift = 2pt] at (0.5,0.25) {$s_2$};
\fill (0.5,0.25) circle[radius=1pt];
\node[right, xshift = 2pt] at (0.75,0.5625) {$s_3$};
\fill (0.75,0.5625) circle[radius=1pt];
\node[above, yshift = 2pt, xshift = -3pt] at (0.9,0.81) {$s_4$};
\fill (0.9,0.81) circle[radius=1pt];

% Tracé des abscisses u
\draw[<->] (-0.75,0.6025) -- (-0.2275,0.9505) node [below, xshift = -1cm, yshift = -2pt, scale = 1]{$u_1$};
\draw[<->] (0.2025,1.1675) -- (0.87,0.795) node [left, xshift = -1cm, yshift = 0.4cm, scale = 1]{$u_3$};
\draw[<->] (0.025,1.010) -- (0.72,0.5425) node [left, xshift = -1cm, yshift = 0.4cm, scale = 1]{$u_2$};

\end{axis}
\end{tikzpicture}}
    \caption{Illustration of the cut locus construction. The point $x_a$ is reached from the normal line issued from $\gamma(s_1)$, so $x_a=\varphi(s_1,u_1)$. The point $x_b$ is reached by two shortest normal lines, issued from $\gamma(s_1)$ and $\gamma(s_3)$; hence $x_b\in\Cut(\gamma)$ and $c_+(s_1)=c_+(s_3)=u_2$. For $u>c_+(s_1)$, the point $\varphi(s_1,u)$ is no longer described by the half-normal line issued from $\gamma(s_1)$ in the one-to-one coordinate system; for instance, $x_c$ is described as $x_c=\varphi(s_4,u_3)$. In this example, the negative cut radii satisfy $c_-(s_1)=c_-(s_2)=c_-(s_3)=c_-(s_4)=+\infty$.}
    %%
    %RAJOUTER LA COURBE DANS LA LEGENDE
    %%
\end{figure}

The \textit{cut locus} of the curve is precisely the set of points reached at the cut radii:
\begin{equation}\label{def_cut_locus}
    \Cut(\gamma)
    =
    \ens{\varphi(s,c_+(s))}{s \in \R,\ c_+(s)<+\infty}
    \cup
    \ens{\varphi(s,-c_-(s))}{s \in \R,\ c_-(s)<+\infty}.
\end{equation}
This leads to the maximal open domain of normal coordinates:
\begin{align*}
    \mathcal{O}_\gamma
    :=
    \ens{(s,u)}{s \in \R, \; -c_-(s)<u<c_+(s)}.
\end{align*}
Thus $\mathcal{O}_\gamma$ is the maximal domain on which the normal-coordinate map
$\varphi(s,u)=\gamma(s)+uN(s)$ defines a global diffeomorphism
$\varphi : \mathcal{O}_\gamma \longrightarrow \R^2 \setminus \Cut(\gamma)$.
Note that Hypothesis \ref{h3_geom_no_self_intersection} precisely ensures that
$\Omega^a \subset \mathcal{O}_\gamma$, so that the strip supporting the waveguide
potential avoids the cut locus. \\
Finally, because the cut locus has Lebesgue measure zero in $\R^2$ (see, e.g., \cite[Proposition 2.1.16]{thesis_riemm_geom}), we have the canonical identification $L^2(\R^2 \setminus \Cut(\gamma)) \cong L^2(\R^2)$. This rigorously justifies the straightening operator $U_\text{straight}:L^2(\R^2)\longrightarrow L^2(\mathcal{O}_\gamma)$.

\bigbreak

We hope to have clarified the concept of cut locus for the reader. To end this appendix, we provide a proof of the claim we make in Hypothesis \ref{h4_geom_not_u_shaped}: 
\begin{proposition}\label{prop:appendixA_cut_radius_infty}
Assume \textup{\ref{h2_geom_compact_curvature}} and that the waveguide is not U-shaped. Then:
\begin{align*}
    c_\pm(s) \underset{\abs{s} \to + \infty}{\longrightarrow} + \infty.
\end{align*}
\end{proposition}
\begin{proof}
    The idea is that, because of the assumption of compact localization of the curvature, for $s$ large enough, the geometry is straight and does not locally induce any new cut locus, and hence does not change $c_\pm$. Consequently, since the distance between the branches grows to infinity, we should find that so do the cut radii maps. \\
    Suppose, for contradiction, that it is not the case. Then $ \exists (s_n)_{n \in \N}$, $M>0$ such that $s_n \underset{n \to +\infty}{\longrightarrow} +\infty$ and $c_+(s_n) \leq M$ for all $n$. For all $n$, set $x_n = \varphi(s_n,c_+(s_n))$; by definition of the cut locus and the cut radii maps, there exists another arc length $s'_n \neq s_n$ such that the distance from $\gamma(s'_n)$ to $x_n$ along its normal line is also $c_+(s_n)$, i.e. $x_n =\varphi(s'_n,\epsilon_nc_+(s_n))$ where $\epsilon_n = \pm 1$. One finds, again by definition: 
    \begin{align*}
        \abs{\gamma(s_n)-\gamma(s'_n)} \leq \abs{\gamma(s_n)-x_n} +\abs{x_n - \gamma(s'_n)} \leq 2c_+(s_n) \leq 2M < +\infty.
    \end{align*}
    To say that this is a contradiction (in the limit $n\to +\infty$) with the not-U-shaped assumption, we need to ensure that $\abs{s_n-s'_n} \underset{n \to +\infty}{\longrightarrow} +\infty$. Using Assumption~\ref{h2_geom_compact_curvature}, let $S>0$ such that $\supp \, \kappa \subset [-S,S]$, so that $\gamma([S,+\infty[) := L_+$ is a ray with a constant normal vector $N_+$. For $n$ large enough, $s_n \geq S$ so $x_n = \gamma(s_n) + c_+(s_n)N_+$ and one sees that $\gamma(s_n)$ is the unique orthogonal projection of $x_n$ on the ray $L_+$. Similarly, $\varphi(s'_n,\epsilon_n c_+(s_n)) = x_n$ corresponds to a distance $c_+(s_n)$ along the normal line to $\gamma(s'_n)$:  if $\gamma(s'_n)$ happened to be on $L_+$, it would also correspond to the orthogonal projection of $x_n$ on $L_+$, which we said was given by $\gamma(s_n)$. By uniqueness of the arc length it would mean $s'_n = s_n$ which is a contradiction, so $\gamma(s'_n) \notin  L_+$ i.e. $s'_n \leq S$ for all $n$. This yields $\abs{s_n-s'_n} \underset{n \to +\infty}{\longrightarrow} +\infty$ and the contradiction with the not-U-shaped assumption. \\
    The case $x_n = \varphi(s_n,-c_-(s_n))$ is similar, and so is the case $s_n \underset{n \to +\infty}{\longrightarrow} -\infty$, so this yields the result.
\end{proof}
 
% Annexe B
\section{Mathematical and physical choices for the model}
\label{app:model}

In this appendix, we specify our modeling choices: as explained in the introduction, \textit{our model being new, we begin by analyzing a simple instance of it} in order to lay solid foundations for subsequent developments. 

%%% ------------------------------------------------------------
%===============================================================
%%% ------------------------------------------------------------

\subsection{The small system - the waveguide}\label{section_small_system_presentation_waveguide}

We briefly recall the difference between hard and soft waveguides, since this choice is central to our model.

\begin{itemize}
    \item \textit{Hard waveguide.} The particle is strictly confined to the strip $\Omega^a$: the wavefunction vanishes outside the strip, and the dynamics is described by the Dirichlet Laplacian $-\Delta_D^{\Omega^a}$ on $L^2(\Omega^a)$. Equivalently, one may view this as the limit of an infinite potential well, equal to $0$ inside $\Omega^a$ and to $+\infty$ outside. See, e.g., \cite{pv_book}.

    \item \textit{Soft waveguide.} The confinement is produced by a finite attractive potential well. The particle is therefore not strictly confined to the strip and may tunnel outside it. In our setting, this is modeled by a bounded non-positive potential $V_\bullet$ on $\R^2$, supported in $\Omega^a$, and the one-particle dynamics is described by the Schrödinger operator $-\Delta+V_\bullet$ on $L^2(\R^2)$. See, e.g., \cite{pv_spsqv}.
\end{itemize}
\noindent We choose a soft waveguide rather than a hard one because \textit{it allows us to quantize the field on $\mathbb{R}^2$ within the standard framework} (plane-wave decomposition, creation--annihilation operators, etc.). By contrast, quantizing the field on a domain with boundary would require boundary-dependent modes and may introduce geometry-dependent renormalization issues, as illustrated by the Casimir effect; see, e.g., \cite{BordagMohideenMostepanenko2001}.

\medbreak

The modeling choices are then quite straightforward: since there are still few soft waveguide models, we built upon the rather comprehensive spectral analysis from \cite{pv_spsqv} and adopted its assumptions in their simplest form, which entails a planar waveguide and a compactly supported curvature.

\medbreak

We now explain the choice made in Hypothesis~\ref{hyp:potential}. We assume that the physical mechanism generating the attractive channel is uniform along the waveguide: each transverse section produces the same potential profile as a function of the normal distance to the reference curve. The waveguide is therefore modeled by replicating a fixed transverse profile along its longitudinal direction. For the straight guide, this amounts to setting $V_S(x_1,x_2)=V(x_2)$, so that the variables separate:
\begin{align*}
-\Delta+V_S(X) = -\partial_{x_1}^2\otimes I + I\otimes\bigl(-\partial_{x_2}^2+V(X_2)\bigr).
\end{align*}
Consequently, $\inf\sigma\bigl(-\Delta+V_S(X)\bigr) = \inf\sigma\bigl(-\partial_{x_2}^2+V(X_2)\bigr)$, i.e. the bottom of the spectrum of the straight guide is determined by the one-dimensional transverse profile.\\ For the curved guide, the same transverse profile is instead replicated along the reference curve: in curvilinear coordinates, we set $V_C(\phi(s,u))=V(u)$. The curved guide is thus obtained by bending the straight guide without changing its transverse profile, and may therefore be viewed as a geometric perturbation of the straight one, as is standard in the spectral analysis of quantum waveguides (see, e.g.,~\cite[Sections 5 to 8]{pv_spsqv} or~\cite[Section 1.4]{pv_book}).

%%% ------------------------------------------------------------
%===============================================================
%%% ------------------------------------------------------------

\subsection{Coupling with the quantum field}

\subsubsection{The physical modeling}

To model a charged particle propagating in a waveguide and interacting with a quantum field, the natural physical starting point is its coupling to the quantized electromagnetic field. The corresponding non-relativistic model is the Pauli--Fierz Hamiltonian. It is obtained by minimally coupling the particle to the quantized Maxwell field; we refer to \cite{Spohn2004} for the physical construction. In the Coulomb gauge, the relevant quantized degrees of freedom are the transverse components of the vector potential. For a spinless particle, the corresponding Hamiltonian has the formal form
\begin{align*}
    \Ha_{\mathrm{PF}}
    =
    \frac{1}{2 \mspace{2mu} m_p}\left(p+qA_\perp(x)\right)^2
    + q V_{\ext}(x)
    + \Ha_f,
\end{align*}
where $V_{\ext}$ is the external electrostatic potential, $A_\perp$ is the quantized transverse vector potential, $q$ is the charge, $m_p$ is the mass of the particle and $\Ha_f$ is the free field Hamiltonian.

Expanding the minimally coupled kinetic term yields the free kinetic energy, a term linear in the quantized vector potential $A_\perp(x)$, and a term quadratic in $A_\perp(x)$. \textit{In the present work, we neglect this quadratic term}. This is not physically realistic for the full Pauli--Fierz model, but it gives a first simplified model in which one can study the coupling between a quantum waveguide and a quantized electromagnetic field. We are thus led to the following formal Hamiltonian:
\begin{align*}
    \Ha = \frac{p^2}{2 \mspace{2mu} m_p} + qV_{\ext}(x) + \Ha_f + \frac{q}{m_p} p \cdot A_\perp(x)
    := \Ha_\el + \Ha_f + g \Ha_I,
\end{align*}
where  $\Ha_\el:=\frac{p^2}{2 \mspace{2mu} m_p}+qV_{\ext}(x)$ is the electronic Hamiltonian, $\Ha_f$ is the free field Hamiltonian, $g = q/m_p$ is a coupling constant, and $\Ha_I$ describes the interaction between the particle and the quantized field. At the formal level, the interaction has the following structure:
\begin{equation}\label{eq:interaction_hamiltonian_physics_version}
    \Ha_I
    = p \cdot A_\perp(x) = 
    \int_{\R^2}
    h_\omega(k)\,
    \big(p\cdot \varepsilon(k)\big)
    \left(
        e^{-ik\cdot x}a(k)
        +
        e^{ik\cdot x}a^\ast(k)
    \right)
    dk.
\end{equation}
Here $\varepsilon(k)$ is a polarization vector, $a(k)$ and $a^\ast(k)$ are the annihilation and creation operators, and $h_\omega$ contains the dependence on the dispersion relation together with an ultraviolet cutoff. This cutoff suppresses large momenta, as expected in an effective non-relativistic model, and also ensures that the interaction is mathematically well-defined.

\subsubsection{The mathematical modeling}

We now translate the physical model described above into the mathematical framework used in Section~\ref{sec:model}. Two standard simplifications are introduced, leading from the Pauli--Fierz model to the massive Nelson-type model studied in this paper. These simplifications are made in order to isolate the geometric and operator-theoretic features of the soft waveguide, leaving the additional difficulties of the full Pauli--Fierz model to subsequent work.

\smallbreak

First, we consider massive bosons rather than massless photons. For photons, the dispersion relation is $\omega(k)=|k|$, meaning that arbitrarily low-energy excitations are possible, leading to the usual infrared difficulties and delicate threshold behavior. We replace the massless dispersion relation by the massive one $\omega(k)=\sqrt{|k|^2+m^2}$, with $m>0$. This can be viewed as an infrared regularization of the physical model. It is also a natural first step, since the massless case is often treated by studying massive approximations and then taking the limit $m\to 0$; see, e.g., \cite{gsnrqed,hiroshima2019ground,analisa}.

Second, we simplify the interaction term. The physical interaction
\eqref{eq:interaction_hamiltonian_physics_version} has the formal structure
\begin{align*}
    \Ha_I
    =
    \int_{\R^2}
    \left(
        W(k)^\ast \otimes a(k)
        +
        W(k) \otimes a^\ast(k)
    \right)
    dk,
    \qquad
    W(k)=h_\omega(k)\,(p\cdot\varepsilon(k))\,e^{-ik\cdot x}.
\end{align*}
Thus it is already very close to a Segal field operator, except that the form
factor $W(k)$ is not scalar-valued: it takes values in operators on the particle space $\Hi_{\el}$, because it contains the momentum operator $p=-i\nabla$. \\
In the present paper, we keep this field-operator structure but simplify the operator-valued form factor. Instead of the momentum-dependent operator $W(k)$, which contains $p=-i\nabla$, we consider, for each $k$, a simplified operator $v(k)$ on $\Hi_{\el}=L^2(\Rx^2)$, namely the multiplication operator by the function $x\mapsto v_x(k)$. With this assumption, the interaction becomes a direct integral of Segal field operators:
\begin{align*}
\Ha_I 
  &= \int_{\R^2}^{\oplus}\phi(v_x)\,dx =  \frac{1}{\sqrt{2}}
    \int_{\R^2}^{\oplus} \,
    \overline{a(v_x)+a^\ast(v_x)}\,dx  
  && \text{in the operator sense,} \\
&= \int_{\R^2}^{\oplus} \int_{\R^2} 
      \big( \cc{v_x(k)}\, a(k) + v_x(k) \, \ac(k) \big)\, dk \, dx 
  && \text{in the quadratic form sense.}
\end{align*}

This type of simplification is standard in the mathematical study of particle-field Hamiltonians. For a detailed discussion of possible hypotheses on the interaction term, and of how they relate to the Pauli--Fierz Hamiltonian, we refer to \cite[Sections~1.1, 1.3 and 1.6]{qednrp}. See also \cite[Section~10.3]{AC_rayleigh_scatt_1} for a discussion of the dipole approximation and its relation to linear field couplings, \cite[Section 1.1]{moller_translation_invariant_nelson_model} for a review of the NRQED models and \cite[Section~14.5]{Arai} for a general presentation of particle-field interaction models.

%%%%   À FAIRE
%%%%   À FAIRE

%%%%   À FAIRE
%%%%   À FAIRE
 
% Annexe C
\section{Fock spaces and second quantization}
\label{app:fock}

\noindent We recall only the elementary Fock-space notation and properties used
in the paper.

\subsection{Symmetric Fock space and field operators}\label{app:fock:field}

Let $\Hi$ be a separable complex Hilbert space and, for $n\geq1$, let $S_n$ denote the orthogonal symmetrizer on the $n$-fold tensor product of $\Hi$.
We use the notation:
\begin{align*}
    \Hi^{\mspace{2mu}\otimes \mspace{2mu} n}:=\bigotimes_{i=1}^{n}\Hi,
    \qquad
    \Hi^{\mspace{2mu} \otimes_s \mspace{2mu} n}:=\operatorname{Ran}S_n
    =\sideset{}{_s}\bigotimes_{i=1}^{n}\Hi.
\end{align*}
\noindent By convention, the empty symmetric tensor product (from $i=1$ to $0$) is
$\Hi^{\otimes_s0}:=\C$. The symmetric Fock space over $\Hi$ is:
\begin{align*}
    \Fs(\Hi):=\bigoplus_{n=0}^{\infty}
    \sideset{}{_s}\bigotimes_{i=1}^{n}\Hi
    =\bigoplus_{n=0}^{\infty}\Hi^{\mspace{2mu} \otimes_s \mspace{2mu} n}.
\end{align*}
\noindent Equivalently:
\begin{align*}
    \Fs(\Hi)
    =\left\{(\psi^{(n)})_{n\in\N}\ \middle|\
    \psi^{(n)}\in\Hi^{\mspace{2mu} \otimes_s \mspace{2mu} n},\
    \sum_{n=0}^{\infty}\norm{\psi^{(n)}}^2<\infty\right\}.
\end{align*}
\noindent The vector $\Omega:=(1,0,\ldots)\in\Fs(\Hi)$ is called the vacuum. For a dense subspace $\mathcal C\subset\Hi$, the finite-particle subspace is:
\begin{align*}
    \Fs^\fin(\mathcal C):=
    \operatorname{Span}\left\{\Omega,\,
    f_1\otimes_s\cdots\otimes_s f_n
    \,\middle|\, n\geq1,\ f_i\in\mathcal C\right\}.
\end{align*}
\noindent See \cite[Section 5.2]{Arai} for the bosonic Fock space and its particle sectors.

\smallbreak

For $f\in\Hi$, first define the creation operator
$a_{\mathrm{fin}}^\ast(f)$ on $\Fs^\fin(\Hi)$ by:
\begin{align*}
    \bigl(a_{\mathrm{fin}}^\ast(f)\psi\bigr)^{(0)}:=0,
    \qquad
    \bigl(a_{\mathrm{fin}}^\ast(f)\psi\bigr)^{(n)}
    :=\sqrt n\,S_n\bigl(f\otimes\psi^{(n-1)}\bigr)
    \quad(n\geq1).
\end{align*}
\noindent Let $a_{\mathrm{fin}}(f)$ be the restriction of
$\bigl(a_{\mathrm{fin}}^\ast(f)\bigr)^\ast$ to $\Fs^\fin(\Hi)$.
Both are closable; $\ac(f)$ and $a(f)$ denote their closures, which satisfy $\ac(f)=a(f)^\ast$.

\noindent On $\Fs^\fin(\Hi)$, they satisfy the canonical commutation relations (CCR):
\begin{align*}
    [a(f),\ac(g)]=\ip{f}{g}_{\Hi}I,
    \qquad
    [a(f),a(g)]=[\ac(f),\ac(g)]=0.
\end{align*}
\noindent See \cite[Section 5.7]{Arai} for the creation and annihilation operators and
the CCR.

\noindent The Segal field operator
\begin{align*}
    \phi(f):=\frac{1}{\sqrt{2}}\,\overline{a(f)+\ac(f)}
\end{align*}
is the normalized self-adjoint part of $a(f)$ and represents the quantized
bosonic field smeared with $f$; see \cite[Section 5.10]{Arai}.

\subsection{Second quantization}\label{app:fock:second_quantization}

 Let $\Hi_1$ and $\Hi_2$ be two separable Hilbert spaces. For a contraction $Q:\Hi_1\to\Hi_2$, its second quantization is:
\begin{align*}
    \Gamma(Q):=\bigoplus_{n=0}^{\infty} \,Q^{\mspace{2mu} \otimes_s \mspace{1mu} n}
    :\Fs(\Hi_1)\longrightarrow\Fs(\Hi_2),
\end{align*}
\noindent where $Q^{\mspace{2mu} \otimes_s \mspace{1mu} n}$ denotes the restriction of $Q^{\mspace{2mu} \otimes \mspace{2mu} n}$ to the
symmetric $n$-particle sector and $Q^{\mspace{2mu} \otimes_s \mspace{2mu}0}=I_{\C}$. Thus $\Gamma(Q)$ applies $Q$ to every boson and fixes the vacuum. It satisfies $\Gamma(Q_2 \, Q_1)=\Gamma(Q_2)\Gamma(Q_1)$; moreover, $\Gamma(Q)$ is an isometry or a unitary operator whenever $Q$ is so; see \cite[Section 5.4]{Arai}.

\smallbreak

If $B$ is self-adjoint on $\Hi$, its differential second quantization is the self-adjoint operator defined sectorwise by:
\begin{align*}
    \left.\dG(B)\right|_{\, \Hi^{\mspace{2mu} \otimes_s \mspace{2mu} n}}
    =\sum_{j=1}^{n}
      I^{\mspace{2mu}\otimes \mspace{2mu}(j-1)}\otimes B\otimes I^{\mspace{2mu} \otimes \mspace{2mu}(n-j)},
    \qquad
    \dG(B)\Omega=0.
\end{align*}
\noindent It is the field observable associated with the one-particle observable $B$:
on the $n$-particle sector, it sums the action of $B$ over the $n$ bosons.

\noindent The boson number operator is:
\begin{align*}
    \Nb:=\dG(I),
    \qquad
    (\Nb\psi)^{(n)}=n \mspace{1mu} \psi^{(n)}.
\end{align*}
\noindent See \cite[Section 5.3]{Arai}. We shall also use:
\begin{align*}
    \Gamma\hspace{-2pt}\left(e^{\mspace{2mu}i\mspace{2mu} t \mspace{1mu}B}\right)=e^{\mspace{2mu}i \mspace{2mu} t \mspace{2mu}\dG(B)},\qquad
    \Gamma(U)\mspace{2mu} \dG(B) \mspace{2mu}\Gamma(U)^\ast=\dG(UBU^\ast),\qquad
    \Gamma(U) \mspace{2mu}\phi(f) \mspace{2mu}\Gamma(U)^\ast=\phi(Uf),
\end{align*}
\noindent where $t\in\R$ and $U$ is unitary; see
\cite[Section 5.4, Theorems 5.8--5.9, and Section 5.16, Theorem 5.32]{Arai}.

\subsection{Extended Fock space}\label{app:fock:extended}

The tensor product $\Fs(\Hi)\otimes\Fs(\Hi)$ is the extended Fock space: its first factor describes the nearby bosons and its second factor the
asymptotically free bosons. The natural unitary operator
\begin{align*}
    U:\Fs(\Hi\oplus\Hi)\longrightarrow\Fs(\Hi)\otimes\Fs(\Hi)
\end{align*}
\noindent is characterized by:
\begin{align*}
    U\Omega=\Omega\otimes\Omega,
    \qquad
    U\ac(h_0,h_\infty)U^\ast
    =\ac(h_0)\otimes I+I\otimes\ac(h_\infty).
\end{align*}
\noindent Under this identification, the differential second quantization on the two
Fock factors is:
\begin{align*}
    \dG^\ext(B):=\dG(B)\otimes I+I\otimes\dG(B).
\end{align*}
\noindent In particular:
\begin{align*}
    \Nb_0:=\Nb\otimes I,
    \qquad
    \Nb_\infty:=I\otimes\Nb,
    \qquad
    \Nb^\ext=\Nb_0+\Nb_\infty.
\end{align*}
\noindent See \cite[Section 5.22, in particular Theorems 5.40 and 5.42]{Arai} for
the natural identification and these transformation properties.

\smallbreak

Let $j_0,j_\infty$ be bounded operators satisfying
$j_0^\ast j_0+j_\infty^\ast j_\infty=I$ and set:
\begin{align*}
\begin{gathered}
    j:\Hi\longrightarrow\Hi\oplus\Hi,
    \qquad jh=(j_0\mspace{2mu}h,j_\infty\mspace{2mu} h),\\
    \Gc(j):=U \mspace{2mu} \Gamma(j):
    \Fs(\Hi)\longrightarrow\Fs(\Hi)\otimes\Fs(\Hi).
\end{gathered}
\end{align*}
Then $j$ and $\Gc(j)$ are isometries; in particular,
$\Gc(j)^\ast \,\Gc(j)=I$. See, e.g., \cite[Section 2.6]{AC_rayleigh_scatt_1}.

\smallbreak

For the particular localization used in Section~\ref{sec:binding}, work
in the boson configuration representation $\Hi=L^2(\R_y^2)$. Choose
real-valued $j_0,j_\infty\in C^\infty(\R^2;[0,1])$ with
$j_0^2+j_\infty^2=1$, $j_0=1$ on $\{|y|\leq1\}$ and $j_0=0$ on
$\{|y|\geq2\}$. For $R>0$, define:
\begin{align*}
    j_R:= \bigl( j_{0,R}, \, j_{\infty,R} \bigr):= \bigl(j_0(Y/R), \, j_\infty(Y/R)\bigr),
    \qquad j_R:\Hi\longrightarrow\Hi\oplus\Hi.
\end{align*}
The corresponding localization operator is:
\begin{align*}
    \Gc(j_R):=U \mspace{2mu}\Gamma(j_R):
    \Fs(\Hi)\longrightarrow\Fs(\Hi)\otimes\Fs(\Hi).
\end{align*}
It is an isometry, hence $\Gc(j_R)^\ast\,\Gc(j_R)=I$. The further relations
involving $\Nb$, $\Nb^\ext$ and $\Gc(j_R)$ that are needed in
Section~\ref{sec:binding} are proved there in
Lemma~\ref{sec:binding:lemma_N_properties_and_bounds}.

\subsection{\texorpdfstring{The $Q$-space representation}{The Q-space representation}}

The $Q$-space isomorphism realizes bosonic Fock space as an $L^2$-space
over a Gaussian probability space. It turns real Segal fields into Gaussian
multiplication operators and equips Fock space with the order structure
needed for positivity arguments. We refer to
\cite[Theorem 6.1]{thomas_thesis} and to
\cite[Sections 5.30.8 and 5.31]{Arai}.

\begin{proposition}[$Q$-space isomorphism]
\label{app:fock:Q-space_isomorphism}
Let $\Hi_\R$ be a real Hilbert space and $\Hi_\C=\Hi_\R+i\Hi_\R$ its
complexification. Then there exist a probability space
$(Q,\mathscr Q,\mathbb P)$ and a unitary operator
$U:\Fs(\Hi_\C)\to L^2(Q,\mathscr Q,\mathbb P)$. In addition, the following
properties hold:
\begin{enumerate}[label=(\roman*)]
\item If $V$ is a contraction on $\Hi_\C$ with
$V\Hi_\R\subset\Hi_\R$, then $U\Gamma(V)U^\ast$ is positivity preserving.
\item Assume $\omega\geq0$ is self-adjoint and injective. If
$e^{-t\omega}\Hi_\R\subset\Hi_\R$ for all $t>0$, then
$Ue^{-t\dG(\omega)}U^\ast$ is positivity improving.
\item If $v\in\Hi_\R$, then $U\phi(v)U^\ast:=\Phi(v)$ acts as multiplication
by a normally distributed random variable with mean $0$ and variance
$\frac{1}{2}\norm{v}^2$. In addition, $(\Phi(v))_{v\in\Hi_\R}$ is a Gaussian process
with mean $0$ and covariance
$\operatorname{Cov}\bigl(\Phi(f),\Phi(g)\bigr)=\frac12\ip{f}{g}_{\Hi_\C}$
for $f,g\in\Hi_\R$.
\item The vacuum is mapped to the constant function equal to $1$, i.e.
$U\Omega=1$.
\end{enumerate}
\end{proposition} 
% Annexe D
\section{About vector-valued functions}
\label{app:vector-valued}

Throughout this appendix, let $\Hi$ be a separable Hilbert space. We tackle the subject of spaces of square-integrable functions taking values in $\Hi$.

%==========================================================================
%           Constant fiber direct integrals
%==========================================================================
\subsection{Constant fiber direct integrals}

Let $(Q,\mu)$ be a measure space. We essentially refer the reader to \cite[Sections \cRM{2.4}, \cRM{13}.16]{RS} and \cite[Sections 2.7, 2.8 and 3.11]{Arai} for the standard framework. 
In this paper, we systematically identify the tensor product with the constant fiber direct integral:
\begin{equation*}
    L^2(Q,d\mu) \otimes \Hi \cong L^2(Q,d\mu; \Hi) = \int_{Q}^\oplus \Hi \, d\mu,
\end{equation*}
realized through the unitary isomorphism $f \otimes g \mapsto f(\cdot) \mspace{1mu}g$. Consequently, for a self-adjoint operator $A$ on $\Hi$ and a measurable function $f$, we identify the corresponding operators (see \cite[Theorem 3.16]{Arai}):
\begin{equation*}
    I \otimes A \cong \int_{Q}^\oplus A \, d\mu \qquad \text{and} \qquad f(X) \otimes I \cong \int_Q^\oplus f(x) \, d\mu(x).
\end{equation*}
Thanks to these unitary equivalences, we will go back and forth between these formalisms without further explanation, and naturally treat any multiplication operator $f(X)$ on $L^2(Q,d\mu)$ as acting on $L^2(Q,d\mu; \Hi)$.

\bigbreak 
A particular non-trivial case arises when considering the graph norm of a self-adjoint operator $A$ on $\Hi$. Its domain $\D(A)$ is a Hilbert space when endowed with the norm $\norm{\cdot}^2_A = \norm{\cdot}^2_\Hi + \norm{A\cdot}^2_\Hi$. This leads to the following space:

\begin{definition}\label{sec:prelim:L2_domain}
Let $A$ be a self-adjoint operator on $\Hi$. We define $L^2\big(Q, d\mu; \D(A)\big)$ as the space of measurable functions $\psi: Q \to \D(A)$ such that $\int_{Q} \left(\norm{\psi(x)}^2_\Hi + \norm{A\psi(x)}^2_\Hi\right) d\mu(x) < + \infty$.
\end{definition}

\begin{proposition}   
The space $L^2\big(Q, d\mu; \D(A)\big)$ continuously and densely embeds into $L^2(Q,d\mu; \Hi)$ through the canonical injection $\iota$. Its range is exactly the domain of the operator $\int_{Q}^\oplus A \,d\mu $, and $\iota$ is a unitary operator between $L^2\big(Q, d\mu; \D(A)\big)$ and $\D\hspace{-2pt}\left(\int_{Q}^\oplus A \,d\mu \right)$ endowed with its graph norm.
\end{proposition}
\begin{proof}
Continuity follows from $\norm{\cdot}_\Hi \leq \norm{\cdot}_A$, density follows from the fact that the domain of a self-adjoint operator is dense. As $\iota$ and $A: \big( \D(A), \norm{\cdot}_A \big) \to \Hi$ are continuous, all the measurability conditions for the first inclusion with the range hold. For the reverse inclusion, we observe that the map $J: \mathcal{G}(A) \to \big( \D(A), \norm{\cdot}_A \big)$, defined by $J(f,Af) = f$ (where $\mathcal{G}(A)$ is the graph of $A$), is isometric, hence continuous. The measurability requirements then hold for the reverse inclusion. One sees that $\iota$ is isometric as shown above, surjective; hence it is unitary. 
\end{proof}

\begin{remark}
    In the following, the domain of $\int_{Q}^\oplus A \,d\mu $ will be denoted by $L^2(Q,d\mu; \D(A))$ by abuse of notation, implicitly understood through the canonical injection. Combining what has been said above and Proposition \ref{app:op:prop_tensor_domain_hilbert_space}, one can state:
    \begin{equation}
        L^2(Q,d\mu) \otimes \D(A) \cong \D(I \otimes A) \cong \D\hspace{-2pt}\left(\int_{Q}^\oplus A \,d\mu \right) \cong L^2\big(Q, d\mu; \D(A)\big).
    \end{equation}
\end{remark}

%==========================================================================
%           Vector-valued Sobolev spaces
%==========================================================================
\subsection{Vector-valued Sobolev spaces}

For the definition of vector-valued Sobolev spaces, we refer the reader to \cite{trace_vectorvalued_sobolev} or \cite[Section 3.1]{QuasiLinearParabolicPb}. In this subsection, we recall the canonical identifications used throughout the paper.

We first state the identification for an arbitrary open set. Let $\mathcal{O}$ be an open subset of $\R^n$. For any $m \in \N$, one has the following unitary isomorphism of Hilbert spaces (see \cite[Theorem 12.7.1]{aubin2000applied}):
\begin{align*}
    H^m(\mathcal{O}) \otimes \Hi \cong H^m(\mathcal{O},\Hi),
\end{align*}
where the isomorphism is induced by the standard one from $L^2(\mathcal{O}) \otimes \Hi \cong L^2(\mathcal{O},\Hi)$, discussed in the previous subsection.

In the particular case $\mathcal{O}=\R^n$, this abstract identification can be explicitly related to the domains of the usual differential operators. Let $p=-i\nabla$ be the momentum operator on $L^2(\R^n)$ and set $p^2=-\Delta$. When working on $\R^n$, we endow $H^m(\R^n)$ with the graph norm of $(p^2)^{m/2}$ and $H^m(\R^n,\Hi)$ with the graph norm of $(p^2)^{m/2}\otimes I$. By applying Proposition~\ref{app:op:prop_tensor_domain_hilbert_space} to the operators $p$ and $(p^2)^{m/2}$, along with these graph norm conventions, we obtain the following operator-theoretic form of vector-valued Sobolev spaces:
\begin{align*}
    H^1(\R^n) \otimes \Hi &\cong \left(\D(p\otimes I),\norm{\cdot}_{p\otimes I} \right) \cong H^1(\R^n,\Hi) \\
    H^m(\R^n) \otimes \Hi &\cong \left(\D\hspace{-2pt}\left((p^2)^{\frac{m}{2}}\otimes I\right), \norm{\cdot}_{(p^2)^{m/2}\otimes I}\right) \cong H^m(\R^n,\Hi),
\end{align*}
for $n \geq 2$. In particular, we identify $p\otimes I$ with the usual distributional derivative $p \psi = -i\nabla\psi$ on $L^2(\R^n,\Hi)$, understood in the sense of vector-valued distributions (see, e.g., \cite[Section III.1.1]{QuasiLinearParabolicPb}).
\begin{remark}
    Unless stated otherwise, we will not distinguish $L^2(Q, d\mu)\otimes \Hi$ and $L^2(Q,d\mu; \Hi)$, or $H^m(\mathcal{O})\otimes \Hi$ and $H^m(\mathcal{O},\Hi)$, and will write $=$ instead of $\cong$. Furthermore, we shall not distinguish between $p$ and $p\otimes I$, using either notation depending on context.
\end{remark}
\subsubsection{A computational property}

We denote $A \cdot x = \sum_{i=1}^n A_i \, x_i$ to write multi-parameter unitary groups compactly.

\begin{proposition}\label{app:func:derivation_unitary_direct_int}
    Let $A = (A_1, \dots, A_n)$ be an $n$-tuple of strongly commuting self-adjoint operators on $\Hi$, $j \in \llbracket 1,n\rrbracket$, and $\psi \in \D(p_{x_j}) \cap  L^2\hspace{-2pt}\left(\R^n, \D(A_j)\right)$. Then $\Lambda := \left( \int_{\R^n}^\oplus e^{- i A \cdot x} dx \right) \psi \in \D(p_{x_j})$ with:
    \begin{align*}  
         \partial_j \Lambda &= \left( \int_{\R^n}^\oplus e^{- i A \cdot x} dx \right) \partial_j \psi - i \left( \int_{\R^n}^\oplus e^{-iA\cdot x}\mspace{2mu}A_j \, dx \right) \psi = \left( \int_{\R^n}^\oplus e^{- i A \cdot x} dx \right) \big( \partial_j - i \, (I \otimes A_j)\big) \, \psi.
    \end{align*}
    Additionally, $\Lambda \in L^2\hspace{-2pt}\left(\R^n, \D(A_j)\right)$.
\end{proposition}
\begin{proof}
    For almost all $x \in \R^n$, define $g_j(x) = e^{-i A \cdot x} \partial_j\psi(x) - i e^{-iA\cdot x}A_j \psi(x)$. To show that $g_j$ is the distributional derivative of $\Lambda$, let $f \in \ciz(\R^n)$. By \cite[Theorem E.1(ii)]{Arai}, we reduce to the scalar case by taking the inner product with $\Phi \in \D(A_j)$:
    \begin{align*}
        \int_{\R^n} \iph{\Lambda(x)}{\Phi}\partial_j f(x) \, dx  = \int_{\R^n} \iph{e^{-i A \cdot x} \psi(x)}{\Phi} \partial_j f(x) \, dx = \int_{\R^n} \iph{ \psi(x)}{e^{i A \cdot x}\mspace{2mu}\Phi \mspace{2mu}\partial_j f(x)} \, dx.  
    \end{align*}
    Let $h: x \mapsto f(x) e^{i A \cdot x} \Phi$, which is strongly differentiable with $\partial_j h(x) = e^{i A \cdot x} \Phi \partial_j f(x) + i f(x) e^{iA \cdot x} A_j \Phi $ by Stone's theorem. Since $\Phi \in \D(A_j)$ and $f \in \ciz(\R^n)$, we find $h \in \D(p_{x_j})$. Consequently, integrating by parts yields:
    \begin{align*}
        \int_{\R^n} \iph{\Lambda(x)}{\Phi}\partial_j f(x) \, dx
        &= \int_{\R^n}\iph{\psi(x)}{\partial_jh (x)} dx - i \int_{\R^n} \iph{\psi(x)}{ e^{iA \cdot x} A_j  \Phi} f(x) \, dx \\
        &= - \int_{\R^n}\iph{g_j(x)}{\Phi} f(x)\, dx. 
    \end{align*}
    By \cite[Theorem E.1(ii)]{Arai} and the density of $\D(A_j)$, we get the desired equality in the distribution sense. The fact that $g_j \in L^2(\R^n,\Hi)$ follows from its definition and $\psi \in \D(p_{x_j}) \cap  L^2\left(\R^n,  \D(A_j)\right)$, and the last claim follows again from Stone's theorem.
\end{proof} 
% Annexe E
\section{Operator theory}
\label{app:operator}

Throughout this section, let $\Hi$, $\Hi_1$ and $\Hi_2$ be separable Hilbert spaces. 

\smallbreak

We frequently rely on the joint functional calculus for strongly commuting self-adjoint operators (see, e.g., \cite[Section 1.8.4]{Arai} or \cite[Section 5.5.2]{Sch}). We summarize some useful operational rules in the following proposition.

\begin{proposition}\label{app:op:standard_properties}
    Let $A$ and $B$ be self-adjoint operators on $\Hi$.
    \begin{enumerate}[label=(\roman*)]
        \item Lower bounds: If $A$ and $B$ strongly commute and are positive, then $\norm{A\psi} \leq \norm{(A+B)\psi}$ for all $\psi \in \D(A) \cap \D(B)$. If they are only lower semi-bounded (with respective lower bounds $E_A, E_B$), then:
        \[ \norm{A\psi} \leq \norm{(A+B)\psi} + \left(\abs{E_A+E_B}+\abs{E_A}\right)\norm{\psi}. \]
        \item Square domains: If $A$ and $B$ strongly commute, then $\D(A^2) \cap \D(B^2) \subset \D(AB)$ and $(A+B)^2 = A^2 + 2AB + B^2$ on $\D(A^2) \cap \D(B^2) \subset \D\hspace{-2pt}\left((A+B)^2\right)$.
        \item Tensor products: For any Borel function $f$, $f(A \otimes I) = f(A) \otimes I$. If $A$ and $B$ strongly commute, so do $A \otimes I$ and $B \otimes I$. Furthermore, if $A$ and $B$ are lower semi-bounded, then $A \otimes I + B \otimes I = (A+B) \otimes I$. (The same holds symmetrically for $I \otimes A$).
        \item Unitary conjugation: For any unitary operator $U$, $U E^A U^\ast$ is the projection-valued measure of $U A U^\ast$.
        \item Multiplication operators: Self-adjoint multiplication operators always strongly commute.
    \end{enumerate}
\end{proposition}

\begin{proof}
    (i) For positive operators, $A+B$ is self-adjoint (see \cite[Corollary 1.7(ii)]{Arai}). By the joint functional calculus, $\norm{A\psi}^2 = \int \lambda_1^2 \, dE^{A,B}_{\psi}(\lambda_1,\lambda_2) \leq \int (\lambda_1 + \lambda_2)^2 \, dE^{A,B}_{\psi}(\lambda_1,\lambda_2) = \norm{(A+B)\psi}^2$. The lower semi-bounded case follows by applying this to the positive operators $A-E_A$ and $B-E_B$. \\
    (ii) The claims on the domains follow from the joint functional calculus and the scalar inequalities $\lambda_1^2\lambda_2^2 \leq \frac{1}{2}(\lambda_1^4 + \lambda_2^4)$ and $(\lambda_1+\lambda_2)^4 \leq 8(\lambda_1^4 + \lambda_2^4)$. Strong commutativity enables the expansion of $(A+B)^2$. \\
    (iii) The first two claims follow from the fact that the projection-valued measures for $A\otimes I$ and $B\otimes I$ are $E^A \otimes I$ and $E^B \otimes I$ (see \cite[Theorem 3.5(ii)]{Arai}). For the sum, $A \otimes I + B \otimes I$ and $(A+B) \otimes I$ are both self-adjoint (the first one thanks to \cite[Proposition 3.5(ii) and Corollary 1.7(ii)]{Arai}) and agree on $\D(A+B) \alten \Hi$, which is a core for $(A+B) \otimes I$. Taking the closure preserves the inclusion, and since a self-adjoint operator has no proper self-adjoint extension, equality holds. \\
    (iv) This follows from Stone's formula and $(U A U^\ast -z)^{-1} = U(A-z)^{-1}U^\ast$ for all $z \in \rho(A)$ (see the proof of \cite[Theorem 1.32]{Arai}). \\
    (v) By functional calculus (see \cite[Example 5.3]{Sch}), $e^{-itf(X)} = e^{-itf}(X)$. These are multiplications by bounded functions and thus commute, yielding the strong commutativity (see \cite[Proposition 1.48(iii)]{Arai}).
\end{proof}

We now recall a result regarding perturbations of tensor product operators.

\begin{proposition}\label{app:op:tensor_prod_relative_boundedness}
    Let $A$ and $B$ be self-adjoint operators such that $B$ is $A$-bounded with relative bound $< 1$. Then $A \otimes I + B \otimes I = (A+B) \otimes I$ (the same holds symmetrically for $I \otimes A$).
\end{proposition}
\begin{proof}
    By \cite[Proposition 3.6(i)]{Arai}, $B\otimes I$ is $A\otimes I$-bounded with the same bound. Thus, $A \otimes I+B\otimes I$ is self-adjoint by the Kato-Rellich theorem. Since $\D(A+B) = \D(A)$, both operators agree on $\D(A)\alten\Hi$, which is a core for both $A\otimes I+B\otimes I$ and $(A+B)\otimes I$. Taking the closure on this common core yields the equality.
\end{proof}

\smallbreak

Continuing with the consequences of relative boundedness, the following standard lemma allows us to control operators of the form $(B-z)^{-1}A$:

\begin{proposition}\label{r_B A bounded if A B bounded}
    Let $A, B$ be self-adjoint operators such that $A$ is $B$-bounded. Then, for all $z \in \rho(B)$, the operator $(B-z)^{-1}A$ is bounded on $\D(A)$ and extends to a bounded operator on $\Hi$ with operator norm $\norm{A(B-\cc{z})^{-1}}$.
\end{proposition}
\begin{proof}
    The $B$-boundedness implies that $A(B-\cc{z})^{-1}$ is a bounded operator on $\Hi$ for all $z \in \rho(B)$. We apply \cite[Proposition 2.21(a)]{Amr} to the bounded operator $(B-\cc{z})^{-1}$ whose range is in $\D(B) \subset \D(A)$: taking the adjoints yields that $((B-\cc{z})^{-1})^\ast A^\ast \subset (A(B-\cc{z})^{-1})^\ast$. Therefore, $((B-\cc{z})^{-1})^\ast A^\ast$ is densely defined and has a bounded closed extension, which means its closure is bounded with norm $\norm{(A(B-\cc{z})^{-1})^\ast} = \norm{A(B-\cc{z})^{-1}}$. 
\end{proof}

%==========================================================================
%           Direct integrals and Fibering
%==========================================================================

Beyond standard algebraic manipulations, strong commutativity allows for simultaneous diagonalization. In our context, we frequently rely on direct integral decompositions to fiber the Hamiltonian. 

\begin{proposition}\label{app_op:simultaneous_diagonalize_direct_integral}
    Let $A$ and $B$ be strongly commuting self-adjoint operators. Assume there exists a unitary operator $U$, a separable Hilbert space $\Hi '$, and a measure space $(Q,\mu)$ such that $U$ diagonalizes $A$, i.e. transforms it into a multiplication operator on $L^2(Q,\mu; \Hi ') = \int_Q^\oplus \Hi ' \, d\mu$:
    \begin{align*}
        UAU^\ast = \int_{Q}^{\oplus} \eta \, d\mu(\eta).
    \end{align*}
    Then there exists a measurable map $Q \ni \eta \mapsto B(\eta)$ with values in the self-adjoint operators on $\Hi '$ such that:
    \begin{align*}
        UBU^\ast = \int_{Q}^{\oplus} B(\eta) \, d\mu(\eta).
    \end{align*}
\end{proposition}
\begin{proof}
    We want to use \cite[Theorem \cRM{13}.85(b)]{RS} on $(UBU^\ast + i)^{-1}$. Since $A$ and $B$ strongly commute, so do $UAU^\ast$ and $UBU^\ast$ (by Proposition \ref{app:op:standard_properties}(iv)) and so do $g(UAU^\ast)$ and $UBU^\ast$ for all measurable real-valued $g$ (see \cite[Theorem 1.36]{Arai}). In particular, this holds if $g$ is bounded. \\
    Using the equivalence (ii) $\iff$ (iv) in \cite[Proposition 5.15]{Sch}, we get that $g(UAU^\ast)$ and $(UBU^\ast + i)^{-1}$ commute for all bounded real-valued $g$. Upon writing $g = \Re g + i \Im g$, this extends to complex-valued $g$. Thus, $(UBU^\ast + i)^{-1}$ commutes with all bounded decomposable operators whose fibers are multiples of the identity. By \cite[Theorem \cRM{13}.84]{RS}, $(UBU^\ast + i)^{-1}$ is a bounded decomposable operator, and applying \cite[Theorem \cRM{13}.85(b)]{RS} concludes the proof.
\end{proof}

\noindent The spectrum of such direct integrals is given by the following result (see \cite[Theorem 2.2]{spectrum_direct_integral}):

\begin{proposition}\label{app:op:prop_spectrum_direct_integrals}
    Let $\R^d \ni \eta \mapsto A(\eta)$ be a strongly continuous map from $\R^d$ to the self-adjoint operators on $\Hi$, and let $A = \int_{\R^d}^{\oplus} A(\eta) \,d\eta$ be the associated self-adjoint operator on $\int_{\R^d}^{\oplus} \Hi \, d\eta$. Assume in addition that the operators $A(\eta)$ have a common domain $\D$ independent of $\eta \in \R^d$. Then:
    \[ \sigma(A) = \overline{\bigcup_{\eta \in \R^d} \sigma\hspace{-1pt}\left(A(\eta)\right)}. \]
\end{proposition}
\begin{proof}
    The hypotheses of \cite[Theorem 2.2]{spectrum_direct_integral} reduce to the ones stated here in the self-adjoint case (see in particular their remark below the theorem regarding the resolvent condition).
\end{proof}
We now discuss a useful unitary isomorphism to determine the domain of tensor product operators:

\begin{proposition}\label{app:op:prop_tensor_domain_hilbert_space}
    Let $A$ be an unbounded closed densely defined operator on $\Hi_1$. The identity map on the algebraic tensor product:
    \begin{align*}
        U: \left(\D(A)\alten\Hi_2, \norm{\cdot}_{A\otimes I}\right) &\longrightarrow (\D(A), \norm{\cdot}_A) \otimes \Hi_2 \\
        \psi &\longmapsto \psi
    \end{align*}
    extends to a unitary operator between $\left( \D(A\otimes I), \norm{\cdot}_{A\otimes I} \right)$ and $(\D(A), \norm{\cdot}_A) \otimes \Hi_2$. The proposition holds symmetrically for $I \otimes A$. \\
    Consequently, if $A$ is a self-adjoint operator, one has $\Q(A \otimes I) \cong \left(\D(|A|^{1/2}), \norm{\cdot}_{|A|^{1/2}}\right) \otimes \Hi_2$.
\end{proposition}
\begin{proof}
    Since $A$ and $A \otimes I$ are closed, $(\D(A), \norm{\cdot}_A)$ and $(\D(A\otimes I), \norm{\cdot}_{A\otimes I})$ are Hilbert spaces. Moreover, $\D(A) \alten \Hi_2$ is a core for $A \otimes I$. By definition of the Hilbert tensor product and of a core:
    \begin{align*}
        \left(\D(A\otimes I), \norm{\cdot}_{A\otimes I} \right) &= \overline{\D(A)\alten\Hi_2}^{\norm{\cdot}_{A\otimes I}} \\ 
        (\D(A), \norm{\cdot}_A) \otimes \Hi_2 &= \overline{\D(A)\alten\Hi_2}^{\norm{\cdot}_{\D(A)\otimes \Hi_2}}
    \end{align*}
    Both norms agree on $\D(A)\alten\Hi_2$: this is easily seen by evaluating them on $\psi = \sum_{n=1}^N f_n \otimes e_n$, where the $e_n$ form an orthonormal set (see \cite[Lemma 2.1]{Arai}). Thus, $U$ extends by boundedness to an isometry on $\D(A\otimes I)$. Its range is closed and contains the dense subspace $\D(A)\alten\Hi_2$, making it surjective, and hence unitary. The statement for quadratic forms follows from $\Q(A) = \D(|A|^{1/2})$ and $\D(|A \otimes I|^{1/2}) = \D(|A|^{1/2} \otimes I)$ (see Proposition~\ref{app:op:standard_properties}(iii)).
\end{proof}

%==========================================================================
%           About positivity preserving operators
%==========================================================================
\subsection{Positivity preserving operators}\label{app:operator:positivity}

We refer to \cite[Section \cRM{13}.12]{RS} or \cite[Section 5]{thomas_thesis} for background on positivity preserving operators. We state a preliminary result and two useful properties of our model:

\begin{proposition}
Let $(X,\mu)$ be a measure space with $\sigma$-finite measure $\mu$. Let $A$ be a lower semi-bounded self-adjoint operator on $L^2(X,d\mu)$ and $B$ a multiplication operator on $L^2(X,d\mu)$. Denote $B_- = \ind_{]-\infty,0]}(B)B$, $B_+ = \ind_{]0, +\infty[}(B)B$. Assume:
\begin{enumerate}[label=(\roman*)]
\item The operator $A$ generates a positivity-improving semigroup.
\item The form domain $\Q(B_+)$ contains a core for $\q_A$ and $\Q(A)\cap \Q(B_+) \subset \Q(B_-)$.
\item The quadratic form $\q := \q_A+\q_{B_+} + \q_{B_-}$ is closed and lower semi-bounded.
\end{enumerate}
Then the lower semi-bounded operator corresponding to $\q$, denoted $C$, generates a positivity-improving semigroup.
\end{proposition}

\begin{proof} 
This is \cite[Theorem 3.4]{thomas_thesis}. More precisely, in the second step of the proof, with our convention $B_- \leq 0$, the bounded approximants have to be taken as $ B_n := \mathbf{1}_{[-n,+\infty[}(B_-)B_- = \mathbf{1}_{[-n,0]}(B)B$. Then $B_n$ is a bounded multiplication operator, $B_n \to B_-$ in the monotone sense required in the proof, and the rest of Dam's argument applies unchanged. 
\end{proof}

\begin{corollary}\label{app:vector_valued:sum_pos_improving}
    Let $(X,\mu)$ be a measure space with $\sigma$-finite measure $\mu$.
    Let $A$ be a lower semi-bounded self-adjoint operator on $L^2(X,d\mu)$
    generating a positivity-improving semigroup and $B$ be a self-adjoint
    multiplication operator. Assume that $\q_B$ is $\q_A$-bounded with relative bound strictly smaller than $1$. Then the self-adjoint operator associated with $\q := \q_A+\q_B$ generates a positivity-improving semigroup.
\end{corollary}

\begin{proof}
    By the KLMN theorem, $\q=\q_A+\q_B$ is closed and lower semi-bounded on
    $Q(A)$. With the notation of the previous proposition, the assumption that $\q_B$ is $\q_A$-bounded implies that $\Q(A)\subset \Q(B_+)\cap \Q(B_-)$. Hence $\Q(B_+)$ contains the form core $\Q(A)$ for $q_A$, and
    $\Q(A)\cap \Q(B_+)\subset \Q(B_-)$. Therefore the previous proposition applies and yields the result.
\end{proof}

\begin{proposition}\label{app:vector_valued:tensor_sum_pos_improving}
Let $(X, \mu)$ and $(Y, \nu)$ be $\sigma$-finite measure spaces. Let $A$ and $B$ be lower semi-bounded self-adjoint operators on $L^2(X, \mu)$ and $L^2(Y, \nu)$, respectively. Assume that $(e^{-tA})_{t>0}$ and $(e^{-tB})_{t>0}$ are positivity improving. Then, the semigroup generated by $H = \overline{A \otimes I + I \otimes B}$ is positivity improving on $L^2(X \times Y, \mu \otimes \nu)$.
\end{proposition}
\begin{proof}
    By \cite[Theorem 3.15]{Arai}, we can factorize the semigroup: $e^{-tH} = (e^{-tA} \otimes I)(I \otimes e^{-tB})$.
    Let $f \in L^2(X \times Y)$ be such that $f \ge 0$ a.e. and $f \neq 0$. Let $g = (I \otimes e^{-tB})f$. Since $f \neq 0$, there exists a subset $X_0 \subset X$ with $\mu(X_0)>0$ such that, for a.e. $x \in X_0$, the slice $f(x,\cdot)$ is a non-zero non-negative function in $L^2(Y)$. Since $e^{-tB}$ is positivity improving, $g(x,\cdot)=e^{-tB}f(x,\cdot)>0$ $\nu$-a.e. for a.e. $x \in X_0$. Hence, by Fubini's theorem, $g>0$ $\mu\otimes\nu$-a.e. on $X_0\times Y$. In particular, $g\ge 0$ a.e. and $g\neq 0$. \\
    Now let $h=(e^{-tA}\otimes I)g=e^{-tH}f$. Since $g>0$ $\mu\otimes\nu$-a.e. on $X_0\times Y$, the slice $g(\cdot,y)$ is a non-zero non-negative function in $L^2(X)$ for $\nu$-a.e. $y\in Y$. Since $e^{-tA}$ is positivity improving, $h(\cdot,y)=e^{-tA}g(\cdot,y)>0$ $\mu$-a.e. for $\nu$-a.e. $y\in Y$. Thus $e^{-tH}f>0$ $\mu\otimes\nu$-a.e.
\end{proof}

%%%%%%%%%%%%%%%%%%%%%%%%%%

%À FAIRE DANS CETTE SECTION

%%%%%%%%%%%%%%%%%%%%%
 
% --- Acknowledgments ---
\section*{Acknowledgments}
The authors wish to thank Jean-Marie Barbaroux and Philippe Briet, for their guidance and support. They are also grateful to Jérémy Faupin for a fruitful discussion on the binding condition.

%=======================================================================
% --- INDEX AND BIBLIOGRAPHY ---
%=======================================================================

\newpage
\newpage
\enlargethispage{3cm}
\section*{Index of notations}
\label{sec:notations}

\footnotesize
\setlength{\parindent}{0pt}
\newcommand{\notationpage}[1]{\hyperlink{page.#1}{#1}}
\newcommand{\notationgroup}[1]{\par\vspace{0.12cm}\noindent\textbf{#1}\par\vspace{0.04cm}}
\newenvironment{notationtable}[1][1.0]
{\renewcommand{\arraystretch}{#1}\begin{tabular}{@{}>{\raggedright\arraybackslash}p{0.27\linewidth}@{\hspace{0.025\linewidth}}>{\raggedright\arraybackslash}p{0.585\linewidth}@{\hspace{0.015\linewidth}}>{\raggedleft\arraybackslash}p{0.09\linewidth}@{}}}
{\end{tabular}}

The last column gives the page of first introduction; page numbers follow the order of the listed symbols.

\begin{minipage}[t]{0.505\textwidth}
\vspace{0pt}
\setlength{\extrarowheight}{3.4pt}

\noindent\textbf{General Conventions}
\par\vspace{0.02cm}
\begin{notationtable}
    $x=(\xu,\xd)$ & Particle position. & \notationpage{3} \\
    $k=(k_1,k_2)$ & Field momentum. & \notationpage{6} \\
    $L^2(\Rx^2)$, $L^2(\Rk^2)$, $L^2(\Ry^2)$ & Particle configuration, field momentum and field configuration spaces. & \notationpage{6} \\
    $f(X)$, $f(S,U)$, $f(K)$ & Associated multiplication operators in Cartesian, curvilinear and momentum coordinates. & \notationpage{4}, \notationpage{6} \\
    $f(\cdot)$ & Function with unspecified argument. & \notationpage{6} \\
    $I$ & Identity operator on the relevant Hilbert space. & \notationpage{5} \\
    $\D(A)$ & Domain of the operator $A$. & \notationpage{8} \\
    $\overline A$ & Closure of the closable operator $A$. & \notationpage{6} \\
    $\Ran(A)$ & Range of the operator $A$. & \notationpage{29} \\
    $\rho(A)$ & Resolvent set of the operator $A$. & \notationpage{29} \\
    $\ind_{\Delta}(A)$ & Spectral projection of $A$ associated with the Borel set $\Delta$. & \notationpage{31} \\
    $\D(\q)$; $\Q(A)$, $\q_A$ & Form domain; for self-adjoint $A$, its form domain and associated form. & \notationpage{9} \\
    $\otimes$, $\alten$, $\otimes_s$ & Hilbert, algebraic and symmetric tensor products. & \notationpage{5}, \notationpage{9}, \notationpage{6} \\
    $\int^\oplus$, $\oplus$ & Direct integral and direct sum. & \notationpage{7}, \notationpage{6} \\
\end{notationtable}

\notationgroup{Fibers and Extended Spaces}
\begin{notationtable}[1.0]
    $p_1$, $P_1$, $P_1^{\tot}$ & Particle, field and total longitudinal momenta. & \notationpage{18} \\
    $\eta$, $U_\fib$ & Fiber momentum and fibering operator. & \notationpage{17}, \notationpage{19} \\
    $\Hi^\ext$, $\Ha_S^\ext(\eta)$ & Extended Hilbert space and extended fiber Hamiltonian. & \notationpage{28} \\
    $j_R$, $\Gc_R$ & Fock-space localization tools. & \notationpage{27} \\
    $\dG^\ext(b)$ & Extended differential second quantization. & \notationpage{28} \\
    $P_1^\ext$, $\Nb^\ext$ & Extended longitudinal-momentum and number operators. & \notationpage{28}, \notationpage{29} \\
\end{notationtable}

\notationgroup{Spectral Quantities and Binding}
\begin{notationtable}[1.0]
    $\sigma$, $\sess$, $\sdisc$ & Spectrum, essential spectrum and discrete spectrum. & \notationpage{7}, \notationpage{7}, \notationpage{10} \\
    $\D_R$ & Functions supported outside $B(0,R)$. & \notationpage{11} \\
    $\Sigma_R$, $\Sigma$ & Localized Persson quantity and ionization threshold. & \notationpage{11}, \notationpage{16} \\
    $\Sigma_S$, $\Sigma_C$ & Straight and curved ionization thresholds. & \notationpage{21} \\
    $E(\eta)$, $\eta_0$ & Fiber energy and one of its minimizers. & \notationpage{20}, \notationpage{26} \\
    $\Sigma(\eta)$ & Fiber ionization threshold. & \notationpage{26} \\
    $E_0^{(1)}(\eta)$ & Lowest one-free-boson fiber threshold. & \notationpage{27} \\
\end{notationtable}

\end{minipage}%
\hfill
\hspace*{-3.0826pt}%
\begin{minipage}[t]{0.465\textwidth}
\vspace{0pt}
\setlength{\extrarowheight}{3.4pt}

\noindent\textbf{Waveguide Geometry}
\par\vspace{0.02cm}
\begin{notationtable}
    $\gamma$, $\kappa$ & Reference curve and its signed curvature. & \notationpage{3} \\
    $(s,u)$, $\varphi$ & Curvilinear coordinates and associated diffeomorphism. & \notationpage{3} \\
    $\mathcal O_\gamma$ & Maximal domain of the coordinate diffeomorphism. & \notationpage{4} \\
    $a$, $\Omega^a$ & Half-width and associated strip. & \notationpage{3} \\
    $\Cut(\gamma)$, $c_\pm(s)$ & Cut locus and cut radii. & \notationpage{3}, \notationpage{4} \\
    $B(0,R)$; $\Omega_{R,\pm}$ & Ball centered at $0$ with radius $R$; strips along the two straight ends. & \notationpage{11}, \notationpage{13} \\
    $\widetilde j_R$, $\widetilde j_{R,\pm}$ & Configuration-space localization functions. & \notationpage{13} \\
\end{notationtable}

\notationgroup{Waveguide Potentials}
\begin{notationtable}[1.0]
    $V$ & Transverse potential well. & \notationpage{4} \\
    $V_S$, $V_C$ & Straight and curved waveguide potentials. & \notationpage{4} \\
    $V_\bullet$ & Either $V_S$ or $V_C$. & \notationpage{4} \\
    $V_{\eff,\Psi}$ & Effective potential in the binding condition. & \notationpage{26} \\
\end{notationtable}

\notationgroup{Fock Space and Field}
\begin{notationtable}
    $\Fs(\Hi)$, $\Omega$ & Symmetric Fock space and vacuum vector. & \notationpage{6}, \notationpage{10} \\
    $a(f)$, $\ac(f)$, $\phi(f)$ & Annihilation, creation and Segal field operators. & \notationpage{6} \\
    $\Gamma(q)$, $\dG(b)$ & Second quantization and its differential version. & \notationpage{8}, \notationpage{6} \\
    $\Nb=\dG(I)$ & Boson number operator. & \notationpage{28} \\
    $\Fs^\infty$ & $\Fsf(\ciz(\Rk^2))$, the finite-particle momentum-space core. & \notationpage{28} \\
    $\omega(K)$, $m$ & Massive dispersion relation and boson mass. & \notationpage{6} \\
    $v_x$, $v$ & Interaction kernel and form factor. & \notationpage{6} \\
\end{notationtable}

\notationgroup{Hamiltonians}
\begin{notationtable}[1.0]
    $\ha_{\el,S}$, $\ha_{\el,C}$ & Straight and curved electronic Hamiltonians. & \notationpage{5} \\
    $\Ha_{\el}$, $\Ha_f$ & Electronic and free-field Hamiltonians. & \notationpage{7} \\
    $\Ha_I$, $\Ha_\free$; $g$ & Interaction and free Hamiltonians; coupling constant. & \notationpage{7} \\
    $\Ha$; $\Ha_C$, $\Ha_S$ & Full Hamiltonian and its curved and straight versions. & \notationpage{7}, \notationpage{3}, \notationpage{16} \\
    $\ha_{\el,0}$, $\Ha_{\el,0}$, $\Ha_{I,0}$ & Transverse particle Hamiltonian, its lift to the fiber space, and fiber interaction term. & \notationpage{19} \\
    $\dG_\eta$, $\Ha_S(\eta)$ & Free-field plus longitudinal fiber kinetic term, and fiber Hamiltonian. & \notationpage{19} \\
\end{notationtable}

\end{minipage}%
\hspace*{3.0826pt}
\normalsize
\par
\newpage
\bibliographystyle{amsalpha-doi} 
\providecommand{\bysame}{\leavevmode\hbox to3em{\hrulefill}\thinspace}
\providecommand{\MR}{\relax\ifhmode\unskip\space\fi MR }
\providecommand{\MRhref}[2]{%
  \href{http://www.ams.org/mathscinet-getitem?mr=#1}{#2}
}
\providecommand{\href}[2]{#2}

\end{document}